\documentclass[reqno,12pt]{amsart}

\usepackage{amssymb,amsmath,epsfig,hyperref}
\usepackage{hypernat}
\usepackage{graphicx}
\usepackage{subfigure}
\usepackage{amsthm}
\usepackage{latexsym}
\usepackage{epstopdf}
\usepackage{ifpdf}
\usepackage{verbatim}
\ifpdf
\fi

\makeatletter
 
 \newcommand{\Rmnum}[1]{\expandafter\@slowromancap\romannumeral #1@}
 \makeatother

\newtheorem{theorem}{Theorem}[section]

\newtheorem{proposition}[theorem]{Proposition}
\newtheorem{lemma}[theorem]{Lemma}

\newtheorem{remark}[theorem]{Remark}

\numberwithin{equation}{section}

\usepackage{geometry}
\begin{document}
\title[RHP for the dKE and fKE Eqns]{Inverse scattering transform for the Kundu-Eckhaus Equation with nonzero boundary condition}

\author[N.Guo]{Ning Guo}

\address{College of Science\\
University of Shanghai for Science and Technology\\
Shanghai 200093\\
People's  Republic of China}

\author[J.Xu]{Jian Xu*}

\address{College of Science\\
University of Shanghai for Science and Technology\\
Shanghai 200093\\
People's  Republic of China}
\email{corresponding author: jianxu@usst.edu.cn
}


\date{\today}

\begin{abstract}
In this paper, we consider the initial value problem for both of the defocusing and focusing Kundu-Eckhaus (KE) equation with non-zero boundary conditions (NZBCs) at infinity by inverse scattering transform method. The solutions of the KE equation with NZBCs can be reconstructed in terms of the solution of an associated $2 \times 2$ matrix Riemann-Hilbert problem (RHP). In our formulation, both the direct and the inverse problems are posed in terms of a suitable uniformization variable which allows us to develop the IST on the standard complex plane instead of a two-sheeted Riemann surface or the cut plane with discontinuities along the cuts. Furthermore, on the one hand, we obtain the N-soliton solutions with simple pole of the defocusing and focusing KE equation with the NZBCs, especially, the explicit one-soliton solutions are given in details. And we prove that the scattering data $a(\zeta)$ of the defocusing KE equation can only have simple zeros. On the other hand, we also obtain the soliton solutions with double pole of the focusing KE equation with NZBCs. And we show that the double pole solutions can be viewed as some proper limit of the two simple pole soliton solutions. Some dynamical behaviors and typical collisions of the soliton solutions of both of the defocusing and focusing KE equation are shown graphically.
\end{abstract}

\maketitle

\section{Introduction}

The nonlinear Schr\"{o}dinger (NLS) equation :
\begin{equation}\label{Nls}
iu_t+u_{xx}+2 \kappa  \left| u \right| ^2u=0,    \kappa = \pm1,
\end{equation}
(where $\kappa=-1$ and $1$ corresponding to the defocusing and focusing NLS equation, respectively). It is a well-known basic model involving the fields of nonlinear optics, plasma, Bose-Einstein condensation, weak nonlinear dispersion wave packet evolution and so on \cite{D.S.W}-\cite{B.P}. Nevertheless, in optic fiber communications systems, high-order nonlinear terms such as third-order dispersion, self-steepness, and self-frequency shift must be considered in order to accurately describe the ultrashort ( femtosecond, even attosecond ) optic pulses, which resulting from the increased the intensity of the incident light field \cite{Qi.D,X.W}.\par

The Kundu-Eckhaus (KE) Equation:
\begin{equation}\label{KE}
iu_t+u_{xx}+2 \kappa  \left| u \right| ^2u+4 T^2\left|u\right|^4u-4i T  \kappa ( \left| u \right| ^2)_xu=0,     \kappa = \pm1,
\end{equation}
where $T$ is a constant, $u(x,t)$ is a complex function. Here, $\kappa$=-1 and $\kappa$=1 are called defocusing KE and focusing KE equation, respectively. This is proposed by Kundu when considering the gauge relationship between the Landau-Lifshitz equation and the high-order Schr\"{o}dinger equation \cite{A.Ku}. It contains quintic nonlinearity and the last term is the Raman effect in nonlinear optics, which represents the non-Kerr nonlinear effect and the self-frequency shift effect, respectively. The KE equation fully describes the propagation of ultrashort optical pulses in non-linear optics and detects the stability of Stokes waves in weakly non-linear dispersive matter waves \cite{Y.K,R.S.J}. It was a complete integrable system with lax pairs \cite{X.G.G}-\cite{L.C.Z2}, Painleve property \cite{C.PA} and Hamiltonian structure \cite{X.G.G.2}.\par

The soliton solutions, rogue wave solutions and the flow wave of the KE equation can be obtained by the Darboux transformations \cite{Qi.D,W.X.M} and the bilinear method \cite{K.S.S.N}. Via a direct method \cite{Z.F1}, higher order extension \cite{A.K1} and the $\tan (\phi (\zeta ))$-expansion method \cite{J.M}, the abundant soliton solutions of the KE equation, etc, were obtained. Recently, the bright soliton solutions and the long-time asymptotics with zero boundary conditions of the KE equation were obtained in \cite{D.S.W,Q.Z.Z}. In fact, the properties of the defocusing and focusing KE equation are quite different due to the difference of their Lax pairs.

The main work of this paper is that the defocusing and focusing KE equation with non-zero boundary conditions are discussed in detail by the Riemann-Hilbert approach. The N-soliton solutions with simple poles of the defocusing and focusing KE equation with the NZBCs are obtained, respectively. The influence factors of the asymptotic phase difference are given by the condition $\theta$, and the explicit one-soliton solutions are given in details. And we prove that the scattering data $a(\zeta)$ of the defocusing KE equation can only have simple zeros. Further, we also obtain the soliton solutions with double poles for the focusing KE equation with NZBCs. And we show that the double-pole solutions can be viewed as some proper limit of the two simple poles soliton solutions. It is worth noting that there are many studies in \cite{ZS}-\cite{APT} for the double-pole and multi-pole solutions with ZBC and \cite{M.Gb} for double-pole solutions with NZBCs. However, it can be seen that the process of solving the double-pole soliton solution is quite complicated. The relationship between the simple-pole and double-pole soliton solutions will help us to avoid the complicated process of solving the residue conditions with double poles, and directly obtain the soliton solutions with double poles from the simple poles soliton solutions through a proper limit, which is the main discovery of this paper.

The rest of this paper is arranged as follows. In section 2, the direct scattering and inverse problems of defocusing KE equation with NZBCs are studied, and the N-soliton solution is obtained. In section 3, we investigate the N-soliton solution of the focusing KE equation with NZBCs with simple poles. In addition, in section 4, we get the double poles soliton solutions of the focusing KE equation with NZBCs. In section 5, we find a relationship between the simple-pole and double-pole soliton solutions. Conclusion and remark are given in section 6.

\section{The Defocusing KE equation with NZBCs case}
In this paper, we consider the boundary value problem (BVP) for the defocusing Kundu-Eckhaus (${\operatorname{KE} ^ + }$) Equations ($\kappa$=-1):
\begin{equation}\label{1.a}
iu_t+u_{xx}-2 \left| u \right| ^2u+4 T^2\left|u\right|^4u+4i T ( \left| u \right| ^2)_xu=0,
\end{equation}
\begin{equation}\label{1.b}
\mathop {\lim }\limits_{x \to  \pm \infty } u(x,t) = {u_ \pm } = {u_0}{e^{i(\alpha x + 2( - q_0^2 + 2{T^2}q_0^4 - \frac{{{\alpha ^2}}}{2})t + {\theta _ \pm })}},
\end{equation}
with $\alpha$, $\theta _ \pm$ are real parameters. $|{u_ \pm }| = u_0 > 0$,  and suppose $u(x,t) - {u_ \pm } \in {L^{1,2}}({\mathbb{R}^ + })$.
The ${\operatorname{KE} ^ + }$ equation admits the following Lax pairs:
\begin{equation}\label{1.c}
{\Phi _x} + iz{\sigma _3}\Phi = (U - iT{U^2}{\sigma _3})\Phi ,
\end{equation}
\begin{equation}\label{1.d}
{\Phi _t} + 2i{z^2}{\sigma _3}\Phi  = (V + 4i{T^2}{U^4}{\sigma _3} - T(U{U_x} - {U_x}U))\Phi ,
\end{equation}
where $z$ is spectral parameter and
\begin{equation}
U = \left( {\begin{array}{*{20}{c}}
0&u\\
{ \bar u}&0
\end{array}} \right) ,
\end{equation}
\begin{equation}
V = 2zU - 2T{U^3} - i({U^2} + {U_x}){\sigma _3} ,
\end{equation}
Denote the Pauli matrices and identity matrix as follows
\begin{equation}
{\sigma _1} = \left( {\begin{array}{*{20}{c}}
0&1\\
1&0
\end{array}} \right) , \quad
{\sigma _2} = \left( {\begin{array}{*{20}{c}}
0&-i\\
i&0
\end{array}} \right) , \quad
{\sigma _3} = \left( {\begin{array}{*{20}{c}}
1&0\\
0&-1
\end{array}} \right) , \quad
\mathbb{E}=\left(
\begin{array}{lll}
  1 & 0\\
  0 & 1\\
\end{array}
\right) .
\end{equation} \par
\subsection{Direct scattering problem with NZBC}

In order to make the Lax pairs of the time part and the space part compatible, consider the boundary condition (\ref{1.b}), introducing transformation as
\begin{equation}\label{1}
q(x,t) = u(x,t) \cdot {e^{ - i(\alpha x + 2( - q_0^2 + 2T^2q_0^4 - \frac{{{\alpha ^2}}}{2})t)}}.
\end{equation}
The ${\operatorname{KE} ^ + }$ equation is rewritten as
\begin{equation}\label{1.1a}
i{q_t} + {q_{xx}} - 2(|q{|^2} - q_0^2 + 2{T^2}q_0^4)q + 2i\alpha {q_x} + 4{T^2}|q{|^4}q + 4iT{(|q{|^2})_x}q = 0,
\end{equation}
with
\begin{equation}\label{1.1b}
\mathop {\lim }\limits_{x \to  \pm \infty } q(x,t) = {q_ \pm } = {q_0}{e^{i{\theta _ \pm }}},
\end{equation}
the corresponding Lax pairs as follows
\begin{equation}\label{1.1c}
{\Psi _x} = X\Psi ,
\end{equation}
\begin{equation}\label{1.1d}
{\Psi _t} = V\Psi ,
\end{equation}
where
\begin{equation}\label{1.1e}
X = Q - i(z + T|q{|^2}+ \frac{\alpha }{2} ){\sigma _3},
\end{equation}
\begin{equation}\label{1.1f}
\begin{split}
V =  (2z - 2T|q{|^2} -\alpha )Q - i(2{z^2}+ |q{|^2} -q_0^2 - 2\alpha T|q|^2 - 4{T^2}|q{|^4} + 2{T^2}q_0^4 - \frac{{{\alpha ^2}}}{2}){\sigma _3},
\end{split}
\end{equation}
$$Q = \left( {\begin{array}{*{20}{c}}
  0&q \\
  {\bar q}&0
\end{array}} \right).$$
One can verify that the Lax pairs satisfied the compatibility condition
\begin{equation}\label{1.1g}
{X_t} - {T_x} + [X,T] = 0,
\end{equation}
where $[X,T]=XT-TX$.
As $x$ tend infinity, the asymptotic behavior of Lax equations (\ref{1.1c})(\ref{1.1d}) at the boundary conditions (\ref{1.1b}) are:
\begin{equation}\label{1.1h}
{\Psi _x} = {X_ \pm }\Psi ,
\end{equation}
\begin{equation}\label{1.1i}
{\Psi _t} = {V_ \pm }\Psi ,
\end{equation}
where
\begin{equation}\label{1.1j}
{X_ \pm } = {Q_ \pm } -i ( z +Tq_0^2 +\frac{\alpha }{2}){\sigma _3}
\end{equation}
\begin{equation}\label{1.1k}
{V_ \pm } = (2z - 2Tq_0^2 - \alpha ){Q_ \pm } -i( 2{z^2} -2\alpha Tq_0^2 -2{T^2}q_0^4 - \frac{{{\alpha ^2}}}{2}){\sigma _3}
\end{equation}
$$Q = \left( {\begin{array}{*{20}{c}}
  0&{q_ \pm} \\
  { {\bar q_\pm}}&0
\end{array}} \right).$$ \par
Then, we have ${T_ \pm } = (2z - 2Tq_0^2 - \alpha ){X_ \pm }$, which allows Lax pairs to correspond to the same eigenvector and Jost solution.

\subsection{Riemann surface and uniformization variable }

\begin{figure}[htpb]
\centering
\subfigure[The complex $\kappa-plane$]{
\begin{minipage}{6cm}
\centering
\includegraphics[width=6.1 cm]{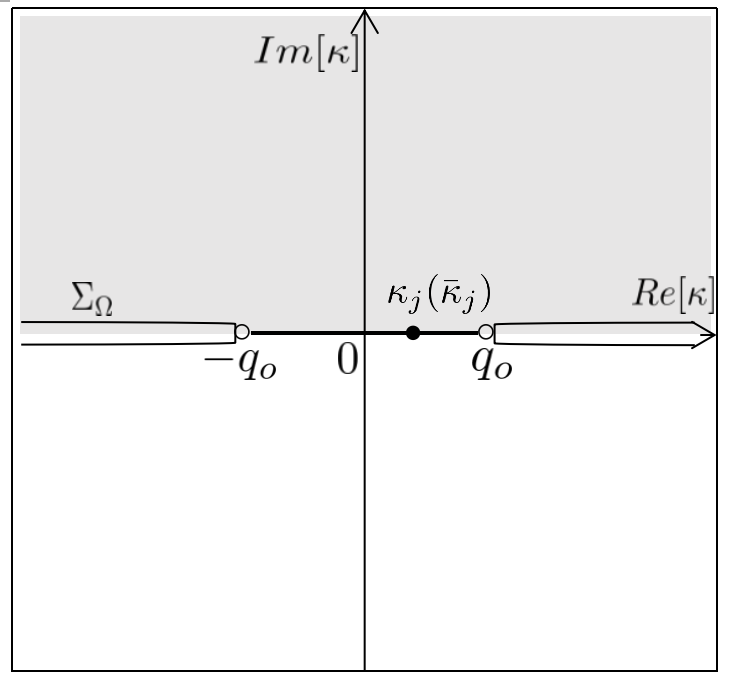}
\end{minipage}
} \qquad
\subfigure[The complex $\zeta-plane$]{
\begin{minipage}{6cm}
\centering
\includegraphics[width=6 cm]{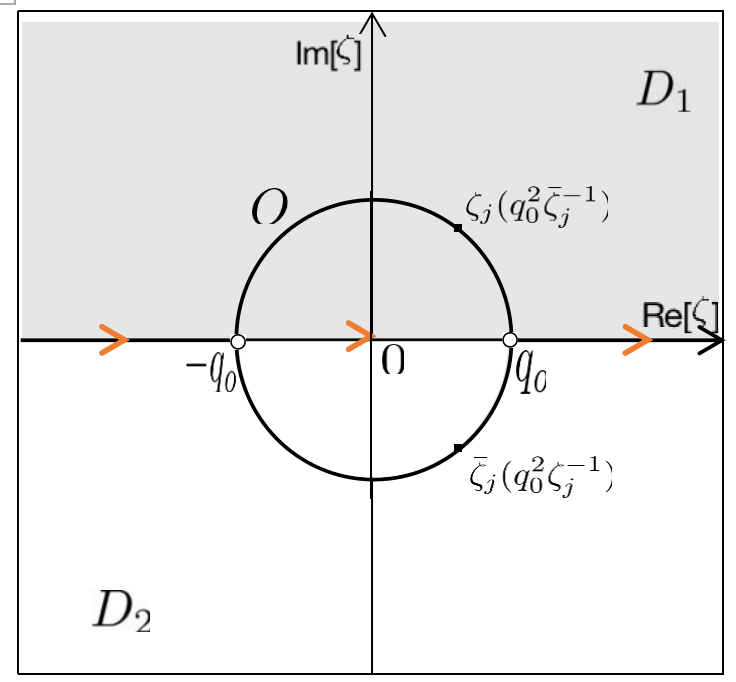}
\end{minipage}
}
\caption{(a) Showing the complex $\kappa-plane$ region where ${\mathop{\rm Im}\nolimits} [\kappa] > 0$(gray) and ${\mathop{\rm Im}\nolimits} [\kappa] < 0$(white).Also shown the discrete spectrums. (b) Showing the regions $D_1$(gray) and $D_2$(white).Also shown the symmetries of the discrete spectrum of the scattering problem and the oriented contours for the Riemann-Hilbert problem(orange).}
\label{fig.1aa}
\end{figure}

In the case of continuum and the spectral parameter $z$ is a real value. With the boundary conditions, we can get the eigenvalues of Lax pair (\ref{1.1c}) corresponding to $ \pm i \lambda $, where ${\lambda ^2} = {(z + Tq_0^2 + \frac {\alpha} {2})^2} - q_0^2 $. Jost solution could be obtained as follows:
\begin{equation}\label{1.2a}
\begin{split}
{\psi _ \pm }(x,t,z) = \left( {\begin{array}{*{20}{c}}
  1&{ - i\frac{{{q_ \pm }}}{{\lambda  + z + Tq_0^2 + \frac{\alpha }{2}}}} \\
  {i\frac{{{{\bar q}_ \pm }}}{{\lambda  + z + Tq_0^2 + \frac{\alpha }{2}}}}&1
\end{array}} \right)&{e^{ - i\lambda (x + 2(z - Tq_0^2 - \frac{\alpha }{2})t){\sigma _3}}} ,\\& \quad z \in \mathbb{R}/\{ \pm {q_0}-Tq_0^2-\frac{\alpha}{2}\}, x \to \infty .
\end{split}
\end{equation}

Letting $\kappa  = z + Tq_0^2 + \frac{\alpha }{2}$, since the $\lambda$ is $\kappa$ two-valued function and there are two branch points at $z =  \pm q_0^2 - Tq_0^2 - \frac{\alpha }{2}$. we need to introduce Riemann surface $\mathbb{K}$ consisting of a sheet $\mathbb{K^+}$ and a sheer $\mathbb{K^-}$ to handle, where the sheet $\mathbb{K^+}$ and the sheer $\mathbb{K^-}$ are glued together along the complex plane cut ${\Sigma _\Omega } = ( - \infty , - {q_0}] \cup [{q_0}, + \infty )$, so that $\lambda(k)$ continuous through the cut. Without loss of generality, we might as well define $\operatorname{Im} (\lambda (k)) > 0$ on $\mathbb{K^+}$ and $\operatorname{Im} (\lambda (k)) < 0$ on $\mathbb{K^-}$. To overcome the multi-value of the square root and know that the genus of Riemann surface is $0$, it is suitable to use the uniformization variable in the process of inverse scattering transform $\zeta$ as follows\cite{L.D.F}\cite{N.N.H}:
\begin{equation}\label{1.2b}
\kappa  = \frac{1}{2}(\zeta  + q_0^2{\zeta ^{ - 1}}), \qquad  \lambda = \frac{1}{2}(\zeta  - q_0^2{\zeta ^{ - 1}}).
\end{equation} \par
It is easy to see the correspondence between the Riemann surface and the complex $\zeta$-plane as follows
\begin{itemize}
\item  The Riemann surface $\mathbb{K^+}$ and $\mathbb{K^-}$ are mapped onto the upper and lower half planes of the $\zeta$-plane.
\item  The cut ${\Sigma _\Omega } = ( - \infty , - {q_0}] \cup [{q_0}, + \infty )$ is mapped onto the real $\zeta$-axis of the $\zeta$-plane, and then $\lambda$ is a real value.  The segments $-{q_0} \leqslant \kappa \leqslant {q_0}$ on $\mathbb{K^+}$ and $\mathbb{K^-}$ are respectively mapped onto the upper and lower semicircles whose radius is $q_0$ and center at the origin of the $\zeta$-plane, which is recorded as $O=\{\zeta \left | {| \zeta | = {q_0}}, \zeta \in \mathbb{C} \right. \}$ , and $\lambda$ is a pure imaginary value.
\item  We have the following correspondence. While $(\kappa , \lambda) \mapsto (\bar \kappa , \bar \lambda )$ corresponds $\zeta \mapsto \bar \zeta$ , and $(\kappa , \lambda) \mapsto (\bar \kappa , -\bar \lambda )$ to $\zeta \mapsto \frac {{q_0^2}} {{\bar \zeta}}$ .
\end{itemize}

\begin{remark}
Up to now, we have discussed the single-valued property of $\kappa$. It can be seen from the definition of $\kappa$ that the single-valued property of $\kappa$ is equivalent to the single-valued property of the spectral parameter $z$.
\end{remark}

From the uniformization transform, the Jost solution(\ref{1.2a}) could be obtained as follows:
\begin{equation}\label{1.2c}
{\psi _ \pm }(x,t,\zeta ) = \left( {\begin{array}{*{20}{c}}
1&{ - i{q_ \pm }{\zeta ^{ - 1}}}\\
{  i{{\bar q}_ \pm }{\zeta ^{ - 1}}}&1
\end{array}} \right){e^{ - i\lambda (x + 2(z - Tq_0^2- \frac{\alpha}{2})t){\mathop{\sigma_3} }}}, \qquad  \zeta  \in \mathbb{R}, x \to \infty  .
\end{equation}

For simplicity, Letting $\Theta (x,t,\zeta ) = \lambda (x + 2(z - Tq_0^2-\frac{\alpha}{2})t)$ and denoting the region ${D_1} = \{ \zeta  \in \mathbb{C}| {\mathop{\rm Im}\nolimits} \zeta  > 0\} $, ${D_2} = \{ \zeta  \in \mathbb{C}| {\mathop{\rm Im}\nolimits} \zeta  < 0\} $.

\subsection{The Jost solutions}

From Jost solution (\ref{1.2c}) we defined the modified Jost solution $ \mu \mathop{{}}\nolimits_{{ \pm }} \left( x,t, \zeta  \right)  $, as follows:

\begin{equation}\label{1.2d}
{\mu _\pm}(x,t,\zeta ) = { \psi _ \pm }{e^{i\Theta (x,t,\zeta ){\mathop{\sigma_3} }}}, \qquad \zeta  \in \mathbb{R}  ,
\end{equation}
such that
\begin{equation}\label{1.2e}
\mathop {\lim }\limits_{x \to  \pm \infty } {\mu _ \pm }(x,t,\zeta ) = {\mathbb {E}_ \pm }(\zeta )+ o(1) ,  \qquad  \zeta  \in \mathbb{R}  ,
\end{equation}

From Lax pairs (\ref{1.1c})(\ref{1.1d}), then $\mu_ \pm(x,t,\zeta )$ solves

\begin{equation}\label{1.2f}
\left\{ \begin{gathered}
  {[\mathbb{E}_ \pm ^{ - 1}{\mu _ \pm }(x,t,\zeta )]_x} + i\lambda [{\sigma _3},\mathbb{E}_ \pm ^{ - 1}{\mu _ \pm }(x,t,\zeta )] = \mathbb{E}_ \pm ^{ - 1}M{\mu _ \pm } \hfill \\
  {[\mathbb{E}_ \pm ^{ - 1}{\mu _ \pm }(x,t,\zeta )]_t} + 2i\lambda (z - Tq_0^2 - \frac{{{\alpha ^2}}}{2})[{\sigma _3},\mathbb{E}_ \pm ^{ - 1}{\mu _ \pm }(x,t,\zeta )] = \mathbb{E}_ \pm ^{ - 1}N{\mu _ \pm } \hfill \\
\end{gathered}  \right.
\end{equation}
where $$M =\Delta Q-iT\Delta q\sigma_3 ,$$
$$N = ( 2z- \alpha) \Delta Q+2T ( q_0^2Q_{ \pm }- |q|^2 Q) -i( 1-2\alpha T \Delta q-4T^2\Delta q ) \sigma_3,$$
$$\Delta Q=Q-Q_{\pm}, \qquad \Delta q = |q|^2 -q_0^2 ,$$
$${\mathbb{E}_ \pm } = \left( {\begin{array}{*{20}{c}}
  1&{ - i{q_ \pm }{\zeta ^{ - 1}}} \\
  {  i{{\bar q}_ \pm }{\zeta ^{ - 1}}}&1
\end{array}} \right).$$
The relationship (\ref{1.2f}) can be written in the full derivative form
\begin{equation}\label{1.2g}
d({e^{i\Theta (x,t,\zeta ){{\hat \sigma }_3}}}\mathbb{E}_ \pm ^{ - 1}{\mu _ \pm }(x,t,\zeta )) = {e^{i\Theta (x,t,\zeta ){{\hat \sigma }_3}}}(Mdx + Ndt).
\end{equation}
Then define two eigenfunctions $\mu_ +(x,t,\zeta )$ and $\mu_ +(x,t,\zeta )$ as
\begin{equation}\label{1.2h}
{\mu _ + }(x,t,\zeta ) = {\mathbb{E}_ + } - \int_x^{ + \infty } {{\mathbb{E}_ + }{e^{ - i\lambda (x - x'){{\hat \sigma }_3}}}(E_ + ^{ - 1}M{\mu _ + }(x',\zeta ))dx'},
\end{equation}
\begin{equation}\label{1.2i}
{\mu _ - }(x,t,\zeta ) = {\mathbb{E}_ - } + \int_{ - \infty }^x {{\mathbb{E}_ - }{e^{ - i\lambda (x - x'){{\hat \sigma }_3}}}(E_ - ^{ - 1}M{\mu _ - }(x',\zeta ))dx'},
\end{equation}
where ${e^{{{\hat \sigma }_3}}}(\Lambda ) = {e^{{\sigma _3}}}(\Lambda ){e^{ - {\sigma _3}}}$.

\begin{proposition}\label{1.2j}
(Analytic property). then ${\mu _ \pm }(x,t,\zeta )$ satisfy the following bounded and analytic properties,
\begin{equation}\label{1.2k}
{\mu _ + } = [{\mu _ + }^{(2)},{\mu _ + }^{(1)}], \qquad {\mu _ - } = [{\mu _ - }^{(1)},{\mu _ - }^{(2)}] ,
\end{equation}
where the superscripts $^{(1),(2)}$ indicate that each column of ${\mu _ \pm }(x,t,\zeta )$ are bounded and analytic in the $D_1$ and $D_2$ regions, respectively.
\end{proposition}

To see this, We represent ${\mu _ \pm }(x,t,\zeta )$ as columns ${\mu _ \pm } = [{\mu _{ \pm 1}},{\mu _{ \pm 2}}]$. For any fixed $t$, from Jost Integral equation (\ref{1.2h}), we know that ${\mu _{ + 1}}$ containing exponential factor ${e^{2i\lambda (x - x')}}$. When ${\mathop{\rm Im}\nolimits} [\lambda ] < 0$, the exponent of the integrand is not only bounded but decayed thanks to $x-x'<0$, which corresponding to the region of $D_1$ in the $\zeta$-plane. Then, the first column of ${\mu _ \pm }(x,t,\zeta )$ is bounded and analytic in $D_1$. By the same method, the bounded and analytic properties of other columns of ${\mu _ \pm }(x,t,\zeta )$ can be obtained, respectively.

~~
\subsubsection{Scattering matrix}
The eigenfunctions ${\psi _ + }(x,t,\zeta )$ and ${\psi _ - }(x,t,\zeta )$ are not independent, satisfy the relation for some matrix $s(\zeta )$ independent of $x$ and $t$.
\begin{equation}\label{1.3a}
{\psi _ - }(x,t,\zeta ) = {\psi _ + }(x,t,\zeta )s(\zeta ) , \qquad  \zeta  \in \mathbb{R}  ,
\end{equation}
where
\begin{equation}\label{1.3e}
s(\zeta ) = \left( {\begin{array}{*{20}{c}}
{a(\zeta )}&{\tilde b(\zeta )}\\
{b(\zeta )}&{\tilde a(\zeta )}
\end{array}} \right) .
\end{equation}
\begin{proposition}\label{1.3b}
(The first symmetry property). The modified Jost functions ${\mu _ \pm }(x,t,\zeta )$ have the following symmetry conditions:
\begin{equation}\label{1.3c}
  {\mathop{\sigma_1} }\overline{{\mu _ \pm }(x,t,\bar \zeta )}{\mathop{\sigma_1} } = {\mu _ \pm }(x,t,\zeta ) .
 \end{equation}
\end{proposition}
\begin{proof}
It is easy to see that the lax pairs satisfies the symmetry of ${\mathop{\sigma_1} }$ and the same as the asymptotic conditions, that is ${\mathop{\sigma_1} }\overline{(X - i\bar z{\mathop{\sigma_3} } )} {\mathop{\sigma_1} } = X - iz{\mathop{\sigma_3} }$, ${\mathop{\sigma_1} }\overline{({X_ \pm } - i\bar z{\mathop{\sigma_3} })}{\mathop{\sigma_1} } = {X_ \pm } - iz{\mathop{\sigma_3} }$, etc. We have ${\mathop{\sigma_1} } \overline {{ \psi _ \pm }(x,t,\bar \zeta )}{\mathop{\sigma_1} } = { \psi _ \pm }(x,t,\zeta )$. From the definition (\ref{1.2d}) of the modified Jost functions ${\mu _ \pm }(x,t,\zeta )$, a direct computation shows (\ref{1.3c}) holds.
 \end{proof}

From the above proof and (\ref{1.3a}), the symmetry of the scattering data could be obtained as follows:
\begin{equation}\label{1.3d}
{\mathop{\sigma_1} }\overline{s(\bar \zeta )}{\mathop{\sigma_1} } = s(\zeta ),  \qquad  \zeta  \in \mathop{\Sigma}\nolimits_\Omega ,
\end{equation}
which means  $\tilde a(\zeta ) = \overline{ a(\bar \zeta )}$,  $\tilde b(\zeta ) =  \overline{ b(\bar \zeta )}$ and $|a(\zeta ){|^2} =  1 + |b(\zeta ){|^2} $. \par
By the analytic property (\ref{1.2j}) and the equation (\ref{1.3a}), we know that $a(\zeta )$ is analytic in $D_1$, $\tilde a(\zeta )$ is analytic in $D_2$. But $b(\zeta)$ and $\tilde b(\zeta)$ are only defined on real $\zeta$-axis.

Introduce the following transformation
\begin{equation}\label{1.3l}
\zeta   \mapsto  q_0^2{\zeta ^{ - 1}},
\end{equation}
which leads to $ \kappa(\zeta)  \to \kappa (  q_0^2{\zeta ^{ - 1}})$ and $\lambda (\zeta ) \to -\lambda ( - q_0^2{\zeta ^{ - 1}})$. Since the Lax pair (\ref{1.1c}) (\ref{1.1d}) only contains $\kappa$, so that the spectral parameter $z$ remained unchanged under this transformation. For convenience, we letting $  q_0^2{\zeta ^{ - 1}} \mapsto \hat \zeta $, $  q_0^2{\bar \zeta ^{ - 1}} \mapsto \check \zeta $.
We can get the following symmetry:
\begin{proposition}\label{1.3m}
(The second symmetry property). The modified Jost functions ${\mu _ \pm }(x,t,\zeta )$ have the following symmetry conditions:
\begin{equation}\label{1.3n}
{\mu _ \pm }(x,t,\zeta ) =  - i{\zeta ^{ - 1}}{\mu _ \pm }(x,t,\hat \zeta ){\mathop{\sigma_3} }{Q_ \pm } .
 \end{equation}
\end{proposition}
\begin{proof}
Since the lax pairs (\ref{1.1c}) (\ref{1.1d}) remains unchanged under this transformation, from the eigenfunctions (\ref{1.2h}) (\ref{1.2i}), a direct computation shows the proposition (\ref{1.3m}) holds.
\end{proof}
Then, we also could obtained the scattering data satisfying the following conditions:
\begin{equation}\label{1.3o}
s(\zeta ) = {({\mathop{\sigma_3} }{Q_ + })^{ - 1}}s(\hat \zeta ){\mathop{\sigma_3} }{Q_ - }, \qquad  \zeta  \in \mathbb{R} .
\end{equation}

Define the reflection coefficient as $r(\zeta ) = \frac{{b(\zeta )}}{{a(\zeta )}}$, then its symmetry can also be obtained
 \begin{equation}\label{1.3s}
r(\zeta ) =  \overline{ \tilde r(\bar \zeta )} = -\frac{{{{\bar q}_ + }}}{{{ q_ + }}}\tilde r(\hat \zeta ) =  - \frac{{{ {\bar q}_ + }}}{{{ q_ + }}}\overline{r(\check \zeta )} ,  \qquad  \zeta  \in \mathbb{R}.
 \end{equation}
According to the Schwarz reflection principle, we can determine the analytical extension regions of the above scattering data on the complex $\zeta$-plane.
\subsubsection{Asymptotic behavior}
With the Wentzel-Kramers-Brillouin(WKB) expansion, the asymptotic formula of ${\mu _ \pm }(x,t,\zeta )$ as $\zeta  \to \infty $ and $\zeta  \to 0 $ can be derived by substituting the following expansion
\begin{equation}\label{1.4a}
 {\mu _ \pm }(x,t,\zeta ) = \sum\limits_{j = 0}^n {\frac{{\mu _ \pm ^{(j)}(x,t)}}{{{\zeta ^j}}} + O(\frac{1}{{{\zeta ^{n + 1}}}})} ,\qquad \zeta  \to \infty ,
\end{equation}
\begin{equation}\label{1.4b}
{\mu _ \pm }(x,t,\zeta ) = \sum\limits_{j =  -1 }^n {{\zeta ^j} \cdot \mu _ \pm ^{(j)}(x,t) + O({\zeta ^{n + 1}})}  ,\qquad \zeta  \to 0 ,
\end{equation}
into the Lax pairs (\ref{1.1c}), (\ref{1.1d}), and comparing the order of $\zeta$.

\begin{proposition}\label{1.4c}
(The asymptotic behavior). The modified Jost function has the following two kinds of asymptotic behaviors
\begin{equation}\label{1.4d}
{\mu _ \pm }(x,t,\zeta ) = {e^{iT\int_{ \pm \infty }^x {(q_0^2 - |q{|^2} )} dy{\mathop{\sigma_3} }}}\left( {\mathbb{E} + \frac{{\mu _ \pm ^{(1)}}(x,t)}{\zeta } + O(\frac{1}{{{\zeta ^2}}})} \right),  \qquad \zeta  \to \infty ,
\end{equation}
where
$$ \mu _ \pm ^{(1)}(x,t) = \left( {\begin{array}{*{20}{c}}
*&{ - iq{e^{ - 2iT\int_{ \pm \infty }^x {(q_0^2 - |q{|^2} )dy} }}}\\
{  i\bar q{e^{2iT\int_{ \pm \infty }^x {(q_0^2 - |q{|^2})dy} }}}&*
\end{array}} \right),$$
and
\begin{equation}\label{1.4e}
{\mu _ \pm }(x,t,\zeta ) = {e^{iT\int_{ \pm \infty }^x {(q_0^2 - |q{|^2})} dy}}( i{\zeta ^{ - 1}}{\mathop{\sigma_3} }{Q_ \pm }) + O(1), \qquad \zeta \to 0 .
\end{equation}
\end{proposition}
From the Wronskian presentations of the scattering data and proposition (\ref{1.4c}), we can also get the asymptotic behavior of the scattering matrix as follows
\begin{proposition}\label{1.4f}
The asymptotic behavior of the scattering matrix can be expressed as follows
\begin{equation}\label{1.4g}
s(\zeta ) = {e^{iT\int_{ - \infty }^{ + \infty } {(q_0^2 - |q{|^2})dy{\mathop{\sigma_3} }} }}\mathbb{E} + O({\zeta ^{ - 1}}),  \qquad \zeta \to \infty ,
\end{equation}
\begin{equation}\label{1.4h}
s(\zeta ) = {e^{-iT\int_{ - \infty }^{ + \infty } {(q_0^2 - |q{|^2})dy{\mathop{\sigma_3} }} }}diag(\frac{{{q_ - }}}{{{q_ + }}},\frac{{{q_ + }}}{{{q_ - }}}) + O(\zeta),  \qquad \zeta \to 0 .
\end{equation}
\end{proposition}
Hence, one can derived the potential function of the KE equation with NZBCs from equation (\ref{1.4d}) and (\ref{1}), we have
\begin{equation}\label{1.4i}
u(x,t) = {e^{i\Upsilon}}\mathop {\lim }\limits_{\zeta  \to \infty } i\zeta {\mu _{ \pm 12}}(x,t){e^{iT\int_{ \pm \infty }^x {(q_0^2 - |q{|^2})} dy}} ,
\end{equation}
where ${\Upsilon=\alpha x + 2( - q_0^2 + 2T^2q_0^4 - \frac{{{\alpha ^2}}}{2})t}$.

\subsection{The zeros of $a(\zeta)$}
Denote $\zeta_j={q_0}{e^{i {\beta _j}}},\beta \in (0,\pi).$
Then, ${\zeta _j} = {q_0}( {A_j} + i{B_j})$, where $A_j^2 + B_j^2 = 1$, and $B_j > 0$. From the equation (\ref{1.2b}), we have
\begin{equation}\label{A-1}
\kappa_j = q_0 A_j ,\qquad \lambda_j = i q_0 B_j.
\end{equation}
From the equation (\ref{1.3a}), there are some constants $b_j$ and $\tilde b_j$ that are independent of x and t, obtaining
\begin{equation}\label{A-2}
{\psi _{ - 1}}(x,t,\zeta _j) = {b_j}{\psi _{ + 2}}(x,t,\zeta _j) ,\qquad {\psi _{ - 2}}(x,t,\bar \zeta _j) = {\tilde b_j}{\psi _{ + 1}}(x,t, \tilde \zeta _j) ,
\end{equation}
here we can see that $\tilde b_j = \bar b_j$. From the second symmetry property (\ref{1.3m}) and $q_0^2\zeta_j^{-1} = \bar \zeta_j$, it is easy to get the following relationship
\begin{equation}\label{A-2}
\bar b_j = -{b_j}{e^{i(\theta_+ + \theta_-)}},
\end{equation}
it points out that $\theta_+ = -\theta_-$, we have $\bar b_j=-b_j$, and it can be inferred that $b_j$ is a pure imaginary number at this time. \par
The spectral problem (\ref{1.1c}) satisfied by $\Psi$ can be rewritten as
\begin{equation}\label{A-3}
{\Psi _x} + i\kappa {\sigma _3}\Psi  + iT(|q{|^2} - q_0^2){\sigma _3}\Psi  = Q\Psi ,
\end{equation}
and the derivative of $\zeta$ can be obtained
\begin{equation}\label{A-4}
{\Psi _{x,\zeta}} + i(\kappa {\sigma _3}\Psi)_{\zeta}  + iT(|q{|^2} - q_0^2){\sigma _3}\Psi_{\zeta}  = Q\Psi_{\zeta} ,
\end{equation}
Note the Wronskian representation of $a(\zeta)$ and definition of simple zeros, we have
\begin{equation}\label{A-5}
\dot a({\zeta }) = (1 - q_0^2{\zeta ^{ - 2}})\left( {W[{\psi _{ - 1,\zeta }},{\psi _{ + 2}}] + W[{\psi _{ - 1}},{\psi _{ + 2,\xi }}]} \right).
\end{equation}
In addition, we can yield
\begin{equation}\label{A-6}
\frac{d}{{dx}}W[{\psi _{ - 1,\zeta }},{\psi _{ + 2}}] = {\psi _{ - 11,\zeta ,x}}{\psi _{ + 22}} + {\psi _{ - 11,\zeta }}{\psi _{ + 22,x}} - {\psi _{ + 12,x}}{\psi _{ - 21,\zeta }} - {\psi _{ + 12}}{\psi _{ - 21,\zeta ,x}}
\end{equation}
\begin{equation}\label{A-7}
\frac{d}{{dx}}W[{\psi _{ - 1}},{\psi _{ + 2,\zeta }}] = {\psi _{ - 11,x}}{\psi _{ + 22,\zeta}} + {\psi _{ - 11}}{\psi _{ + 22,\zeta ,x}} - {\psi _{ + 12,\zeta ,x}}{\psi _{ - 21}} - {\psi _{ + 12,\zeta }}{\psi _{ - 21,x}}
\end{equation}
From equations (\ref{A-4}) and (\ref{A-5}), we can solve
\begin{equation}\label{A-8}
\frac{d}{{dx}}W[{\psi _{ - 1,\zeta }},{\psi _{ + 2}}] =  - \frac{i}{2}(1 - \frac{{q_0^2}}{{{\zeta ^2}}})\left( {{\psi _{ - 11}}{\psi _{ + 22}} + {\psi _{ - 21}}{\psi _{ + 12}}} \right),
\end{equation}
\begin{equation}\label{A-9}
\frac{d}{{dx}}W[{\psi _{ - 1}},{\psi _{ + 2,\zeta }}] = \frac{i}{2}(1 - \frac{{q_0^2}}{{{\zeta ^2}}})\left( {{\psi _{ - 11}}{\psi _{ + 22}} + {\psi _{ - 21}}{\psi _{ + 12}}} \right).
\end{equation}
Integrate the equation (\ref{A-8}) from $-l$ to $x$, and integrate the equation (\ref{A-9}) from $x$ to $l$,
\begin{equation}\label{A-10}
W[{\psi _{ - 1}},{\psi _{ + 2,\zeta }}] =  - \int_x^l {\frac{i}{2}(1 - \frac{{q_0^2}}{{{\zeta ^2}}})\left( {{\psi _{ - 11}}{\psi _{ + 22}} + {\psi _{ - 21}}{\psi _{ + 12}}} \right)dx},
\end{equation}
\begin{equation}\label{A-11}
W[{\psi _{ - 1,\zeta }},{\psi _{ + 2}}] = \int_{ - l}^x { - \frac{i}{2}(1 - \frac{{q_0^2}}{{{\zeta ^2}}})\left( {{\psi _{ - 11}}{\psi _{ + 22}} + {\psi _{ - 21}}{\psi _{ + 12}}} \right)dx}.
\end{equation}
From the equations (\ref{A-10}) and (\ref{A-11}), let $\zeta = \zeta_j$, pay attention to the relationship (\ref{A-2}), and $l$ tend to infinity. Then
\begin{equation}\label{A-12}
\dot a({\zeta _j}) =  - \int_{ - \infty }^{ + \infty } {i2{b_j}({\psi _{ + 12}} \times {\psi _{ + 22}})dx}.
\end{equation}
Note the equation (\ref{1.3n}), Then
\begin{equation}\label{A-13}
\dot a({\zeta _j}) = 2{b_j}{q_ + }\zeta _j^{ - 1}\int_{ - \infty }^{ + \infty } {|{\psi _{ + 22}}{|^2}dx} .
\end{equation}
So the scattering data $a(\zeta_j)$ can only be a simple zero.

\subsection{Residue conditions}

Different from the focusing Kundu-Eckhaus equation, the scattering spectrum problem (\ref{1.c}) and (\ref{1.1c}) is a self-conjugate operator, and such eigenvalues $\zeta_j$ can only be real. It can be inferred that the corresponding $\zeta_j$ must be on the upper semicircle $O/\{ \pm {q_0}\}$ of the the complex $\zeta$-plane due to the discrete eigenvalues of scattering problem are the zeros of $a(\zeta)$. Recall that the scattering data $a(\zeta)$ is analytic in the upper half of the complex $\zeta$-plane and has a finite number $N_1$ of simple zeros on the upper semicircle $O$.
That is, $a({\zeta _j}) = 0$, $\dot a({\zeta _j}) \ne 0$ and $\left\{ {{\zeta _j} \in \mathbb{C}\left| {|{\zeta _j}|} \right. = {q_0},{\mathop{\rm Im}\nolimits} {\zeta _j} > 0,j = 1,2, \ldots ,{N_1}} \right\}$, here and after $\dot{\Lambda } $ is defined as the derivative of $\zeta $. \par

In fact, $a(\zeta)$ and $\tilde a(\zeta)$ have zeros on the upper and lower semicircles $O$, respectively, unless the $\psi _ \pm$ linearly dependent at the branch points $\pm q_0$, then the Jost solutions $\psi _ \pm$ are continuous, and the scattering coefficients usually have simple poles at $\zeta=\pm q_0$. In this case, the scattering data $a(\zeta)$ and $\tilde a(\zeta)$ are non-singular near the corresponding branch point and the point $\zeta=q_0$ or $\zeta=-q_0$ is called a virtual level in scattering theory \cite{L.D.F,G.Bi}.

From the symmetry conditions (\ref{1.3d}) and (\ref{1.3o}), we know that
\begin{equation}
a({\zeta _j}) = \tilde a({\bar \zeta _j}) = \tilde a({\hat \zeta _j}) = a({\check \zeta _j}) = 0 .
\end{equation} \par
Thus, denoting the discrete spectrum set of the scattering problem is $$ U = \left\{ {{\zeta _j},{{\bar \zeta }_j},{{\hat \zeta }_j},{\check \zeta _j}} \right\}_{j = 1}^{{N_1}} .$$ \par
We have the following residue condition
\begin{equation}\label{1.5c}
{\mathop{\rm Re}\nolimits} {s_{\zeta  = {\zeta _j}}}\left[ {\mu _ - ^{(1)}(x,t,\zeta )/a(\zeta )} \right] = {c_j}{e^{2i\Theta ({\zeta _j})}}\mu _ + ^{(1)}(x,t,{\zeta _j}) ,
\end{equation}
where ${c_j} = \frac{{{b_j}}}{{\dot a({\zeta _j})}}$.

\begin{equation}\label{1.5d}
{\mathop{\rm Re}\nolimits} {s_{\zeta  = {{\bar \zeta }_j}}}\left[ {\mu _ - ^{(2)}(x,t,\zeta )/\tilde a(\zeta )} \right] = {\tilde c_j}{e^{ - 2i\Theta ({{\bar \zeta }_j})}}\mu _ + ^{(2)}(x,t,{\bar \zeta _j}) ,
\end{equation}
where ${\tilde c_j} = \frac{{{{\tilde b}_j}}}{{\dot {\tilde a}({{\bar \zeta }_j})}}$, from symmetric conditions (\ref{1.3d}), known ${\bar c_j} = {\tilde c_j}$.

Furthermore,
\begin{equation}\label{1.5e}
{\mathop{\rm Re}\nolimits} {s_{\zeta  = {{\hat \zeta }_j}}}\left[ {\mu _ - ^{(2)}(x,t,\zeta )/\tilde a(\zeta )} \right] = {\tilde c_{{N_1} + j}}{e^{  -2i\Theta ({{\hat \zeta }_j})}}\mu _ + ^{(2)}(x,t,{\hat \zeta _j}) ,
\end{equation}
where ${\tilde c_{{N_1} + j}} = \frac{{{{\tilde b}_j}}}{{\dot {\tilde a}({{\hat \zeta }_j})}}$.
\begin{equation}\label{1.5f}
{\mathop{\rm Re}\nolimits} {s_{\zeta  = {\check \zeta _j}}}\left[ {\mu _ - ^{(1)}(x,t,\zeta )/ a(\zeta )} \right] = {c_{{N_1} + j}}{e^{ 2i\Theta ({\check \zeta _j})}}\mu _ + ^{(1)}(x,t,{\check \zeta _j}) ,
\end{equation}
where ${c_{{N_1} + j}} = \frac{{{b_j}}}{{\dot a({\check \zeta _j})}}$, and $\overline{{c_{{N_1} + j}}} =  {\tilde c_{{N_1} + j}}$.
\begin{remark}
As shown in fig.(\ref{fig.1aa}), while the discrete eigenvalues of scattering problems located on the semicircle $O/\{ \pm {q_0}\}$, we have ${\zeta _j} = q_0^2{\bar \zeta _j}$ and ${\bar \zeta _j} = q_0^2{\zeta _j}$ by the transformation (\ref{1.3l}).
\end{remark}
\subsection{Riemann-Hilbert problem}
Define the sectionally meromorphic matrices
\begin{equation}\label{1.6a}
M(x,t,\zeta ) = \left\{ \begin{array}{l}
M_1={e^{iT\int_x^\infty  {(q_0^2-|q{|^2} )dy{\mathop{\sigma_3} }} }}\left( {\frac{{\mu _ - ^{(1)}(x,t,\zeta )}}{{a(\zeta )}},\mu _ + ^{(1)}} \right),\zeta  \in {D_1} .\\
M_2={e^{iT\int_x^\infty  {(q_0^2-|q{|^2} )dy{\mathop{\sigma_3} }} }}\left( {\mu _ + ^{(2)},\frac{{\mu _ - ^{(2)}(x,t,\zeta )}}{{\tilde a(\zeta )}}} \right),\zeta  \in {D_2}.
\end{array} \right.
\end{equation}
Then, we can show $M(x,t,\zeta)$ satisfies the Riemann-Hilbert problem:\par
(1) Analyticity: \label{1.6b} \par
 \qquad  \qquad \qquad  \quad $M(x,t,\zeta)$ is analytic in $({D_1} \cup {D_2})/ U$ and has simple poles in the set $U$.\par
(2) Jump condition:
\begin{equation}\label{1.6c}
{M_1}(x,t,\zeta ) = {M_2}(x,t,\zeta )  J(x,t,\zeta ), \qquad \zeta  \in {\mathbb{R}} ,
\end{equation}
where
$$J(x,t,\zeta ) = \left( {\begin{array}{*{20}{c}}
{1 - |r(\zeta ){|^2}}&{ -\tilde r(\zeta ){e^{ - 2i\Theta (x,t,\zeta )}}}\\
{r(\zeta ){e^{2i\Theta (x,t,\zeta )}}}&1
\end{array}} \right) .$$ \par
(3) Asymptotic behavior:\label{1.6d} \par
From the proposition (\ref{1.4c})(\ref{1.4f}) and definition (\ref{1.3n}), we have
\begin{equation}\label{1.6e}
{M_{1,2}}(x,t,\zeta ) = \mathbb{E} + O(\frac{1}{\zeta }), \qquad \zeta  \to \infty ,
\end{equation}
\begin{equation}\label{1.6f}
{M_{1,2}}(x,t,\zeta ) =   i{\zeta ^{ - 1}}{\mathop{\sigma_3} }{Q_ + ^T } + O(1)  , \qquad \zeta  \to 0 .
\end{equation}\par
Letting $\eta _j=\zeta _j(q_0^2\bar \zeta_j^{-1})$, then equation (\ref{1.6c}) can be regularized as follows,
\begin{equation}\label{1.6g}
\begin{array}{l}
{M_1}(x,t,\zeta ) - \mathbb{E} - i{\zeta ^{ - 1}}{\mathop{\sigma_3} }{Q_ + ^T } - \sum\limits_{j = 1}^{{N_1}} {\left[{\frac{{{\mathop{\rm Re}\nolimits} {s_{\zeta  = {\eta _j}}}{M_1}(\zeta )}}{{\zeta  - {\eta _j}}}} \right]}  - \sum\limits_{j = 1}^{{N_1}} {\left[ {\frac{{{\mathop{\rm Re}\nolimits} {s_{\zeta  = {{\bar \eta }_j}}}{M_2}(\zeta )}}{{\zeta  - {{\bar \eta }_j}}}} \right]} \\
 = {M_2}(x,t,\zeta ) - \mathbb{E} - i{\zeta ^{ - 1}}{\mathop{\sigma_3} }{Q_ + ^T } - \sum\limits_{j = 1}^{{N_1}} {\left[ {\frac{{{\mathop{\rm Re}\nolimits} {s_{\zeta  = {{\bar \eta }_j}}}{M_2}(\zeta )}}{{\zeta  - {{\bar \eta }_j}}}} \right]}  - \sum\limits_{j = 1}^{{N_1}} {\left[ {\frac{{{\mathop{\rm Re}\nolimits} {s_{\zeta  = {\eta _j}}}{M_1}(\zeta )}}{{\zeta  - {\eta _j}}}} \right]}  + {M_2}(\zeta )  \tilde J(\zeta) ,
\end{array}
\end{equation}
where
$$\tilde J = \left( \begin{array}{*{20}{c}}
{-|r(\zeta )|^2}&{-\tilde r(\zeta ){e^{ - 2i\Theta (\zeta )}}}\\
{  r(\zeta ){e^{2i\Theta (\zeta )}}}&0
\end{array}
\right).$$
The left side of the above formula is analytic in $D_1$, and the right side is analytic in $D_2$ except for the last item ${M_2}(\zeta ) \tilde J(\zeta)$. The asymptotic behaviors of both are $O(\frac{1}{\zeta })$ as $\zeta  \to \infty $ and $O(1)$ as $\zeta  \to 0$. From (\ref{1.4d}) and (\ref{1.4e}), $\tilde J$ are $O(\frac{1}{\zeta })$ as $\zeta  \to \pm\infty $ and $O(\zeta)$ as $\zeta  \to 0$ along the real axis.
Introduce the Cauchy projectors and by Plemelj's formulae, the equation (\ref{1.6g}) can be written as

\begin{equation}\label{1.6i}
\begin{split}
M(x,t,\zeta ) = \mathbb{E} + i{\zeta ^{ - 1}}{\mathop{\sigma_3} }{Q_ + ^T } + \sum\limits_{j = 1}^{{N_1}} {\left[ {\frac{{{\mathop{\rm Re}\nolimits} {s_{\zeta  = {\eta _j}}}{M_1}(x,t,\zeta )}}{{\zeta  - {\eta _j}}} + \frac{{{\mathop{\rm Re}\nolimits} {s_{\zeta  = {{\bar \eta }_j}}}{M_2}(x,t,\zeta )}}{{\zeta  - {{\bar \eta }_j}}}} \right]}  \\
- \frac{1}{{2\pi i}}\int_{- \infty}^{+ \infty} {\frac{{{M_2}(x,t,\eta )\tilde J(x,t,\eta )}}{{\eta  - \zeta }}d\eta },
\end{split}
\end{equation}
note that, the solution of the Riemann-Hilbert problem (\ref{1.6i}) with NZBC for $M_1$ and $M_2$ differs only in the last integral item corresponding to $P_+$ and $P_-$ Cauchy projector, respectively.
\begin{remark}
The solution of the Riemann Hilbert problem under non-zero boundary conditions constitutes a closed algebraic system.
\end{remark}
\begin{proof}
By the solution of RHP, calculate the second column of $M$ at points $\zeta  = \eta_j= { \zeta _j} (q_0^2  \bar \zeta _j^{ - 1})$, then
\begin{equation}\label{1.6j}
\begin{split}
{\left( {\begin{array}{*{20}{c}}
{{e^{i{\omega _ + }}}}\\
{{e^{ - i{\omega _ + }}}}
\end{array}} \right)} \ast \mu _{+}^{(1)}(x,t,{\eta _j}) = \left( \begin{array}{*{20}{c}}
  i\eta _j^{ - 1}{\bar q_ + }\\
  1
\end{array} \right) + \sum\limits_{j_1 = 1}^{{N_1}} {{\left( {\begin{array}{*{20}{c}}
{{e^{i{\omega _ + }}}}\\
{{e^{ - i{\omega _ + }}}}
\end{array}} \right)} \ast \frac{{{{\tilde c}_{j_1}}{e^{ - 2i\Theta ({{\bar \eta }_{j_1}})}}}}{{{\eta _j} - {{\bar \eta }_{j_1}}}}} \mu _+^{(2)}(x,t,{\bar \eta _{j_1}}) \\- \frac{1}{{2\pi i}}\int_{- \infty}^{+ \infty} {\frac{{{{({M_1}\tilde J)}_2}}}{{\eta  - {\eta _j}}}d\eta }, \qquad j=1,2,...,N_1 ,
\end{split}
\end{equation}
here and after, $' \ast '$ stands for Hadamard product. From the symmetry property (\ref{1.3m}), we have
\begin{equation}\label{1.6k}
\begin{split}
\mu _ + ^{(2)}({\bar \eta _j}) = \left( {\begin{array}{*{20}{c}}  {{e^{ - i{\omega _ + }}} \left( - \frac{{{{\bar q}_ + }}}{{{q_ + }}}\right)} \\   {e^{i{\omega _ + }}} \left(i{\eta _j}q_ + ^{ - 1} \right) \end{array}} \right) + i{\eta _j}q_ + ^{ - 1}\left\{ {\sum\limits_{{j_1} = 1}^{{N_1}} {\frac{{{{\tilde c}_{{j_1}}}{e^{ - 2i\Theta ({{\bar \eta }_{{j_1}}})}}}}{{{\eta _j} - {{\bar \eta }_{{j_1}}}}}} \mu _ + ^{(2)}({{\bar \eta }_{{j_1}}}) - \frac{1}{{2\pi i}}\int_{ - \infty }^{ + \infty } {\frac{{{{({M_1}\tilde J)}_2}}}{{\eta  - {\eta _j}}}d\eta } } \right\}, \\ \qquad j=1,2,...,N_1 ,
\end{split}
\end{equation}
where
 \begin{equation}\label{1.6l}
{\left( {\begin{array}{*{20}{c}}
{{e^{i{\omega _ + }}}}&{{e^{ - i{\omega _ + }}}}
\end{array}} \right)^{\rm T}} = {\left( {\begin{array}{*{20}{c}}
{{e^{i\int_x^{ + \infty } {T(q_0^2 - |q{|^2})dy} }}}&{{e^{-i\int_x^{ + \infty } {T(q_0^2 - |q{|^2})dy} }}}
\end{array}} \right)^{\rm T}}.
 \end{equation} \par
From these equations (\ref{1.6k}), for $j=1,2,...,N_1$, $N_1$ equations and $N_1$ unknown $\mu _+^{(2)}(x,t,{\bar \eta _j})$ are formed. Combined with the RHP solution (\ref{1.6i}), a closed algebraic system of $M$ is provided by the scattering data.
\end{proof}

\subsection{Trace formulae and the condition $(\theta) $}
Letting $T\int_{ - \infty }^{ + \infty } {(q_0^2 - |q{|^2})dy}  \mapsto  {\omega _0}$, obtaining
\begin{equation}\label{1.7d}
 a(\zeta ) = \exp \left[ { - \frac{1}{{2\pi i}}\int_{-\infty}^{+\infty} {\frac{{\log [1 - |r(\zeta' ){|^2}]}}{{\zeta'  - \zeta }}d\zeta' } } \right]\prod\limits_{j = 1}^{{N_1}} {\frac{{(\zeta  - {\zeta _j})}}{{(\zeta  - {{\bar \zeta }_j})}}} {e^{i{\omega _ 0 }}},
\end{equation}
similarly, we have
\begin{equation}\label{1.7e}
\tilde a(\zeta ) = \exp \left[ {\frac{1}{{2\pi i}}\int_{-\infty}^{+\infty} {\frac{{\log [1 - |r(\zeta' ){|^2}]}}{{\zeta'  - \zeta }}d\zeta' } } \right]\prod\limits_{j = 1}^{{N_1}} {\frac{{(\zeta  - {{\bar \zeta }_j})}}{{(\zeta  - {\zeta _j})}}} {e^{ - i{\omega _ 0 }}},
\end{equation}
equations (\ref{1.7d}) and (\ref{1.7e}) are so called trace formulas. As $\zeta  \to 0$, from the symmetry of the scattering data (\ref{1.3o}), obtaining
\begin{equation}\label{1.7f}
\arg \left( {\frac{{{q_ + }}}{{{q_ - }}}} \right) - 2{\omega _0} = 2\sum\limits_{j = 1}^{{N_1}} {\arg ({\zeta _j})}  + \frac{1}{{2\pi }}\int_{-\infty}^{+\infty} {\frac{{\log{ [1 - |r(\zeta' ){|^2}  ]}}}{\zeta' }d\zeta' } ,
\end{equation}
this relationship is called the condition $(\theta) $ in \cite{L.D.F}, also known $''$theta$''$ condition, which establishes the relationship between the ${\omega _ 0}$, the asymptotic phase difference, the discrete spectrum values and the scattering coefficients. \par
Moreover, notice that $r(\zeta )=-\frac{{{{\bar q}_ + }}}{{{ q_ + }}}\tilde r(\hat \zeta )$, we have
\begin{equation}\label{1.7g}
\int_{ - \infty }^{ + \infty } {\frac{{\ln [1 - |r(\zeta' ){|^2}]}}{{\zeta'}}} d\zeta ' = 2\int_{ - \infty }^{ - {q_0}} {\frac{{\ln [1 - |r(\zeta' ){|^2}]}}{{\zeta '}}} d\zeta ' + 2\int_{{q_0}}^{ + \infty } {\frac{{\ln [1 - |r(\zeta' ){|^2}]}}{{\zeta '}}} d\zeta ' ,
\end{equation}
owing to $0 < 1- |r(\zeta){|^2}<1$, which indicates that the first term on the right side is positive and the second term is negative, and there is no symmetry between the two integrals. Therefore, in principle, the radiative part of the spectrum provides a nontrivial contribution to the asymptotic phase difference of the potential. It is noteworthy that the ${\omega _ 0}$ will also yield a certain phase difference, which is similar to the focusing KE equation.

\subsection{Reconstruction formula for the potential}
According to the solution of RHP(\ref{1.6i}), the asymptotic expansion formula of $M(x,t,\zeta)$ as $\zeta  \to \infty $ can be derived as follows,
\begin{equation}\label{1.8a}
M(x,t,\zeta ) = \mathbb{E} + \frac{{{M^{(1)}}(x,t,\zeta )}}{\zeta } + O(\frac{1}{{{\zeta ^2}}}), \qquad   \zeta  \to \infty ,
\end{equation}
where
\begin{equation}\label{1.8b}
 \begin{split}
\begin{array}{l}
{M^{(1)}}(\zeta ) =\!   i{\mathop{\sigma_3} }{Q_ + ^{T} } + \sum\limits_{j = 1}^{{N_1}} {\left[ {\begin{array}{*{20}{c}}
{{\left( {\begin{array}{*{20}{c}}
{{e^{i{\omega _ + }}}}\\
{{e^{ - i{\omega _ + }}}}
\end{array}} \right)} \ast  {c_j}{e^{2i\Theta ({\eta _j})}}\mu _ + ^{(1)}({\eta _j}),}&{ {\left( {\begin{array}{*{20}{c}}
{{e^{i{\omega _ + }}}}\\
{{e^{ - i{\omega _ + }}}}
\end{array}} \right)} \ast {{\tilde c}_j}{e^{ - 2i\Theta ({{\bar \eta }_j})}}\mu _ + ^{(2)}({{\bar \eta }_j})}
\end{array}} \right]}
\end{array}\\  + \frac{1}{{2\pi i}}\int_{-\infty}^{+\infty} {({M_2}\tilde J)d\eta }.
\end{split}
\end{equation}\par
From the sectionally meromorphic matrices, letting ${e^{-iT\int_x^\infty  {(q_0^2 - |q{|^2})} dy{\mathop{\sigma_3} }}}M(x,t,\zeta ){e^{ - i\Theta (x,t,\zeta ){\mathop{\sigma_3} }}}$ and choose $M(x,t,\zeta)=M_2(x,t,\zeta)$. Bring it into the Lax pari (\ref{1.1c}) and calculate the coefficient of the 1,2 elements of $M(x,t,\zeta)$ in $\zeta^0$, then we have the following results.
\begin{proposition}\label{1.8c}
The reconstruction formula for the potential with simple poles of the defocusing Kundu-Eckhaus equation with NZBCs can be expressed as follows
\begin{equation}\label{1.8d}
u(x,t) = {e^{ i(\Upsilon-2\omega_+)}} \left[ - {\bar q_ + } + i{e^{ + i{\omega _ + }}}\sum\limits_{j = 1}^{{N_1}} {{{\tilde c}_j}{e^{ - 2i\Theta ({{\bar \eta }_j})}}\mu _{ + 1,2}^{(2)}(x,t,{{\bar \eta }_j})}  + \frac{1}{{2\pi }}\int_{ - \infty }^{ + \infty } {{{\left( {{M_2}\tilde J} \right)}_{1,2}}d\eta } \right]
\end{equation}
\end{proposition}

\subsection{Reflectionless potential: Multi-soliton solutions}
The soliton correspond to the scattering data $b(\zeta)$ and the reflect coefficients $r(\zeta )$ vanishes identically with it is referred as reflectionless, i.e $\tilde J \equiv 0$.
From the closed algebraic system (\ref{1.6k}), we have
\begin{equation}\label{1.9a}
\sqrt {{{\tilde c}_j}} \psi _{ + 1,2}^{(2)}({\zeta _j}) = {e^{ - i{\omega _ + }}}\left( { - \frac{{{{\bar q}_ + }}}{{{q_ + }}}\sqrt {{{\tilde c}_j}} {e^{ - i\Theta ({{\bar \zeta }_j})}}} \right) - \sum\limits_{{j_1} = 1}^{{N_1}} {\frac{{i{\zeta _j}q_ + ^{ - 1}{{\tilde c}_{{j_1}}}{e^{ - i\Theta ({{\bar \zeta }_{{j_1}}})}}\sqrt {{{\tilde c}_j}} {e^{ - i\Theta ({{\bar \zeta }_j})}}}}{{{\bar \zeta _j} - {{ \zeta }_{{j_1}}}}}\psi _{ + 1,2}^{(2)}({{\bar \zeta }_j})} ,
\end{equation}
let
\begin{equation}\label{1.9.1a}
{\psi _j} = \sqrt {{{\tilde c}_j}} \psi _{ + 1,2}^{(2)}({\zeta _j}),{f_j} = \sqrt {{{\tilde c}_j}} {e^{ - i\Theta ({{\bar \zeta }_j})}},{g_j} =  - \frac{{{{\bar q}_ + }}}{{{q_ + }}}{f_j},
\end{equation}
\begin{equation}\label{1.9.2a}
{B_{j,{j_1}}} = \frac{{{f_j}(i{\zeta _j}q_ + ^{ - 1}){f_{{j_1}}}}}{{{{\bar \zeta }_j} - {\zeta _{{j_1}}}}},
\end{equation}
then,(\ref{1.9a}) can be written in matrix forms,
\begin{equation}\label{1.9.3a}
\Psi  = {e^{ - i{\omega _ + }}}g{\left( {\mathbb{E} + B} \right)^{ - 1}},
\end{equation}
combined with the potential formula (\ref{1.8d}), then
\begin{equation}\label{1.9.4a}
u(x,t) = {e^{i(\Upsilon - 2{\omega _ + })}}\left( { - {{\bar q}_ + }\frac{{\det (\mathbb{E} + B')}}{{\det (\mathbb{E} + B)}}} \right).
\end{equation}
In order to express the multi-soliton solutions intuitively, we have
\begin{equation}\label{1.9.4a}
\det (\mathbb{E} + B') = 1 + \sum\limits_{r = 1}^{{N_1}} {\sum\limits_{1 \leqslant {n_1} < {n_2} <  \ldots  \leqslant {n_r} \leqslant {N_1}} {B'({n_{1,}}{n_2}, \ldots ,{n_r})} } ,
\end{equation}
\begin{equation}\label{1.9.4a}
\det (\mathbb{E} + B) = 1 + \sum\limits_{r = 1}^{{N_1}} {\sum\limits_{1 \leqslant {n_1} < {n_2} <  \ldots  \leqslant {n_r} \leqslant {N_1}} {B({n_{1,}}{n_2}, \ldots ,{n_r})} } ,
\end{equation}
where ${B'({n_{1,}}{n_2}, \ldots ,{n_r})}$ and ${B({n_{1,}}{n_2}, \ldots ,{n_r})}$ are $rth$-order principal minors of $B'$ and $B$, respectively.
\begin{equation}\label{1.9.4a}
\begin{split}
B'({n_{1,}}{n_2}, \ldots ,{n_r}) &= {(i)^r}{\prod\limits_j {f_j^2{{\bar \zeta }_j}q_ + ^{ - 1}\left( {{{\bar \zeta }_j} - {\zeta _j}} \right)} ^{ - 1}}{\prod\limits_{j < {j_1}} {\left| {\frac{{ - {\zeta _j} + {\zeta _{{j_1}}}}}{{ - {{\bar \zeta }_j} + {\zeta _{{j_1}}}}}} \right|} ^2} \\& = {( - 1)^r}{(q_0^{ - 1})^r}{e^{ - i\sum\nolimits_j {{\theta _ + }} }}{e^{ - i\sum\nolimits_j {{\beta _j}} }}{e^{ - 2i\sum\nolimits_j {\Theta ({{\bar \zeta }_j})} }}\prod\limits_j {\frac{{{{\tilde c}_j}}}{{2\sin ({\beta _j})}}\prod\limits_{j < {j_1}} {\frac{{{{\sin }^2}\left( {\frac{{{\beta _j} - {\beta _{{j_1}}}}}{2}} \right)}}{{{{\sin }^2}\left( {\frac{{{\beta _j} + {\beta _{{j_1}}}}}{2}} \right)}}} } ,
\end{split}
\end{equation}
\begin{equation}\label{1.9.4a}
\begin{split}
B({n_{1,}}{n_2}, \ldots ,{n_r}) &= {(i)^r}{\prod\limits_j {f_j^2{\zeta _j}q_ + ^{ - 1}\left( {{{\bar \zeta }_j} - {\zeta _j}} \right)} ^{ - 1}}{\prod\limits_{j < {j_1}} {\left| {\frac{{ - {\zeta _j} + {\zeta _{{j_1}}}}}{{ - {{\bar \zeta }_j} + {\zeta _{{j_1}}}}}} \right|} ^2} \\& = {( - 1)^r}{(q_0^{ - 1})^r}{e^{ - i\sum\nolimits_j {{\theta _ + }} }}{e^{i\sum\nolimits_j {{\beta _j}} }}{e^{ - 2i\sum\nolimits_j {\Theta ({{\bar \zeta }_j})} }}\prod\limits_j {\frac{{{{\tilde c}_j}}}{{2\sin ({\beta _j})}}\prod\limits_{j < {j_1}} {\frac{{{{\sin }^2}\left( {\frac{{{\beta _j} - {\beta _{{j_1}}}}}{2}} \right)}}{{{{\sin }^2}\left( {\frac{{{\beta _j} + {\beta _{{j_1}}}}}{2}} \right)}}} } ,
\end{split}
\end{equation}
where $j,{j_1} \in \{ {n_1},{n_2}, \ldots ,{n_r}\} $.

\subsubsection{One-soliton solution for $|\zeta_1|=q_0$} \label{2}
In this case, the scattering spectrum problem has two discrete eigenvalues, i.e, one is $\zeta_1$, another is $\bar \zeta_1$. It yield that the algebraic system (\ref{1.6j}) reduces to the following equations,
\begin{equation}\label{2.1a}
 \mu _ + ^{(1)}({\zeta _1}) = \left( {\begin{array}{*{20}{c}}
  {i\zeta _1^{ - 1}{{\bar q}_ + }} \\
  1
\end{array}} \right) * \left( {\begin{array}{*{20}{c}}
  {{e^{ - i{\omega _ + }}}} \\
  {{e^{i{\omega _ + }}}}
\end{array}} \right) + \frac{{{{\tilde c}_{{j_1}}}{e^{ - 2i\Theta ({{\bar \zeta }_1})}}}}{{{\zeta _1} - {{\bar \zeta }_1}}}\mu _ + ^{(2)}({\bar \zeta _1}),
\end{equation}
from the proposition (\ref{1.3m}), we have
\begin{equation}\label{2.1b}
\mu _ {+2}^{(1)}(x,t,\zeta ) =  - i{\zeta ^{ - 1}}{q_ {+} }\mu _ {+1}^{(2)}(x,t,\hat \zeta ).
\end{equation}
Then,
\begin{equation}\label{2.1c}
\mu _{ + 1,2}^{(2)}({\bar \zeta _1}) =  - \frac{{{{\bar q}_ + }}}{{{q_ + }}}{e^{ - i{\omega _ + }}}{\left( {1 - \frac{{i{\zeta _1}q_ + ^{ - 1}{{\tilde c}_1}{e^{ - 2i\Theta ({{\bar \zeta }_1})}}}}{{{\zeta _1} - {{\bar \zeta }_1}}}} \right)^{ - 1}} ,
\end{equation}
Hence, by (\ref{1.8d}), we get
\begin{equation}\label{2.1d}
u = {e^{ i(\Upsilon-2\omega_+)}}\left( { - {{\bar q}_ + } \cdot \frac{{1 - {{{{\bar \zeta }_1}} \mathord{\left/ {\vphantom {{{{\bar \zeta }_1}} {{\zeta _1}}}} \right. \kern-\nulldelimiterspace} {{\zeta _1}}} \cdot {{i{\zeta _1}q_ + ^{ - 1}{{\tilde c}_1}{e^{ - 2i\Theta ({{\bar \zeta }_1})}}} \mathord{\left/ {\vphantom {{i{\zeta _1}q_ + ^{ - 1}{{\tilde c}_1}{e^{ - 2i\Theta ({{\bar \zeta }_1})}}} {({\zeta _1} - \bar \zeta )}}} \right. \kern-\nulldelimiterspace} {({\zeta _1} - \bar \zeta_1 )}}}}{{1 - {{i{\zeta _1}q_ + ^{ - 1}{{\tilde c}_1}{e^{ - 2i\Theta ({{\bar \zeta }_1})}}} \mathord{\left/ {\vphantom {{i{\zeta _1}q_ + ^{ - 1}{{\tilde c}_1}{e^{ \Delta}}} {({\zeta _1} - \bar \zeta )}}} \right. \kern-\nulldelimiterspace} {({\zeta _1} - \bar \zeta_1 )}}}}} \right),
\end{equation}
now recall that $\zeta_1=q_0e^{i\beta_1}$, then
\begin{equation}\label{2.1d1}
u = {e^{ i(\Upsilon-2\omega_+)}}\left( { - {q_0}\frac{{{{\tilde c}_1} - 2{e^{i({\beta _1} + {\theta _ + })}}{e^{ \Delta_1 }}{q_0}\sin ({\beta _1})}}{{{{\tilde c}_1}{e^{i(2{\beta _1} + {\theta _ + })}} + 2{e^{i({\beta _1} + 2{\theta _ + })}}{e^{ \Delta_1 }}{q_0}\sin ({\beta _1})}}} \right)
\end{equation}
where $\Delta_1 =  - 2{q_0} sin(\beta_1) (x + (2{q_0}cos({\beta _1}) - 4Tq_0^2 - 2\alpha )t)$, $\beta \in (0, \pi)$, $\theta_+$, $T$, $q_0$, $\alpha$ are some real constants.\par
If we choose ${\tilde c_1}= \gamma_1 {e^{i\gamma_2}}(\gamma_1, \gamma_2\in \mathbb{R}) $, then via symbolic computation, obtaining
\begin{equation}\label{2.1d2}
{\omega _ + } = 2T{q_0}(1+ \arctan\{ cot(\gamma_2  + {\beta _1} - {\theta _ + }) - \frac{{2{e^{\Delta_1} }{q_0}\csc (\gamma_2  + {\beta _1} - {\theta _ + })sin({\beta _1})}}{\gamma_1 }\})csc(\gamma_2  + {\beta _1} - {\theta _ + }),
\end{equation}
fig.(\ref{2.1df2})-fig.(\ref{2.1df3}) shows some of the dynamic structures of this case.
%
%

\begin{figure}[htpb]
\centering
{
\begin{minipage}[t]{0.315\linewidth}
\centering
\includegraphics[width=5cm]{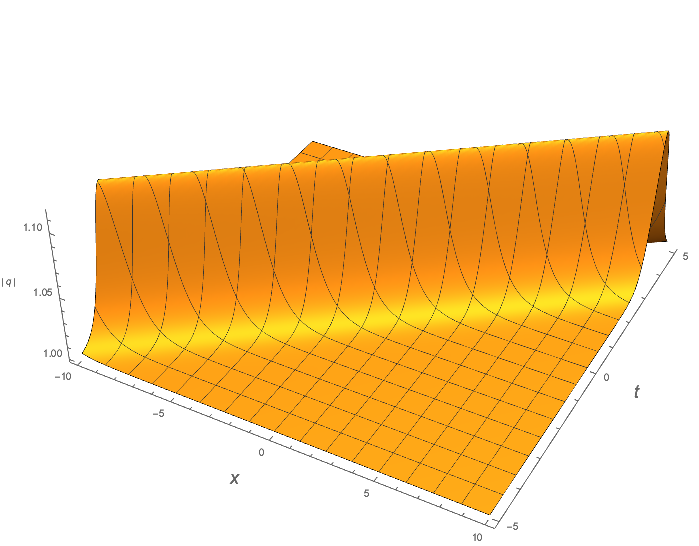}
\end{minipage}
}
{
\begin{minipage}[t]{0.315\linewidth}
\centering
\includegraphics[width=5.5cm]{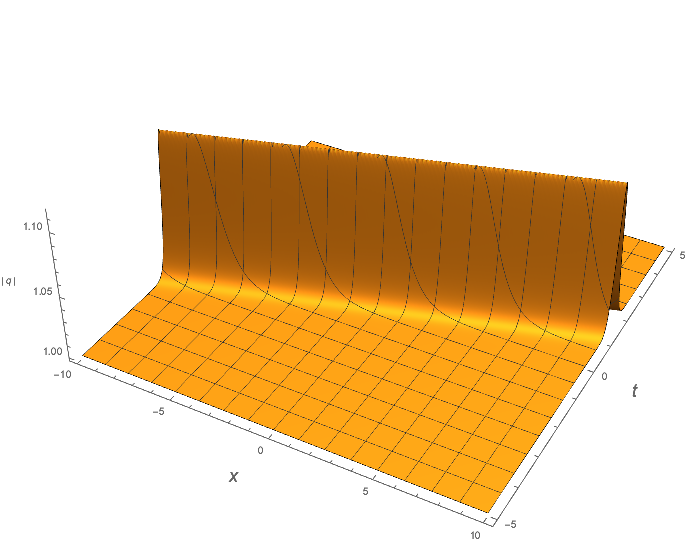}
\end{minipage}
}
{
\begin{minipage}[t]{0.315\linewidth}
\centering
\includegraphics[width=5.5cm]{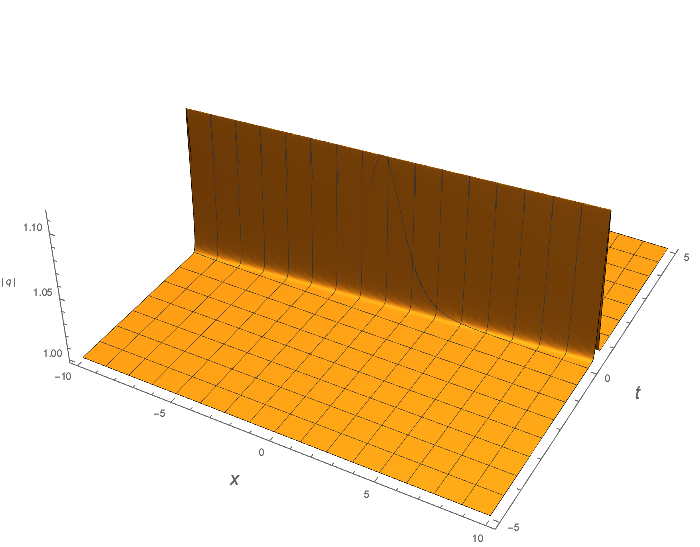}
\end{minipage}
}

\subfigure[]
{
\begin{minipage}[t]{0.31\linewidth}
\centering
\includegraphics[width=4cm]{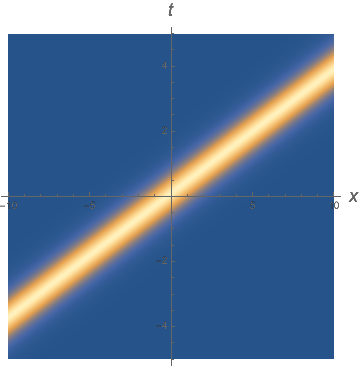}
\end{minipage}
}
\subfigure[]
{
\begin{minipage}[t]{0.31\linewidth}
\centering
\includegraphics[width=4cm]{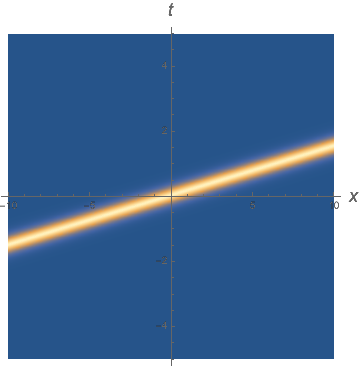}
\end{minipage}
}
\subfigure[]
{
\begin{minipage}[t]{0.31\linewidth}
\centering
\includegraphics[width=4cm]{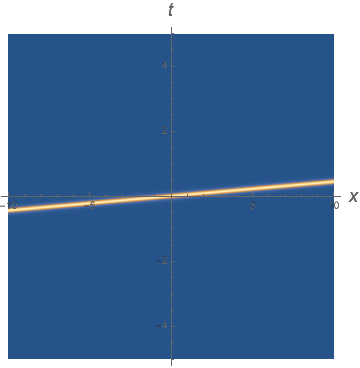}
\end{minipage}
}
\centering
\caption{Bright soliton solutions of the defocusing KE equation with parameters $q_0=1$,$c_1=\frac{\pi}{4}e^{\frac{\pi}{4}}$, $\theta_+=\frac{\pi}{4}$, $\alpha=2$, $\beta_1=\frac{\pi}{4}$, $T=\frac {1} {100}$(left), $T=1$(middle), $T=5$(right). (a)$\sim$(c) display the density of the corresponding anti-dark soliton.}
\label{2.1df2}
\end{figure}
\begin{figure}[htb]
\centering
\subfigure[]
{\includegraphics[scale=0.2]{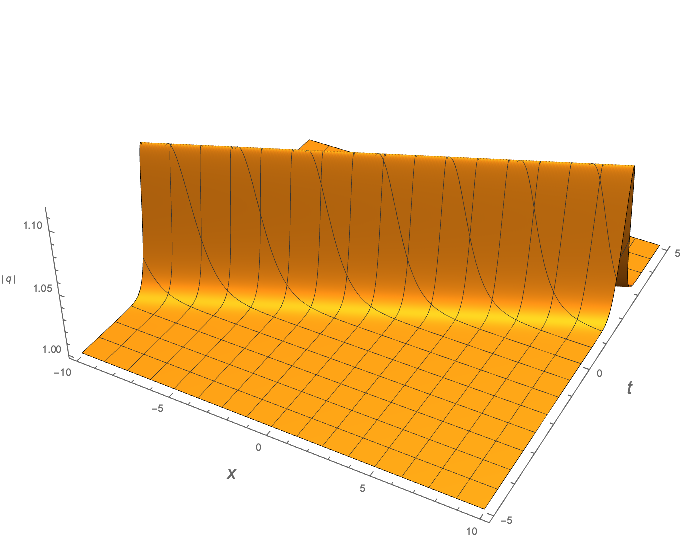}}\qquad \quad
\subfigure[]
{\includegraphics[scale=0.2]{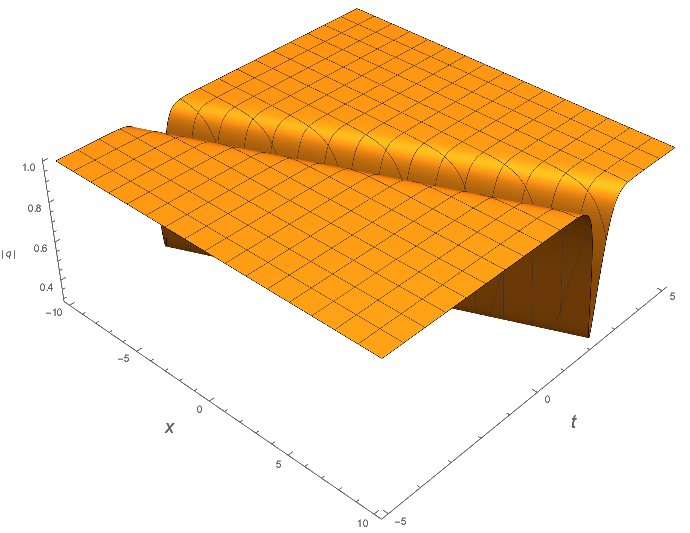}}
\caption{One-soliton solution of the defocusing KE equation with parameters $q_0=1$, $\theta_+=\frac{\pi}{4}$, $\beta=\frac {\pi} {4}$, $\alpha=2$, $\beta_1=\frac{\pi}{4}$, $T=\frac {1} {2}$, $c_1=2e^{\frac{\pi i}{3}}$(left),  $c_1=2e^{\frac{5\pi i}{3}}$(right).}
  \label{2.1df3}
\end{figure}
Equation (\ref{2.1d}) indicates
\begin{equation}\label{2.1e}
u \to \left\{ \begin{gathered}
  {e^{ i(\Upsilon-2\omega_+)}}( - {q_0}{e^{ - i{\theta _ + }}}), \qquad  \quad x \to  + \infty  \hfill \\
  {e^{ i(\Upsilon-2\omega_+)}}( - {q_0}{e^{ - i{\theta _ + }}} \cdot {e^{ - 2i{\beta _1}}}), x \to  - \infty  \hfill \\
\end{gathered}  \right.
\end{equation}
\begin{remark}
If we choose ${\tilde c_1} =  - \frac{{{\zeta _1} - {{\bar \zeta }_1}}}{{i{\zeta _1}q_ + ^{ - 1}}} $, then (\ref{2.1d}) can be written as
\begin{equation}\label{2.1f}
u = {e^{ i(\Upsilon-2\omega_+)}}\left( { - {e^{ - i({\beta _1} + {\theta _ + })}}\{ {\kappa _1} + {\lambda _1}tanh( - \frac{\Delta_1 }{2})\} } \right),
\end{equation}
where ${\omega _ + } =  - T{q_0}\sin ({\beta _1})( - 1 + \tanh ( - \frac{\Delta_1 }{2}))$.\par
Note that when we choose the parameters $\theta_+ =\pi$, $T=0$ and $\alpha=0$, the equation (\ref{2.1f}) is consistent with the result of equation (48.15) of NLS equation in \cite{N.N.H} which obtained by Zakharov-Shabat equation of inverse scattering.
\end{remark}
Further, from equation (\ref{2.1f}), obtained
\begin{equation}\label{2.1g}
|u{|^2} = q_0^2 - \lambda _1^2 sech^{2}( - \frac{\Delta_1 }{2}),
\end{equation}
this shows that the solution at this time is a dark-soliton, as is shown in fig.(\ref{2.1g-f}).
\begin{figure}[htpb]
\centering
{
\begin{minipage}[t]{0.325\linewidth}
\centering
\includegraphics[width=5.5cm]{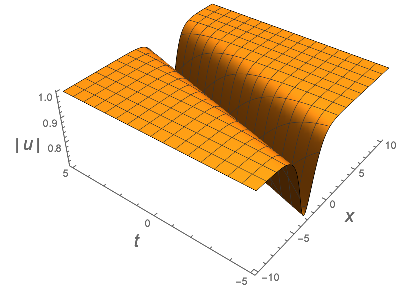}
\end{minipage}%
}%
{
\begin{minipage}[t]{0.325\linewidth}
\centering
\includegraphics[width=5.5cm]{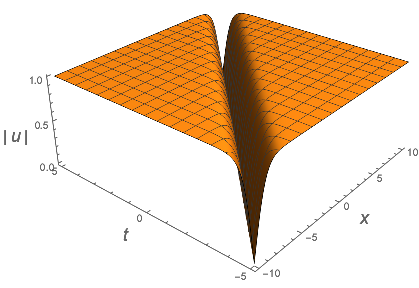}
\end{minipage}%
}%
{
\begin{minipage}[t]{0.325\linewidth}
\centering
\includegraphics[width=5.5cm]{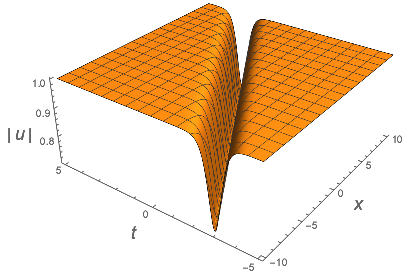}
\end{minipage}
}%

\subfigure[]
{
\begin{minipage}[t]{0.31\linewidth}
\centering
\includegraphics[width=4cm]{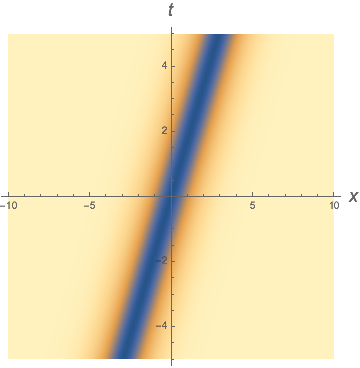}
\end{minipage}%
}
\subfigure[]
{
\begin{minipage}[t]{0.31\linewidth}
\centering
\includegraphics[width=4cm]{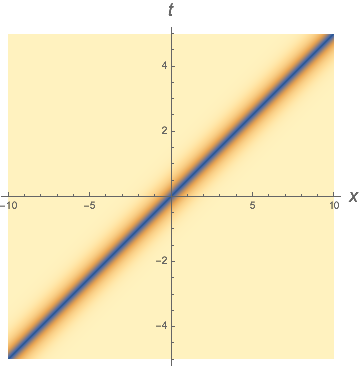}
\end{minipage}%
}
\subfigure[]
{
\begin{minipage}[t]{0.31\linewidth}
\centering
\includegraphics[width=4cm]{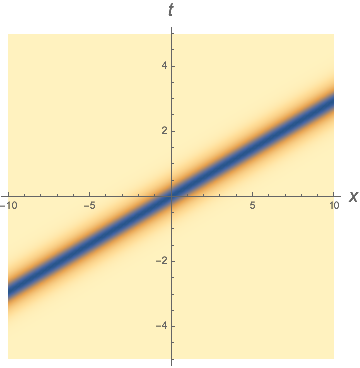}
\end{minipage}
}
\centering
\caption{Dark soliton solutions of defocusing KE equation with parameters $q_0=1$, $q_+= e^{i\pi}$,$T=\frac {1} {2}$, $\alpha=2$, $\beta_1=\frac{\pi}{4}$(middle), $\beta_1=\frac{\pi}{2}$(middle) $\beta_1=\frac{3\pi}{4}$(right). (a)$\sim$(c) display the density of the corresponding one-soliton solution}
\label{2.1g-f}
\end{figure}

\subsubsection{Other cases}
 For the case of $N_1=2$, from the equation (\ref{1.9.4a}), we get
\begin{equation}\label{2.2a}
u(x,t) = {e^{i(\Upsilon  - 2{\omega _ + })}}\left( { - {{\bar q}_ + }\frac{{M'}}{M}} \right)
\end{equation}
where
\begin{equation}\label{2.2b}
\begin{split}
M' =&B'(1) + B'(2) + B'(1,2) \\=&1 - {e^{{\Delta _1} - i({\theta _ + } - {\beta _1})}}\frac{{{{\tilde c}_1}}}{{2{q_0}\sin ({\beta _1})}} - {e^{{\Delta _2} - i({\theta _ + } - {\beta _2})}}\frac{{{{\tilde c}_2}}}{{2{q_0}\sin ({\beta _2})}} \\&+ {e^{{\Delta _1} + {\Delta _2} - i(2{\theta _ + } - {\beta _1} - {\beta _2})}}\frac{{{{\tilde c}_1}{{\tilde c}_2} \cdot {{\csc }^2}({\beta _1} + {\beta _2})}}{{4q_0^2\sin ({\beta _1})\sin ({\beta _2})}}{\sin ^2}\left( {\frac{{{\beta _1} - {\beta _2}}}{2}} \right),
\end{split}
\end{equation}

\begin{equation}\label{2.2c}
\begin{split}
M =& B(1) + B(2) + B(1,2) \\= &1 - {e^{{\Delta _1} - i({\theta _ + } + {\beta _1})}}\frac{{{{\tilde c}_1}}}{{2{q_0}\sin ({\beta _1})}} - {e^{{\Delta _2} - i({\theta _ + } + {\beta _2})}}\frac{{{{\tilde c}_2}}}{{2{q_0}\sin ({\beta _2})}} \\&+ {e^{{\Delta _1} + {\Delta _2} - i(2{\theta _ + } + {\beta _1} + {\beta _2})}}\frac{{{{\tilde c}_1}{{\tilde c}_2} \cdot {{\csc }^2}({\beta _1} + {\beta _2})}}{{4q_0^2\sin ({\beta _1})\sin ({\beta _2})}}{\sin ^2}\left( {\frac{{{\beta _1} - {\beta _2}}}{2}} \right)
\end{split}
\end{equation}

\begin{figure}[htpb]
\centering
\subfigure[]
{
\begin{minipage}[t]{0.325\linewidth}
\centering
\includegraphics[width=5cm]{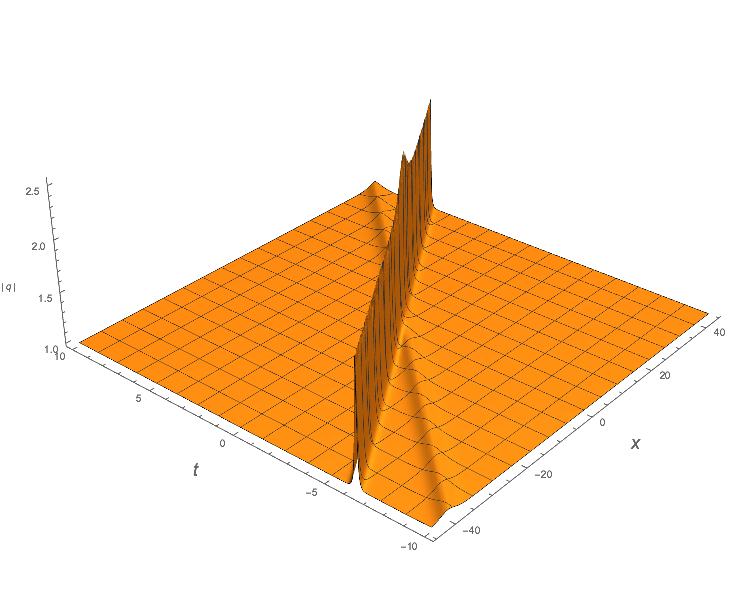}
\end{minipage}%
}%
\subfigure[]
{
\begin{minipage}[t]{0.325\linewidth}
\centering
\includegraphics[width=5cm]{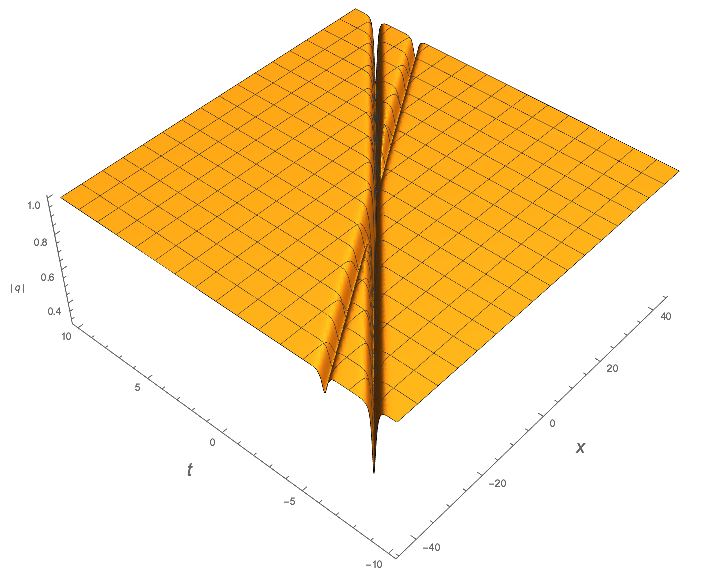}
\end{minipage}%
}%
\subfigure[]
{
\begin{minipage}[t]{0.325\linewidth}
\centering
\includegraphics[width=5cm]{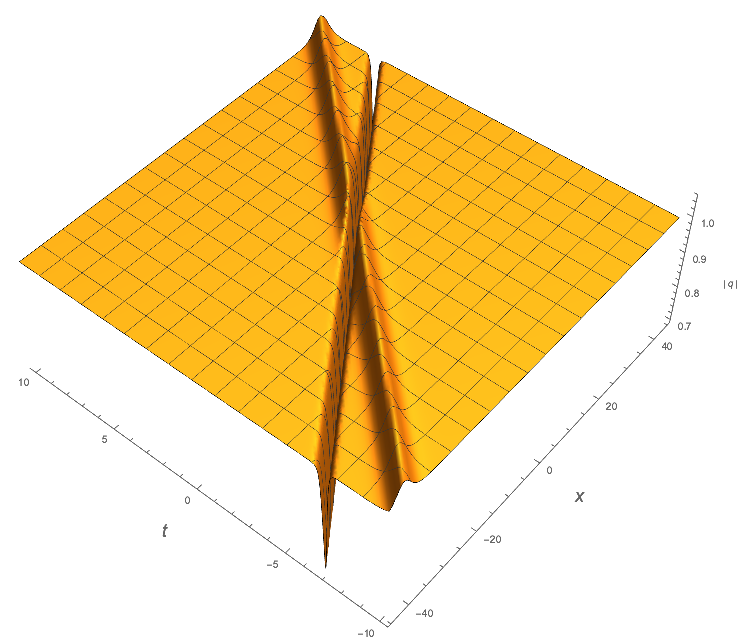}
\end{minipage}
}%

\centering
\caption{Two-soliton solutions of defocusing KE equation with NZBCs $q_0=1$ with parameters $\theta=\frac {\pi} {3}$, $\alpha=2$, $T=\frac{1}{2}$: $(a)$2-bright soliton solutions with parameters $c_1=3e^{\frac{3\pi i}{4}}$, $c_2=e^{\frac{4\pi i}{5}}$, $\beta_1=\frac{\pi}{8}$, $\beta_2=\frac{2\pi}{3}$; $(b)$2-dark soliton solutions with parameters $c_1=2e^{-\frac{\pi i}{4}}$, $c_2=3e^{-\frac{\pi i}{4}}$, $\beta_1=\frac{\pi}{3}$, $\beta_2=\frac{3\pi}{4}$; $(c)$bright-dark soliton solutions with parameters $c_1=2e^{\frac{3\pi i}{4}}$, $c_2=e^{-\frac{\pi}{5}}$, $\beta_1=\frac{\pi}{8}$, $\beta_2=\frac{5\pi}{9}$.}
\label{2.2a-f}
\end{figure}

For the case of $N_1=3$, the dynamic structure is shown in fig.(\ref{2.3f}).
\begin{figure}[htpb]
\centering
\subfigure[]
{
\begin{minipage}[t]{0.325\linewidth}
\centering
\includegraphics[width=4.5cm]{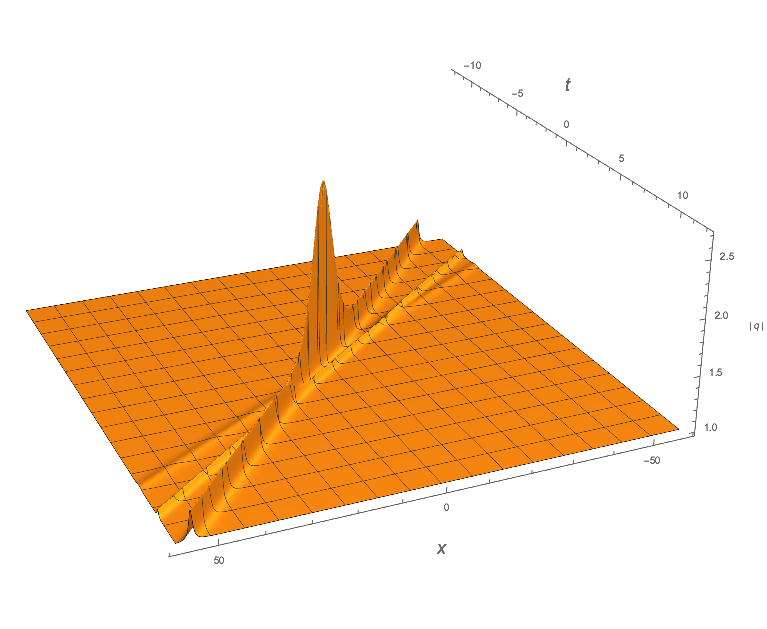}
\end{minipage}%
}%
\subfigure[]
{
\begin{minipage}[t]{0.325\linewidth}
\centering
\includegraphics[width=4.5cm]{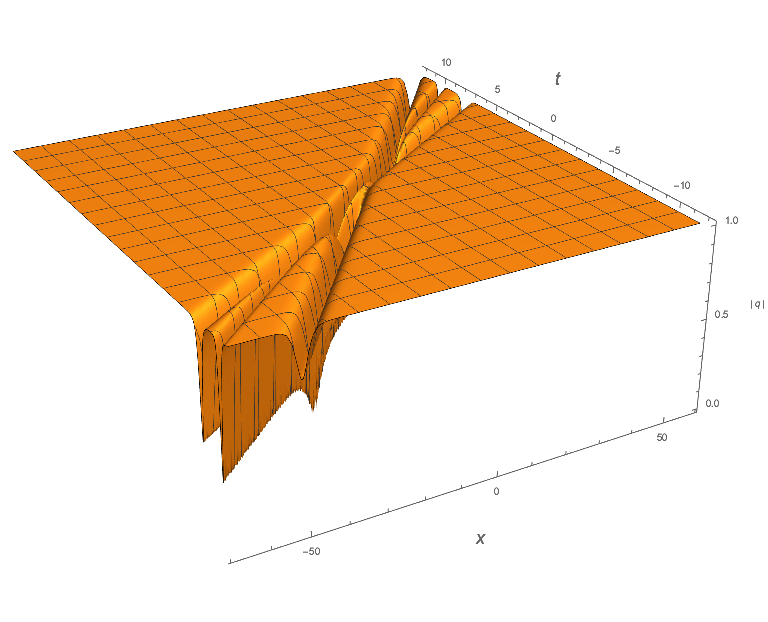}
\end{minipage}%
}%
\subfigure[]
{
\begin{minipage}[t]{0.325\linewidth}
\centering
\includegraphics[width=4.5cm]{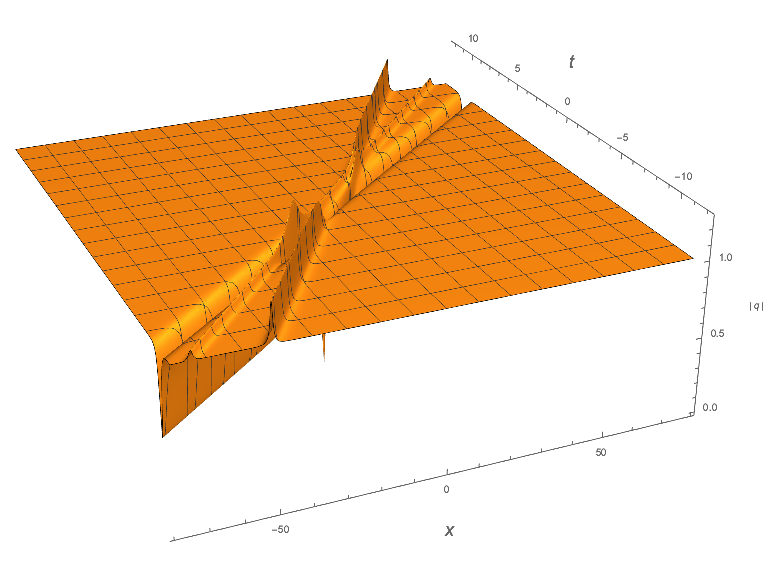}
\end{minipage}
}%
\centering
\caption{Three-soliton solutions of defocusing KE equation with NZBCs $q_0=1$ with parameters $\theta=\frac {\pi} {3}$, $\alpha=2$, $T=\frac{1}{2}$: $(a)$3-bright soliton solutions with parameters $c_1=1-2i$, $c_2=1-2i$, $c_3=1-2i$, $\beta_1=\frac{\pi}{6}$, $\beta_2=\frac{7\pi}{8}$, $\beta_3=\frac{\pi}{6}$; $(b)$2-dark soliton solutions with parameters $c_1=1+2i$, $c_2=1+2i$, $c_3=1+2i$, $\beta_1=\frac{\pi}{6}$, $\beta_2=\frac{\pi}{2}$, $\beta_3=\frac{6\pi}{7}$; $(c)$bright-bright-dark soliton solutions with parameters $c_1=1-2i$, $c_2=1-2i$, $c_3=1-2i$, $\beta_2=\frac{7\pi}{8}$, $\beta_3=\frac{\pi}{2}$.}

\label{2.3f}
\end{figure}

\section{The Focusing KE equation with NZBCs case with simple poles}
We concern the non-zero boundary value problem for the focusing Kundu-Eckhaus Equations($\kappa$=1):
\begin{equation}\label{1.a}
iq_t+q_{xx}+2 \left| q \right| ^2q+4 T^2\left|q\right|^4q-4i T ( \left| q \right| ^2)_xq=0,
\end{equation}
\begin{equation}\label{1.b}
\lim\limits_ { x \rightarrow \pm \infty }q(x,0)=q_{\pm}= {q_0}{e^{i(\alpha x + {\theta _ \pm })}},
\end{equation}
with $|{q_ \pm }| = q_0 > 0$, $\alpha$, $\theta _ \pm$ are real parameters,  and suppose $q(x,t) - {q_ \pm } \in {L^1}({R^ + })$.

The focusing KE equation admits the following Lax pairs:
\begin{equation}\label{1.c}
{\Phi _x} + iz{\sigma _3}\Phi = (Q - iT{Q^2}{\sigma _3})\Phi ,
\end{equation}
\begin{equation}\label{1.d}
{\Phi _t} + 2i{z^2}{\sigma _3}\Phi  = (W + 4i{T^2}{Q^4}{\sigma _3} - T(Q{Q_x} - {Q_x}Q))\Phi ,
\end{equation}
where $T$ is a constant, $z$ is spectral parameter and
\begin{equation}
Q = \left( {\begin{array}{*{20}{c}}
0&q\\
{ - \bar q}&0
\end{array}} \right) ,
\end{equation}
\begin{equation}
W = 2zQ - 2T{Q^3} - i({Q^2} + {Q_x}){\sigma _3} ,
\end{equation}

\subsection{Direct scattering problem with NZBC}
In order to make the Lax pairs of the time part and the space part compatible, the following gauge transformation is introduced,
\begin{equation}
\Psi (x,t) = \left( {\begin{array}{*{20}{c}}
{{e^{ - i(\alpha x+(2q_0^2 + 4{T^2}q_0^4-\alpha ^2)t)/2}}}&0\\
0&{{e^{i(\alpha x +(2q_0^2 + 4{T^2}q_0^4-\alpha ^2)t)/2}}}
\end{array}} \right)\Phi (x,t) ,
\end{equation}
the new Lax pairs as follows:
\begin{equation}\label{2.c}
{\Psi _x} + i(z+\frac{\alpha}{2}){\sigma _3}\Psi  = U\Psi  ,
\end{equation}
\begin{equation}\label{2.d}
{\Psi _t} + 2i({z^2} + \frac{1}{4}(2q_0^2 + 4{T^2}q_0^4-\alpha ^2)){\sigma _3}\Psi  = V\Psi  ,
\end{equation}
where
\begin{equation}\label{2.e}
U = \left( {\begin{array}{*{20}{c}}
{iT|q{|^2}}&{q{e^{ -i(\alpha x+(2q_0^2 +4 {T^2}q_0^4-\alpha ^2)t)}}}\\
{ - \bar q{e^{i(\alpha x+(2q_0^2 + 4{T^2}q_0^4-\alpha ^2)t)}}}&{ - iT|q{|^2}}
\end{array}} \right) ,
\end{equation}
\begin{equation}\label{2.f}
V = \left( {\begin{array}{*{20}{c}}
{i|q{|^2} + 4i{T^2}|q{|^4} + T(q  {{\bar q}_x} - {q_x}\bar q)}&{(i{q_x} + 2zq + 2T|q{|^2}q){e^{ -i(\alpha x+(2q_0^2 + {4T^2}q_0^4-\alpha ^2)t)}}}\\
{(i{{\bar q}_x} - 2z\bar q - 2T|q{|^2}\bar q){e^{ i(\alpha x+(2q_0^2 +4 {T^2}q_0^4-\alpha ^2)t)}}}&{ - i|q{|^2} - 4i{T^2}|q{|^4} - T(q  {{\bar q}_x} - {q_x}\bar q)}
\end{array}} \right) .
\end{equation} \par
One can verify that the Lax pairs satisfied the compatibility condition
\begin{equation}\label{1.1g}
{U_t} - {V_x} + [U,V] = 0,
\end{equation}
where $[U,V]=UV-VU$. With the non-zero boundary conditions and $x$ tend to infinity, the asymptotic behavior of the Lax equations (\ref{1.c})(\ref{1.d}) with the boundary value problem ($\mathop {\lim }\limits_{x \to  \pm \infty } q(x,0) = {q_ \pm }$) can be written as:

\begin{equation}\label{2.a}
{\Psi _{ \pm x}} = {U_ \pm }{\Psi _ \pm },
\end{equation}
\begin{equation}\label{2.b}
{\Psi _{ \pm t}} = {V_ \pm }{\Psi _ \pm },
\end{equation}
where
\begin{equation}
{U_ \pm } = \left( {\begin{array}{*{20}{c}}
  {iTq_0^2 - iz - \frac{{i\alpha }}{2}}&{{q_0}{e^{i{\theta _ \pm }}}} \\
  {{-q_0}{e^{ - i{\theta _ \pm }}}}&{ - iTq_0^2 + iz + \frac{{i\alpha }}{2}}
\end{array}} \right) ,
\end{equation}
\begin{equation}
{V_ \pm } = \left( {\begin{array}{*{20}{c}}
  {2i{T^2}q_0^4 -2 i\alpha Tq_0^2 - 2i{z^2} + \frac{{i{\alpha ^2}}}{2}}&{{q_0} (2Tq_0^2 + 2z - \alpha ){e^{i{\theta _ \pm }}}} \\
  {{-q_0}(  2Tq_0^2 + 2z - \alpha ){e^{ - i{\theta _ \pm }}}}&{ - 2i{T^2}q_0^4 +2 i \alpha Tq_0^2 + 2i{z^2} - \frac{{i{\alpha ^2}}}{2}}
\end{array}} \right) ,
\end{equation}
then, we have ${V_ \pm } = (2z + 2Tq_0^2-\alpha){U_ \pm }$, which allows the Lax pairs to correspond to the same eigenvectors and Jost solutions.
\subsection{Riemann surface and uniformization variable }
\begin{figure}[htpb]
\centering
\subfigure[The first sheet of Riemann $\tilde z-plane$]{
\begin{minipage}{6cm}
\centering
\includegraphics[width=6.1 cm]{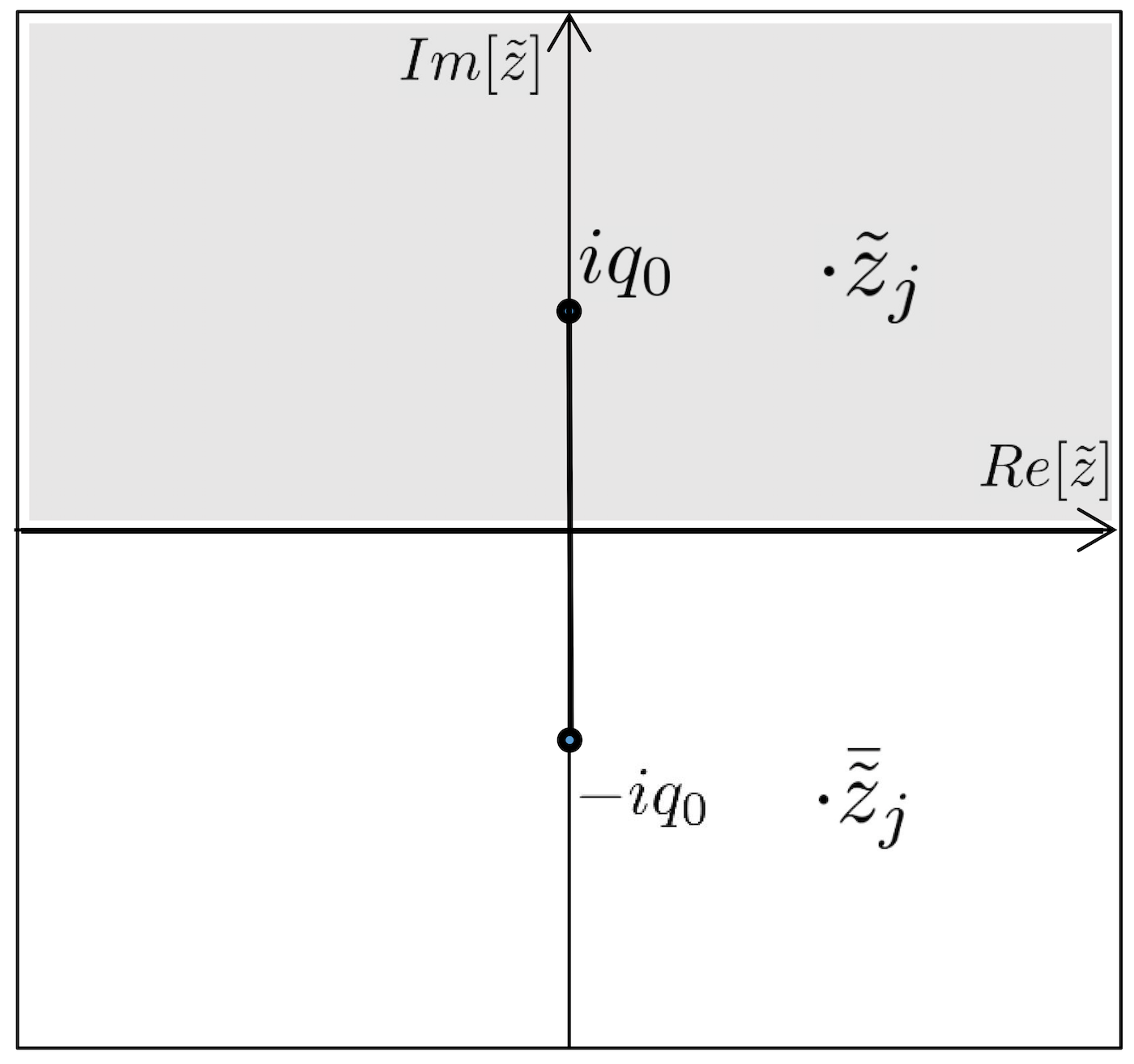}
\end{minipage}
} \qquad
\subfigure[The complex  $\zeta-plane$]{
\begin{minipage}{6cm}
\centering
\includegraphics[width=6 cm]{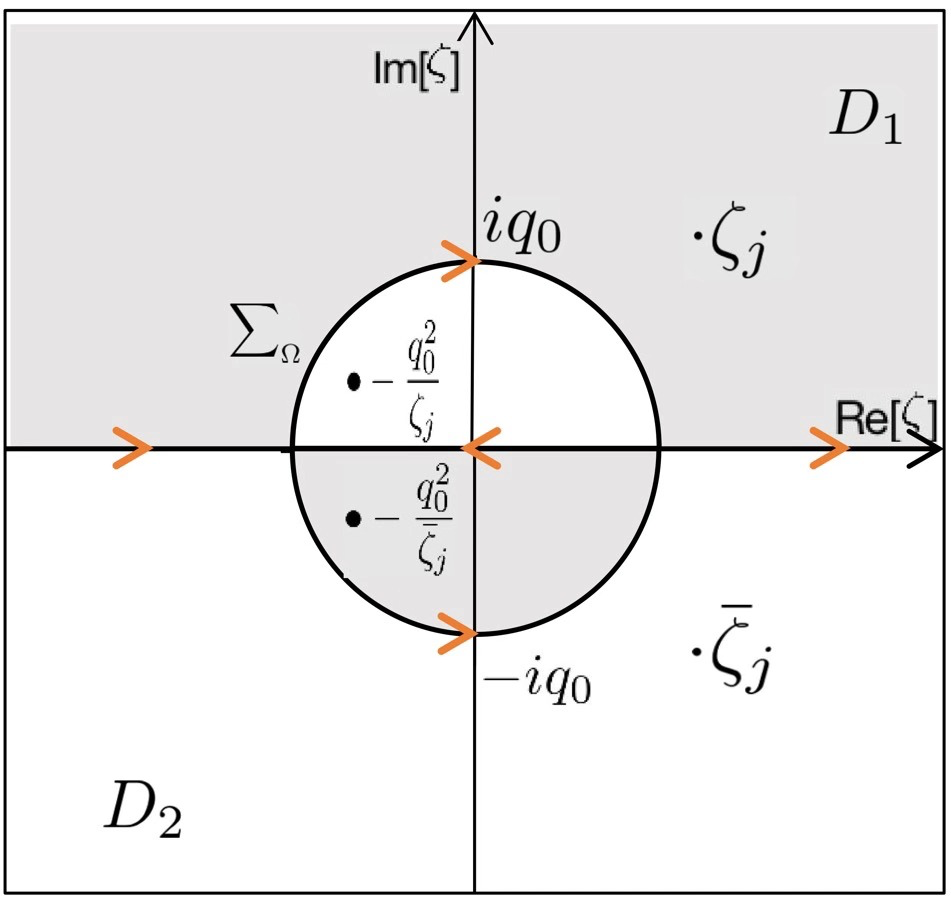}
\end{minipage}
}
\caption{(a) Showing the first sheet of Riemann $\tilde z-plane$ region where ${\mathop{\rm Im}\nolimits} [\tilde z] > 0$(gray) and ${\mathop{\rm Im}\nolimits} [\tilde z] < 0$(white).Also shown the discrete spectrums. (b) Showing the regions $D_1$(gray) and $D_2$(white). Also shown the symmetries of the discrete spectrum of the scattering problem and the oriented contours for the Riemann-Hilbert problem(orange).}
\label{fig.1}
\end{figure}

In the case of continuum and the spectral parameter $z$ is a real value. We can get the Lax pair(\ref{2.c}) eigenvalues corresponding to $ \pm i\sqrt {{{\tilde z}^2} + q_0^2} $ while $x \to  \pm \infty $  and $q(x,0)=q_\pm$. The simultaneous solution of the Lax equations (\ref{2.c})  and (\ref{2.d}) which so called Jost solution could be obtained as follows:
\begin{equation}\label{2.g}
{\psi _ \pm }(x,t,z) = \left( {\begin{array}{*{20}{c}}
1&{i\bar q_ \pm ^{ - 1}(\tilde z - \lambda )}\\
{iq_ \pm ^{ - 1}(\tilde z - \lambda )}&1
\end{array}} \right){e^{ - i\lambda (x + 2(z + Tq_0^2-\frac{\alpha}{2})t){\sigma _3}}}, \quad z \in \mathbb{R} ,
\end{equation}
where  $\tilde z = z - Tq_0^2+ \frac{\alpha}{2}, \lambda  = \sqrt {{{\tilde z}^2} + q_0^2} $. Thanks to the $\lambda$ is $\tilde z$ two-valued function, we need to introduce Riemann surface to handle it, but this process are complicated. To overcome the multi-value of the square root and be able to known that the genus of Riemann surface is $0$, it is convenient to use the uniformization variable $\zeta$ as\cite{L.D.F}\cite{N.N.H}:
\begin{equation}\label{2.i}
\lambda  = \frac{1}{2}(\zeta  + q_0^2{\zeta ^{ - 1}}), \qquad  \tilde z = \frac{1}{2}(\zeta  - q_0^2{\zeta ^{ - 1}}).
\end{equation} \par
The eigenvalues $\pm i\sqrt {\tilde z + Tq_0^2 }$ of (\ref{2.g}) have two branch points at $\tilde z = i{q_0}$ and $\tilde z = -i{q_0}$. In the sense of square root, which one needs introduce a two-sheet Riemann surface to obtained single-valued $\lambda (z)$ on each sheet. In addition, the branch cut interval is $[ - i{q_0},i{q_0}]$. Along this branch cut and the real axis of the $\tilde z-plane$, $\lambda (\tilde z)$ is a real-valued function, and $\lambda (\tilde z)$ is opposite to each other on the first sheet and second sheet of the Riemann surfaces. More precisely, on the first sheet Riemann surface, where the ${\mathop{\rm Im}\nolimits} [\tilde z] > 0$ region and the ${\mathop{\rm Im}\nolimits} [\tilde z] < 0$ region are located in the upper-half plane and in the lower-half plane, respectively. The  ${\mathop{\rm Im}\nolimits} [\tilde z] > 0$ region and the  ${\mathop{\rm Im}\nolimits} [\tilde z] < 0$ region are respectively located in the lower-half plane and the upper-half plane of the second sheet Riemann surface. Through the uniformization variable, we have established a mapping relationship between the Riemann $\tilde z-plane$ and the $\zeta-plane$ (as shown in Figure(\ref{fig.1})).In the two sheets Riemann $\tilde z-plane$, the branch cut was mapped to the $\zeta-plane$ ,which formed a circle with radius $q_0$ (marked as $'$$q_0$-circle$'$), the first sheet Riemann $\tilde z-plane$ was mapped to the outside of the $q_0$-circle, and the second sheet Riemann $\tilde z-plane$ was mapped to the inside of the $q_0$-circle and denoted $\mathop{\Sigma}\nolimits_\Omega  = \mathbb{R} \cup (q_0$-circle$)\backslash \{  \pm i{q_0}\} $.

\begin{remark}
Up to now, we have discussed the single-valued property of $\tilde z$. It can be seen from the definition of $\tilde z$ that the single-valued property of $\tilde z$ is equivalent to the single-valued property of the spectral parameter $z$.
\end{remark}

After, we will consider the complex $\zeta $-plane instead of two-sheet Riemann surface. Jost solution (\ref{2.g}) could be obtained as follows:
\begin{equation}\label{2.j}
{\psi _ \pm }(x,t,\zeta ) = \left( {\begin{array}{*{20}{c}}
1&{ - i{q_ \pm }{\zeta ^{ - 1}}}\\
{ - i{{\bar q}_ \pm }{\zeta ^{ - 1}}}&1
\end{array}} \right){e^{ - i\lambda (x + 2(z + Tq_0^2-\frac{\alpha}{2})t){\mathop{\sigma_3} }}}, \qquad  \zeta  \in \mathop{\Sigma}\nolimits_\Omega  ,
\end{equation}

Define ${\tilde \psi _ \pm }(x,t,\zeta )$ as the simultaneous solution of lax pairs (\ref{2.c}) and (\ref{2.d}).Then ${\tilde \psi _ \pm }(x,t,\zeta )$ satisfy the asymptotic conditions
\begin{equation}\label{2.k}
{\tilde \psi _ \pm }(x,t,\zeta ) = {\psi _ \pm }(x,t,\zeta ) + o(1) ,
\end{equation}
denoting the region ${D_1} = \{ \zeta  \in C|(|\zeta {|^2} - q_0^2) \cdot {\mathop{\rm Im}\nolimits} \zeta  > 0\} $ and ${D_2} = \{ \zeta  \in C|(|\zeta {|^2} - q_0^2) \cdot {\mathop{\rm Im}\nolimits} \zeta  < 0\} $. Denote this symbol $\Theta (x,t,\zeta ) = \lambda (x + 2(z + Tq_0^2-\frac{\alpha}{2})t)$.

 ~~~
\subsection{The Jost solutions}

From Jost solution (\ref{2.k}) we defined the modified Jost solution $ \mu \mathop{{}}\nolimits_{{ \pm }} \left( x,t, \zeta  \right)  $, as follows:
\begin{equation}\label{2.l}
{\mu _\pm}(x,t,\zeta ) = {\tilde \psi _ \pm }{e^{i\Theta (x,t,\zeta ){\mathop{\sigma_3} }}}, \qquad \zeta  \in \mathop{\Sigma}\nolimits_\Omega  ,
\end{equation}
such that
\begin{equation}\label{2.m}
\mathop {\lim }\limits_{x \to  \pm \infty } {\mu _ \pm }(x,t,\zeta ) = {\tilde E_ \pm }(\zeta )+ o(1) ,  \qquad  \zeta  \in \mathop{\Sigma}\nolimits_\Omega  ,
\end{equation}
where
\begin{equation}\label{2.n}
{\tilde E_ \pm }(\zeta ) = \left( {\begin{array}{*{20}{c}}
1&{ - i{q_ \pm }{\zeta ^{ - 1}}}\\
{ - i{{\bar q}_ \pm }{\zeta ^{ - 1}}}&1
\end{array}} \right) .
\end{equation}
By the constant variation approach, we obtained two Jost Integral equation as follows:
\begin{equation}\label{2.o}
\mu_+(x,t,\zeta ) = {\tilde E_ + } - \int_x^\infty  {{\tilde E}_ + }{e^{ - i\lambda (x - x'){{\hat\sigma_3}}}} ({\tilde E_ + }^{ - 1}(U - {U_ + }){\mu _ + }(x',t,\zeta ))dx' ,
\end{equation}
\begin{equation}\label{2.p}
{\mu _ - }(x,t,\zeta ) = {\tilde E_ - } + \int_{ - \infty }^x {{\tilde E}_ - }{e^{ - i\lambda (x - x'){{\hat\sigma_3}}}} ({\tilde E_ - }^{ - 1}(U - {U_ - }){\mu _ - }(x',t,\zeta ))dx' ,
\end{equation}
where   $e^{{{\hat\sigma_3}}}(\Lambda ) = {e^{{\sigma_3}}}(\Lambda ){e^{ - {\sigma_3}}}$.
\begin{proposition}\label{2.q}
(Analytic property). then ${\mu _ \pm }(x,t,\zeta )$ satisfy the following bounded and analytic properties,
\begin{equation}\label{2.r}
{\mu _ + } = [{\mu _ + }^{(2)},{\mu _ + }^{(1)}], \qquad {\mu _ - } = [{\mu _ - }^{(1)},{\mu _ - }^{(2)}] ,
\end{equation}
where the superscripts $^{(1),(2)}$ indicate that each column of ${\mu _ \pm }(x,t,\zeta )$ are bounded and analytic in the $D_1$ and $D_2$ regions, respectively.
\end{proposition}

To see this, We represent ${\mu _ \pm }(x,t,\zeta )$ as a column ${\mu _ \pm } = [{\mu _{ \pm 1}},{\mu _{ \pm 2}}]$. Since for any fixed $t$, we know that ${\mu _{ + 1}}$ containing exponential factor ${e^{2i\lambda (x - x')}}$. When ${\mathop{\rm Im}\nolimits} [\lambda ] < 0$, the exponent of the integrand is not only bounded but decayed thanks to $x-x'<0$, which corresponding to the region of $D_1$ in the $\zeta$-plane. Then, the first column of ${\mu _ \pm }(x,t,\zeta )$ is bounded and analytic in $D_1$. By the same method, the bounded and analytic properties of other columns of ${\mu _ \pm }(x,t,\zeta )$ can be obtained, respectively.

\subsubsection{Scattering matrix}
The eigenfunctions ${\tilde \psi _ + }(x,t,\zeta )$ and ${\tilde \psi _ - }(x,t,\zeta )$ are not independent, satisfy the relation for some matrix $s(\zeta )$ independent of $x$ and $t$.
\begin{equation}\label{2.s}
{\tilde \psi _ - }(x,t,\zeta ) = {\tilde \psi _ + }(x,t,\zeta )s(\zeta ) , \qquad  \zeta  \in \mathop{\Sigma}\nolimits_\Omega  ,
\end{equation}
where
\begin{equation}\label{2.w}
s(\zeta ) = \left( {\begin{array}{*{20}{c}}
{a(\zeta )}&{\tilde b(\zeta )}\\
{b(\zeta )}&{\tilde a(\zeta )}
\end{array}} \right) ,
\end{equation}
\begin{proposition}\label{2.t}
(The first symmetry property). The modified Jost functions ${\mu _ \pm }(x,t,\zeta )$ have the following symmetry conditions:
\begin{equation}\label{2.u}
  {\mathop{\sigma_2} }\overline{{\mu _ \pm }(x,t,\bar \zeta )}{\mathop{\sigma_2} } = {\mu _ \pm }(x,t,\zeta ) .
 \end{equation}
\end{proposition}
\begin{proof}
It is easy to see that the lax pairs satisfies the symmetry of ${\mathop{\sigma_2} }$ and the same as the asymptotic conditions, that is ${\mathop{\sigma_2} }\overline{(U - i\bar z{\mathop{\sigma_3} } )} {\mathop{\sigma_2} } = U - iz{\mathop{\sigma_3} }$, ${\mathop{\sigma_2} }\overline{({U_ \pm } - i\bar z{\mathop{\sigma_3} })}{\mathop{\sigma_2} } = {U_ \pm } - iz{\mathop{\sigma_3} }$, etc. We have ${\mathop{\sigma_2} } \overline {{\tilde \psi _ \pm }(x,t,\bar \zeta )}{\mathop{\sigma_2} } = {\tilde \psi _ \pm }(x,t,\zeta )$. From the definition (\ref{2.l}) of the modified Jost functions ${\mu _ \pm }(x,t,\zeta )$, a direct computation shows (\ref{2.t}) holds.
\end{proof}

From the above proof and (\ref{2.s}), the symmetry of the scattering data could be obtained as follows:
\begin{equation}\label{2.v}
{\mathop{\sigma_2} }\overline{s(\bar \zeta )}{\mathop{\sigma_2} } = s(\zeta ),  \qquad  \zeta  \in \mathop{\Sigma}\nolimits_\Omega  .
\end{equation}
then there is $\tilde a(\zeta ) = \overline{ a(\bar \zeta )}$,  $\tilde b(\zeta ) =  -\overline{ b(\bar \zeta )}$,  $|a(\zeta ){|^2} + |b(\zeta ){|^2} = 1$ and $\det (s(\zeta )) = 1 $. \par
By the analytic property (\ref{2.p}), we know that $a(\zeta )$ is analytic in $D_1$, $\tilde a(\zeta )$ is analytic in $D_2$. But $b(\zeta)$ and $\tilde b(
\zeta)$ are only defined on $\mathop{\Sigma}\nolimits_\Omega$.

We introduce the following transformation,
\begin{equation}\label{2.a.d}
\zeta   \mapsto  -q_0^2{\zeta ^{ - 1}} ,
\end{equation}
it leads to $\tilde z(\zeta ) \to \tilde z( - q_0^2{\zeta ^{ - 1}})$ and $\lambda (\zeta ) \to -\lambda ( - q_0^2{\zeta ^{ - 1}})$. For convenience, letting $ - q_0^2{\zeta ^{ - 1}} \mapsto \hat \zeta $, $ - q_0^2{\bar \zeta ^{ - 1}} \mapsto \check \zeta $.
We can get the following symmetry:
\begin{proposition}\label{2.a.e}
(The second symmetry property). The modified Jost functions ${\mu _ \pm }(x,t,\zeta )$ have the following symmetry conditions:
\begin{equation}\label{2.a.f}
{\mu _ \pm }(x,t,\zeta ) =  - i{\zeta ^{ - 1}}{\mu _ \pm }(x,t,\hat \zeta ){\mathop{\sigma_3} }{Q_ \pm } .
 \end{equation}
\end{proposition}
\begin{proof}
Since the lax pairs (\ref{2.c}) (\ref{2.d}) remains unchanged under this transformation, from the asymptotic conditions (\ref{2.m}) (\ref{2.n}), a direct computation shows the proposition (\ref{2.a.e}) holds.
\end{proof}
Then, according to the above symmetry relationship, we know that
\begin{equation}\label{2.a.i}
s(\zeta ) = {({\mathop{\sigma_3} }{Q_ + })^{ - 1}}s(\hat \zeta ){\mathop{\sigma_3} }{Q_ - }, \qquad  \zeta  \in \mathop{\Sigma}\nolimits_\Omega .
 \end{equation}

Define the reflection coefficient as $r(\zeta ) = \frac{{b(\zeta )}}{{a(\zeta )}}$, then its symmetry can also be obtained
 \begin{equation}\label{2.a.n}
r(\zeta ) =  - \overline{ \tilde r(\bar \zeta )} = \frac{{{{\bar q}_ + }}}{{{ q_ + }}}\tilde r(\hat \zeta ) =  - \frac{{{ {\bar q}_ + }}}{{{ q_ + }}}\overline{r(\check \zeta )} ,  \qquad  \zeta  \in \mathop{\Sigma}\nolimits_\Omega .
 \end{equation} \par
According to the Schwarz reflection principle, we can determine the analytical extension regions of the above scattering data on the complex $\zeta$-plane.

\subsubsection{Asymptotic behavior}
With the Wentzel-Kramers-Brillouin(WKB) expansion, the asymptotic formula of ${\mu _ \pm }(x,t,\zeta )$ as $\zeta  \to \infty $ and $\zeta  \to 0 $ can be derived by substituting the following expansion
\begin{equation}\label{2.a.o}
 {\mu _ \pm }(x,t,\zeta ) = \sum\limits_{j = 0}^n {\frac{{\mu _ \pm ^{(j)}(x,t)}}{{{\zeta ^j}}} + O(\frac{1}{{{\zeta ^{n + 1}}}})} ,\qquad \zeta  \to \infty ,
\end{equation}
\begin{equation}\label{2.a.p}
{\mu _ \pm }(x,t,\zeta ) = \sum\limits_{j =  -1 }^n {{\zeta ^j} \cdot \mu _ \pm ^{(j)}(x,t) + O({\zeta ^{n + 1}})}  ,\qquad \zeta  \to 0 ,
\end{equation}
into the Lax pairs (\ref{2.c}), (\ref{2.d}), and comparing the order of $\zeta$.

\begin{proposition}\label{2.a.q}
(The asymptotic behavior). The modified Jost function has the following two kinds of asymptotic behaviors
\begin{equation}\label{2.a.r}
{\mu _ \pm }(x,t,\zeta ) = {e^{iT\int_{ \pm \infty }^x {(|q{|^2} - q_0^2)} dy{\mathop{\sigma_3} }}}\left( {\mathbb{E} + \frac{{\mu _ \pm ^{(1)}}(x,t)}{\zeta } + O(\frac{1}{{{\zeta ^2}}})} \right),  \qquad \zeta  \to \infty ,
\end{equation}
where
$$ \mu _ \pm ^{(1)}(x,t) = \left( {\begin{array}{*{20}{c}}
*&{ - iq{e^{ -i(\alpha x+(2q_0^2 + {4T^2}q_0^4-\alpha ^2)t)}}{e^{ - 2iT\int_{ \pm \infty }^x {(|q{|^2} - q_0^2)dy} }}}\\
{ - i\bar q{e^{i(\alpha x+(2q_0^2 + {4T^2}q_0^4-\alpha ^2)t)}}{e^{2iT\int_{ \pm \infty }^x {(|q{|^2} - q_0^2)dy} }}}&*
\end{array}} \right),$$
and
\begin{equation}\label{2.a.s}
{\mu _ \pm }(x,t,\zeta ) = {e^{iT\int_{ \pm \infty }^x {(|q{|^2} - q_0^2)} dy}}( - i{\zeta ^{ - 1}}{\mathop{\sigma_3} }{Q_ \pm }) + O(1), \qquad \zeta \to 0 .
\end{equation}
\end{proposition}
From the Wronskian presentations of the scattering data and proposition (\ref{2.a.q}), we can also get the asymptotic behavior of the scattering matrix as follows
\begin{proposition}\label{2.a.t}
The asymptotic behavior of the scattering matrix can be expressed as follows
\begin{equation}\label{2.a.u}
s(\zeta ) = {e^{iT\int_{ - \infty }^{ + \infty } {(|q{|^2} - q_0^2)dy{\mathop{\sigma_3} }} }}\mathbb{E} + O({\zeta ^{ - 1}}),  \qquad \zeta \to \infty ,
\end{equation}
\begin{equation}\label{2.a.v}
s(\zeta ) = diag(\frac{{{q_ + }}}{{{q_ - }}},\frac{{{q_ - }}}{{{q_ + }}}){e^{-iT\int_{ - \infty }^{ + \infty } {(|q{|^2} - q_0^2)dy{\mathop{\sigma_3} }} }} + O(\zeta),  \qquad \zeta \to 0 .
\end{equation}
\end{proposition}
Hence, one can derived the potential function of the KE equation with NZBCs from equation (\ref{2.a.r}), we have
\begin{equation}\label{2.a.w}
q(x,t) ={e^{i\Upsilon}} \mathop {\lim }\limits_{\zeta  \to \infty } (i\zeta {\mu _{ \pm 12}}(x,t)){e^{iT\int_{ \pm \infty }^x {(|q{|^2} - q_0^2)} dy}} ,
\end{equation}
where ${\Upsilon=\alpha x+(2q_0^2 + {4T^2}q_0^4-\alpha ^2)t}$.

\subsection{Residue conditions with simple poles}
In general, the discrete spectrum of the scattering problem is a set of all values of $\zeta  \in \mathbb{C}\backslash {\mathop{\Sigma} _\Omega }$, so the eigenfunctions existed in ${L^2}(\mathbb{R})$\cite{V.E.Z,W.X.M2}. Suppose the scattering data $a({\zeta _j}) = 0$, $\dot a({\zeta _j}) \ne 0$ and $$\{ {{\zeta _j} \in C\left| {|{\zeta _j}|} \right. > {q_0}, {\mathop{\rm Im}\nolimits} {\zeta _j} > 0,  j = 1,2, \ldots ,{N_1}} \},$$ here and after $\dot{\Lambda } $ is defined as the derivative of $\zeta $.
From the symmetry conditions (\ref{2.v}) and (\ref{2.a.i}), we know that
\begin{equation}
a({\zeta _j}) = \tilde a({\bar \zeta _j}) = \tilde a({\hat \zeta _j}) =  a({\check \zeta _j}) = 0 .
\end{equation} \par
Thus, denoting the discrete spectrum set of the scattering problem is $$ U = \left\{ {{\zeta _j},{{\bar \zeta }_j},{{\hat \zeta }_j},{\check \zeta _j}} \right\}_{j = 1}^{{N_1}} .$$ \par
We have the following residue condition
\begin{equation}\label{3.c}
{\mathop{\rm Re}\nolimits} {s_{\zeta  = {\zeta _j}}}\left[ {\mu _ - ^{(1)}(x,t,\zeta )/a(\zeta )} \right] = {c_j}{e^{2i\Theta ({\zeta _j})}}\mu _ + ^{(1)}(x,t,{\zeta _j}) ,
\end{equation}
where ${c_j} = \frac{{{b_j}}}{{\dot a({\zeta _j})}}$.
\begin{equation}\label{3.d}
{\mathop{\rm Re}\nolimits} {s_{\zeta  = {{\bar \zeta }_j}}}\left[ {\mu _ - ^{(2)}(x,t,\zeta )/\tilde a(\zeta )} \right] = {\tilde c_j}{e^{ - 2i\Theta ({{\bar \zeta }_j})}}\mu _ + ^{(2)}(x,t,{\bar \zeta _j}) ,
\end{equation}
where ${\tilde c_j} = \frac{{{{\tilde b}_j}}}{{\dot {\tilde a}({{\bar \zeta }_j})}}$.
\begin{equation}\label{3.f}
{\mathop{\rm Re}\nolimits} {s_{\zeta  = {{\hat \zeta }_j}}}\left[ {\mu _ - ^{(2)}(x,t,\zeta )/\tilde a(\zeta )} \right] = {\tilde c_{{N_1} + j}}{e^{ - 2i\Theta ({{\hat \zeta }_j})}}\mu _ + ^{(2)}x,t,{\hat \zeta _j}) ,
\end{equation}
where ${\tilde c_{{N_1} + j}} = \frac{{q\mathop{{}}\nolimits_{{0}}^{{2}}}}{{ \zeta \mathop{{}}\nolimits_{{j}}^{{2}}}}\frac{{q\mathop{{}}\nolimits_{{+}}}}{{ \overline {q}\mathop{{}}\nolimits_{{+}}}}  {c_j}$.
\begin{equation}\label{3.g}
{\mathop{\rm Re}\nolimits} {s_{\zeta  = {\check \zeta _j}}}\left[ {\mu _ - ^{(1)}(x,t,\zeta )/ a(\zeta )} \right] = {c_{{N_1} + j}}{e^{  2i\Theta ({\check \zeta _j})}}\mu _ + ^{(1)}(x,t,{\check \zeta _j}) ,
\end{equation}
where ${c_{{N_1} + j}} =  \frac{{q\mathop{{}}\nolimits_{{0}}^{{2}}}}{{{\bar \zeta} \mathop{{}}\nolimits_{{j}}^{{2}}}}\frac{{\overline {q}\mathop{{}}\nolimits_{{+}}}}{{ q\mathop{{}}\nolimits_{{+}}}}{\tilde c_j}$. \par
By symmetric conditions (\ref{2.u}) and (\ref{2.v}), known ${\bar c_j} =  - {\tilde c_j}$ and $\overline{{c_{{N_1} + j}}} =  - {\tilde c_{{N_1} + j}}$.

\subsection{Riemann-Hilbert problem}
From equation (\ref{2.u})(\ref{2.v}) and their respective analytical extension regions. Define the sectionally meromorphic matrices
\begin{equation}\label{4.a}
M(x,t,\zeta ) = \left\{ \begin{array}{l}
M_1={e^{iT\int_x^\infty  {(|q{|^2} - q_0^2)dy{\mathop{\sigma_3} }} }}\left( {\frac{{\mu _ - ^{(1)}(x,t,\zeta )}}{{a(\zeta )}},\mu _ + ^{(1)}} \right),\zeta  \in {D_1} .\\
M_2={e^{iT\int_x^\infty  {(|q{|^2} - q_0^2)dy{\mathop{\sigma_3} }} }}\left( {\mu _ + ^{(2)},\frac{{\mu _ - ^{(2)}(x,t,\zeta )}}{{\tilde a(\zeta )}}} \right),\zeta  \in {D_2}.
\end{array} \right.
\end{equation}
Then, we can show $M(x,t,\zeta)$ satisfies the Riemann-Hilbert problem:\par
(1) Analyticity: \label{4.a.1} \par
 \qquad  \qquad \qquad  \quad $M(x,t,\zeta)$ is analytic in $({D_1} \cup {D_2}) / U$ and has simple poles in the set $U$.\par
(2) Jump condition:
\begin{equation}\label{4.b}
{M_1}(x,t,\zeta ) = {M_2}(x,t,\zeta )  J(x,t,\zeta ), \qquad \zeta  \in {\Sigma _\Omega } ,
\end{equation}
where
$$J(x,t,\zeta ) = \left( {\begin{array}{*{20}{c}}
{1 + |r(\zeta ){|^2}}&{  \overline{r(\zeta )}{e^{ - 2i\Theta (x,t,\zeta )}}}\\
{r(\zeta ){e^{2i\Theta (x,t,\zeta )}}}&1
\end{array}} \right) .$$ \par
(3) Asymptotic behavior:\label{4.b.1} \par
 From the proposition (\ref{2.a.q})(\ref{2.a.t}) and definition (\ref{2.a.f}),
\begin{equation}\label{4.c}
{M_{1,2}}(x,t,\zeta ) = \mathbb{E} + O(\frac{1}{\zeta }), \qquad \zeta  \to \infty ,
\end{equation}
\begin{equation}\label{4.d}
{M_{1,2}}(x,t,\zeta ) =  - i{\zeta ^{ - 1}}{\mathop{\sigma_3} }{Q_ + } + O(1)  , \qquad \zeta  \to 0 .
\end{equation}\par
Letting  ${\zeta _j} \mapsto {\eta _j}$, $ - q_0^2\bar \zeta _j^{ - 1} \mapsto {\eta _{{N_1} + j}}$, $j = 1,2,...,{N_1}$.
Then equation (\ref{4.b}) can be regularized as follows,
\begin{equation}\label{4.e}
\begin{array}{l}
{M_1}(x,t,\zeta ) - \mathbb{E} + i{\zeta ^{ - 1}}{\mathop{\sigma_3} }{Q_ + } - \sum\limits_{j = 1}^{2{N_1}} {\left[{\frac{{{\mathop{\rm Re}\nolimits} {s_{\zeta  = {\eta _j}}}{M_1}(\zeta )}}{{\zeta  - {\eta _j}}}} \right]}  - \sum\limits_{j = 1}^{2{N_1}} {\left[ {\frac{{{\mathop{\rm Re}\nolimits} {s_{\zeta  = {{\bar \eta }_j}}}{M_2}(\zeta )}}{{\zeta  - {{\bar \eta }_j}}}} \right]} \\
 = {M_2}(x,t,\zeta ) - \mathbb{E} + i{\zeta ^{ - 1}}{\mathop{\sigma_3} }{Q_ + } - \sum\limits_{j = 1}^{2{N_1}} {\left[ {\frac{{{\mathop{\rm Re}\nolimits} {s_{\zeta  = {{\bar \eta }_j}}}{M_2}(\zeta )}}{{\zeta  - {{\bar \eta }_j}}}} \right]}  - \sum\limits_{j = 1}^{2{N_1}} {\left[ {\frac{{{\mathop{\rm Re}\nolimits} {s_{\zeta  = {\eta _j}}}{M_1}(\zeta )}}{{\zeta  - {\eta _j}}}} \right]}  + {M_2}(\zeta )  \tilde J(\zeta) ,
\end{array}
\end{equation}
where
$$\tilde J = \left( \begin{array}{*{20}{c}}
{|r(\zeta )|^2}&{\overline{r(\zeta )}{e^{ - 2i\Theta (\zeta )}}}\\
{  r(\zeta ){e^{2i\Theta (\zeta )}}}&0
\end{array}
\right).$$
The left side of the above formula is analytic in $D_1$, and the right side is analytic in $D_2$ except for the last item ${M_2}(\zeta ) \cdot \tilde J$. The asymptotics behavior of both are $O(\frac{1}{\zeta })$ as $\zeta  \to \infty $ and $O(1)$ as $\zeta  \to 0$. From (\ref{2.a.r}) and (\ref{2.a.s}), $\tilde J$ are $O(\frac{1}{\zeta })$ as $\zeta  \to \pm\infty $ and $O(\zeta)$ as $\zeta  \to 0$ along the real axis.
Here introduce the Cauchy projectors ${P_ \pm }$ along ${\mathop{\Sigma} _\Omega }$ and by Plemelj's formulae, the equation (\ref{4.e}) can be written as

\begin{equation}\label{4.g}
\begin{split}
M(x,t,\zeta ) = \mathbb{E} - i{\zeta ^{ - 1}}{\mathop{\sigma_3} }{Q_ + } + \sum\limits_{j = 1}^{2{N_1}} {\left[ {\frac{{{\mathop{\rm Re}\nolimits} {s_{\zeta  = {\eta _j}}}{M_1}(x,t,\zeta )}}{{\zeta  - {\eta _j}}} + \frac{{{\mathop{\rm Re}\nolimits} {s_{\zeta  = {{\bar \eta }_j}}}{M_2}(x,t,\zeta )}}{{\zeta  - {{\bar \eta }_j}}}} \right]}  \\
- \frac{1}{{2\pi i}}\int_{{\mathop{\Sigma} _\Omega }} {\frac{{{M_2}(x,t,\eta )\tilde J(x,t,\eta )}}{{\eta  - \zeta }}d\eta } ,
\end{split}
\end{equation}

where the integral path defined by$\int_{{\mathop{\Sigma} _\Omega }} {} $  is shown in Figure(\ref{fig.1}).
Note that, the solution of the Riemann-Hilbert problem (\ref{4.g}) with NZBC for $M_1$ and $M_2$ differs only in the last integral item corresponding to $P_+$ and $P_-$ Cauchy projector, respectively.
\begin{remark}
The solution of the Riemann Hilbert problem under non-zero boundary conditions constitutes a closed algebraic system.
\end{remark}
\begin{proof}
From the residue relationship (\ref{3.c})-(\ref{3.g}) and the distribution of discrete spectra. By the solution of RHP, calculate the second column at points $\zeta  = { \zeta _j}$ and $\zeta  =  - q_0^2  \bar \zeta _j^{ - 1}$. Then
\begin{equation}\label{4.h}
\begin{split}
{\left( {\begin{array}{*{20}{c}}
{{e^{i{\omega _ + }}}}\\
{{e^{ - i{\omega _ + }}}}
\end{array}} \right)} \ast \mu _{+}^{(1)}(x,t,{\eta _j}) = \left( \begin{array}{*{20}{c}}
 - i\eta _j^{ - 1}{q_ + }\\
  1
\end{array} \right) + \sum\limits_{j_1 = 1}^{2{N_1}} {{\left( {\begin{array}{*{20}{c}}
{{e^{i{\omega _ + }}}}\\
{{e^{ - i{\omega _ + }}}}
\end{array}} \right)} \ast \frac{{{{\tilde c}_{j_1}}{e^{ - 2i\Theta ({{\bar \eta }_{j_1}})}}}}{{{\eta _j} - {{\bar \eta }_{j_1}}}}} \mu _+^{(2)}(x,t,{\bar \eta _{j_1}}) \\- \frac{1}{{2\pi i}}\int_{{\mathop{\Sigma} _\Omega }} {\frac{{{{({M_1}\tilde J)}_2}}}{{\eta  - {\eta _j}}}d\eta }, \qquad j=1,2,...,2N_1 ,
\end{split}
\end{equation}
here and after, $' \ast '$ stands for Hadamard product.
Similarly, the first column of the RHP solution can also be obtained at $\zeta  = {\bar \zeta _j}$ and $\zeta  =  - q_0^2 \zeta _j^{ - 1}$,
\begin{equation}\label{4.i}
\begin{split}
{\left( {\begin{array}{*{20}{c}}
{{e^{i{\omega _ + }}}}\\
{{e^{ - i{\omega _ + }}}}
\end{array}} \right)}  \ast \mu _{+}^{(2)}(x,t,{\bar \eta _j}) = \left( {\begin{array}{*{20}{c}}
1\\
{ - i{{\bar \eta }_j^{ -1 }}{{\bar q}_ + }}
\end{array}} \right) + \sum\limits_{j_2 = 1}^{2{N_1}} {{\left( {\begin{array}{*{20}{c}}
{{e^{i{\omega _ + }}}}\\
{{e^{ - i{\omega _ + }}}}
\end{array}} \right)} \ast \frac{{{c_{j_2}}{e^{2i\Theta ({\eta _{j_2}})}}}}{{{{\bar \eta }_j} - {\eta _{j_2}}}}\mu _+^{(1)}(x,t,{\eta _{j_2}})}  \\ - \frac{1}{{2\pi i}}\int_{{\mathop{\Sigma} _\Omega }} {\frac{{{{({M_2}J)}_1}}}{{\eta  - {{\bar \eta }_j}}}d\eta },  \qquad j=1,2,...,2N_1 ,
\end{split}
\end{equation}
where
 \begin{equation}\label{4.j}
{\left( {\begin{array}{*{20}{c}}
{{e^{i{\omega _ + }}}}&{{e^{ - i{\omega _ + }}}}
\end{array}} \right)^{\rm T}} = {\left( {\begin{array}{*{20}{c}}
{{e^{i\int_x^{ + \infty } {T(|q{|^2} - q_0^2)dy} }}}&{{e^{-i\int_x^{ + \infty } {T(|q{|^2} - q_0^2)dy} }}}
\end{array}} \right)^{\rm T}}.
 \end{equation}
From these equations (\ref{4.h}) and (\ref{4.i}), for $j=1,2,...,2N_1$, $2N_1$ equations and $2N_1$ unknown $\mu _+^{(2)}(x,t,{\bar \eta _j})$ are formed. Combined with the RHP solution (\ref{4.g}), a closed algebraic system of $M$ is provided by the scattering data.
\end{proof}
~~~
\subsection{Reconstruction formula for the potential with simple poles}
According to the solution of RHP(\ref{4.g}), the asymptotic expansion formula of $M(x,t,\zeta)$ as $\zeta  \to \infty $ can be derived as follows
\begin{equation}\label{4.k}
M(x,t,\zeta ) = \mathbb{E} + \frac{{{M^{(1)}}(x,t,\zeta )}}{\zeta } + O(\frac{1}{{{\zeta ^2}}}), \qquad   \zeta  \to \infty ,
\end{equation}
where
\begin{equation}\label{4.l}
 \begin{split}
\begin{array}{l}
{M^{(1)}}(\zeta ) =\!  - i{\mathop{\sigma_3} }{Q_ + } + \sum\limits_{j = 1}^{2{N_1}} {\left[ {\begin{array}{*{20}{c}}
{{\left( {\begin{array}{*{20}{c}}
{{e^{i{\omega _ + }}}}\\
{{e^{ - i{\omega _ + }}}}
\end{array}} \right)} \ast  {c_j}{e^{2i\Theta ({\eta _j})}}\mu _ + ^{(1)}({\eta _j}),}&{ {\left( {\begin{array}{*{20}{c}}
{{e^{i{\omega _ + }}}}\\
{{e^{ - i{\omega _ + }}}}
\end{array}} \right)} \ast {{\tilde c}_j}{e^{ - 2i\Theta ({{\bar \eta }_j})}}\mu _ + ^{(2)}({{\bar \eta }_j})}
\end{array}} \right]}
\end{array}\\  +\frac{1}{{2\pi i}}\int_{{\mathop{\Sigma} _\Omega }} {({M_2}\tilde J)d\eta }.
\end{split}
\end{equation}
Letting ${e^{-iT\int_x^\infty  {(|q{|^2} - q_0^2)} dy{\mathop{\sigma_3} }}}M(x,t,\zeta ){e^{ - i\Theta (x,t,\zeta ){\mathop{\sigma_3} }}}$ and choose $M(x,t,\zeta)=M_2(x,t,\zeta)$. Bring it into the Lax pari (\ref{2.c}) and calculate the coefficient of the 1,2 elements of $M(x,t,\zeta)$ in $\zeta^0$.Then we have the following results
\begin{proposition}\label{4.m}
The reconstruction formula for the potential with simple poles of the focusing Kundu-Eckhaus equation with NZBCs can be expressed as follows
\begin{equation}\label{4.n}
\begin{split}
q(x,t) = {e^{ i ( \Upsilon-{\omega _ +) }}} \{ {q_ + }{e^{ - i{\omega _ + }}} + i\sum\limits_{j = 1}^{2{N_1}} {{{\tilde c}_j}{e^{ - 2i\Theta ({{\bar \eta }_j})}}\mu _{ + 1,2}^{(2)}({{\bar \eta }_j}) + } {e^{ - i{\omega _ + }}}\frac{1}{{2\pi }}\int_{{\Sigma _\Omega }} {{{\left( {{M_2}\tilde J} \right)}_{1,2}}d\eta } \} .
\end{split}
\end{equation}

\end{proposition}
~~~
\subsection{Trace formulae and the condition $(\theta) $ with simple poles}
Letting $T\int_{ - \infty }^{ + \infty } {(|q{|^2} - q_0^2)dy}  \mapsto  {\omega _0}$, obtaining
\begin{equation}\label{4.s}
 a(\zeta ) = \exp \left[ { - \frac{1}{{2\pi i}}\int_{{\mathop{\Sigma} _\Omega }} {\frac{{\log [1 + r(\eta )\overline{r(\bar \eta )}]}}{{\eta  - \zeta }}d\eta } } \right]\prod\limits_{j = 1}^{{N_1}} {\frac{{(\zeta  - {\zeta _j})(\zeta  + q_0^2\bar \zeta _j^{ - 1})}}{{(\zeta  - {{\bar \zeta }_j})(\zeta  + q_0^2\zeta _j^{ - 1})}}} {e^{i{\omega _ 0 }}},
\end{equation}
\begin{equation}\label{4.t}
\tilde a(\zeta ) = \exp \left[ {\frac{1}{{2\pi i}}\int_{{\mathop{\Sigma} _\Omega }} {\frac{{\log [1 + r(\eta )\overline{r(\bar \eta )}]}}{{\eta  - \zeta }}d\eta } } \right]\prod\limits_{j = 1}^{{N_1}} {\frac{{(\zeta  - {{\bar \zeta }_j})(\zeta  + q_0^2\zeta _j^{ - 1})}}{{(\zeta  - {\zeta _j})(\zeta  + q_0^2\bar \zeta _j^{ - 1})}}} {e^{ - i{\omega _ 0 }}},
\end{equation}
equations (\ref{4.s}) and (\ref{4.t}) are so called trace formulas. when $\zeta  \to 0$, from the symmetry of the scattering data (\ref{2.a.i}), obtaining
\begin{equation}\label{4.u}
\arg \left( {\frac{{{q_ + }}}{{{q_ - }}}} \right) - 2{\omega _0} = 4\sum\limits_{j = 1}^{{N_1}} {\arg ({\zeta _j})}  + \frac{1}{{2\pi }}\int_{{\mathop{\Sigma} _\Omega }} {\frac{{\log [1 + r(\zeta )\overline{r(\bar \zeta )}]}}{\eta }d\eta } ,
\end{equation}
this relationship is called the condition $(\theta) $ in \cite{L.D.F}, also called $''$theta$''$ condition, which establishes the relationship between the asymptotic phase difference, the $\omega _0$, the discrete spectrum values and the scattering coefficients.

\subsection{Reflectionless potential: simple-pole soliton solutions}
The soliton correspond to the scattering data $\left\{ {a(\zeta ),b(\zeta ),{c_j}} \right\}$ for which $b(\zeta)$ vanishes identically. From the closed algebraic system (\ref{4.h}) and (\ref{4.i}) combined with the reconstruction formula (\ref{4.n}). By Cramer's rule, the expression for the potential can be written as
\begin{equation}\label{4.v}
q(x,t) = {e^{ i( \Upsilon - 2{\omega _ + )}}}\cdot {q_ + } \left(1- { \frac{{\left| {{M^\Delta }} \right|}}{{\left| M \right|}}} \right) ,
\end{equation}
where $|M|$ represents the $2{N_1} \times 2{N_1}$ determinants of matrix $M$, $|M^\Delta|$ is a $(2{N_1} + 1) \times (2{N_1} + 1)$ determinant and
$${M_{j,{j_1}}} = {\mathbb{E}_{j,{j_1}}} - \sum\limits_{{j_2} = 1}^{2{N_1}} {\frac{{{{ c}_{{j_2}}}{e^{2i\Theta ({\eta _{{j_2}}})}}}}{{{{\bar \eta }_j} - {\eta _{{j_2}}}}}}  \cdot \frac{{{{\tilde c}_{{j_1}}}{e^{-2i\Theta ({{\bar \eta }_{{j_1}}})}}}}{{{\eta _{{j_2}}} - {{\bar \eta }_{{j_1}}}}},    \qquad  j,j_1=1,...,2N_1 ,$$
$$ {M^\Delta } = \left( {\begin{array}{*{20}{c}}
{{G^{\rm T}}}&0\\
M&F
\end{array}} \right) , \qquad {G^{\rm T}} = \left( {{g_1}, \ldots ,{g_j}, \ldots ,{g_{2{N_1}}}} \right), \quad  {g_j} = {\tilde c_j}{e^{-2i\Theta ({{\bar \eta }_j})}}, \quad j=1,...,2N_1 ,$$
$$F = {({f_1},...,{f_j},...,{f_{2{N_1}}})^{\rm T}},  \qquad   {f_j} = iq_ + ^{ - 1} + \sum\limits_{{j_2} = 1}^{2{N_1}} {\frac{{{c_{{j_2}}}{e^{2i\Theta ({\eta _{{j_2}}})}}}}{{({{\bar \eta }_j} - {\eta _{j_2}}){\eta _{{j_2}}}}}}, \quad j=1,...,2N_1 ,$$   \par
Notice that this formula is an implicit solution because $q(x,t)$ contains the terms of ${e^{ - 2i{\omega _ + }}}$, where ${\omega _ + }$ is a definite integral, but we can get the explicit modulus solutions for the nonlinear equation (\ref{4.v}). \par
Some simple cases are discussed as follows. Since the KE equation admits a scaling symmetry, letting $q_0=1$ and $\alpha=0$ without losing generality,. \par

\begin{itemize}
\item[ $\bullet$ ] As $N_1=1$, when discrete spectrum is a purely imaginary eigenvalue. Simplify the algebraic system, letting ${c_1} = a{e^{ia}}$, ${\zeta _1} = ib$, $T_1=1$, the one-soliton solution of Eq. (\ref{4.v}) reads as
\begin{equation}\label{4.w}
q = {e^{ - 2i{\omega _ + }}}\frac{{{e^{ - {\varphi _1}}}\left( {4{b^2}k_1^2{e^{{\varphi _2}}} + 4i{k_1}{e^{{k_1}(t + x)/a }}( - {b^4}{e^{i{b^2}t}} + {e^{{\varphi _3}}})a  + k_2^2{e^{i({\varphi _3} + {b^2}t)/2}}{a ^2}} \right)}}{{4{b^2}k_1^2{e^{2b(t + x)}} + k_2^2{e^{2(t + x)/b}}{a ^2} - 2{b^2}{k_1}{e^{{k_2}(t + x)}}a \sin (\theta_+  + a  - \frac{{({b^4} - 1)t}}{{2{b^2}}})}},
\end{equation}
where
$${\varphi _1} = \frac{{i(t + b(2ba  + b( - 6 + {b^2})t + 4i(t + x)))}}{{2{b^2}}}, {\varphi _2} = \frac{{i(2{b^2}(\theta_+  + a ) + t + {b^4}t - 4ib{k_1}(t + x)))}}{{2{b^2}}}$$
$${\varphi _3} = i(2(\theta_+  + a ) + \frac{t}{{{b^2}}})$$
with
$${k_1} = {b^2} - 1, {k_2} = {b^2} + 1.$$
\item[-] For example, based on symbolic computation, when choosing $\theta_+  = \pi /2$, $a=1$, $b=2$. obtaining
\begin{equation}\label{4.x}
q(x,t) = {e^{ - 2i{\omega _ + }}}({q_1}/{q_2}),
\end{equation}
with
$${q_1} = -25e^{i+\frac{19it}{8}} + 12{e^{\frac{1}{2}(4i+(3+i)t+3x)}} - 144{e^{  i + (3 + \frac{{19i}}{8})t + 3x} + 192{e^{\frac{1}{4}((6+17i)t+6x)}}},$$
$${q_2} = -25e^{i+\frac{15it}{8}} + 48i{e^{2i + \frac{3(t+x)}{2}}} + 48i{e^{\frac{3}{4}((2+5i)t+2x)} - 144{e^{i+(3+\frac{15i}{8})t+3x}}},$$
$${\omega _ + } =  - \frac{{3( - 25 + 48{e^{6t + \frac{{3(t+x)}}{2}}}\cos \left( {1 - \frac{{15t}}{4}} \right))}}{{25 + 144{e^{3(t+x)}} - 96{e^{\frac{{3(t+x)}}{2}}}\cos \left( {1 - \frac{{15t}}{4}} \right)}}.$$
Figure (\ref{fig.5-1}) shows the dynamic structure of this case.
\end{itemize}

\begin{figure}[htpb]
\centering
{
\begin{minipage}{6cm}
\centering

\includegraphics[width=6.4 cm]{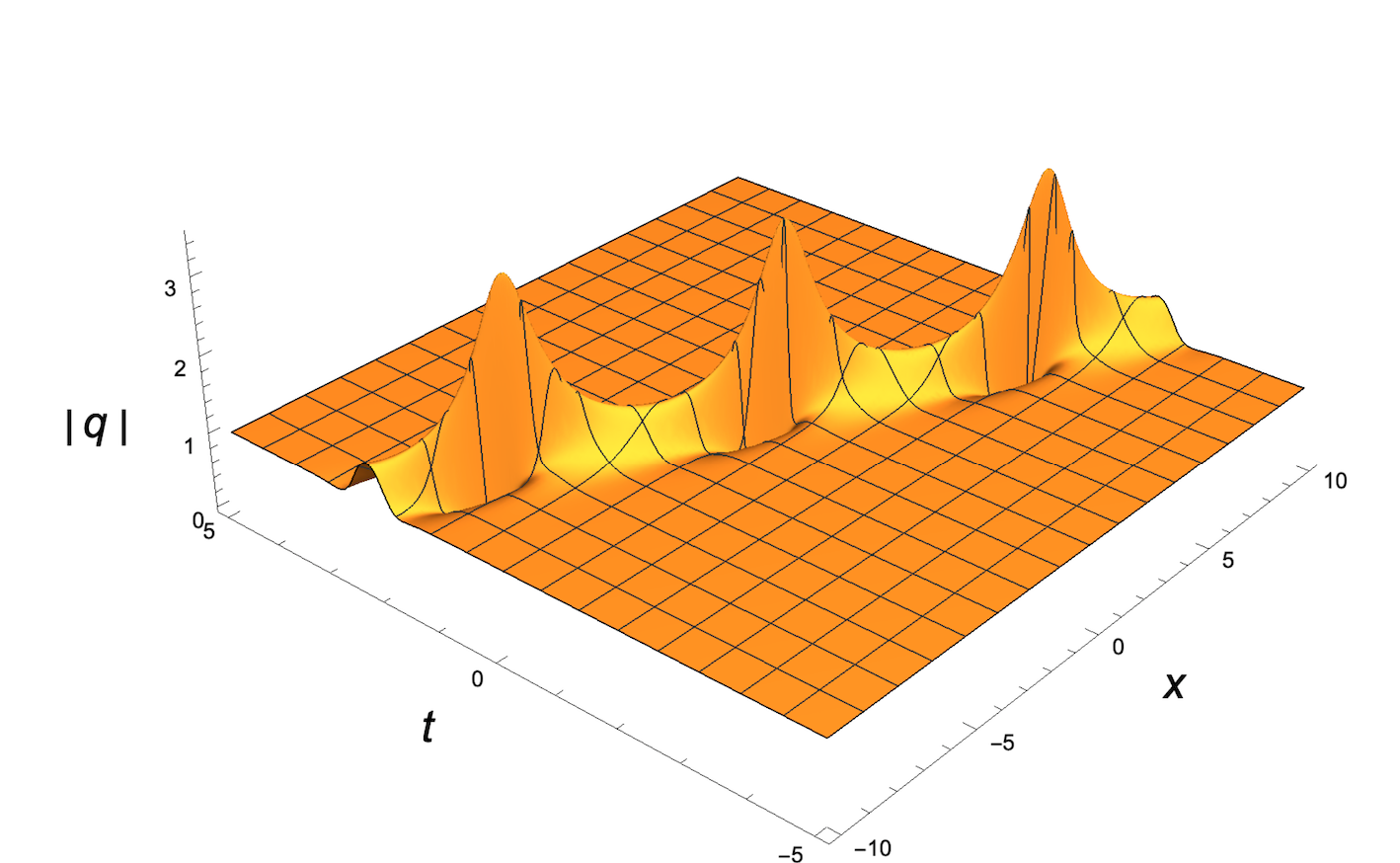}
\end{minipage}
} \qquad
{
\begin{minipage}{6cm}
\centering
\includegraphics[width=6 cm]{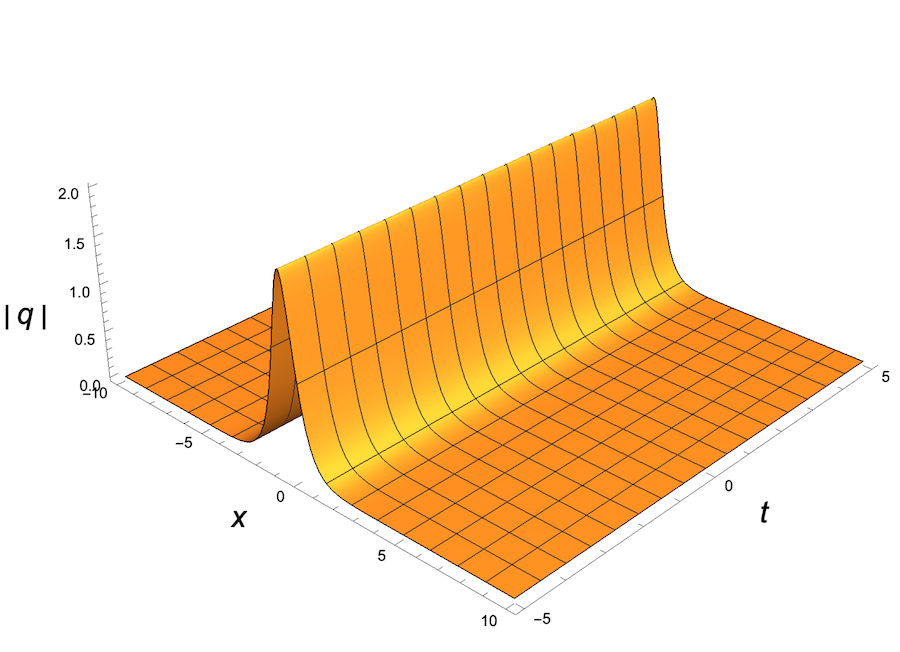}
\end{minipage}
}
\caption{Simple-pole soliton solutions of the focusing KE equation. Left: breather solution with parameters $q_0=1$, $\theta_+ = \pi/2$, $a =1$, $b=2$, $T_1=1$; Right: bright soliton through the limit ${q_ + } \to 0$ and take the same remaining parameters as the left side.}
\label{fig.5-1}
\end{figure}

It can be seen from the condition ($\theta$)(\ref{4.u}) that when the eigenvalues is a pure imaginary value, the asymptotic phase difference is $2(\pi  + {\omega _0})$, in which case corresponding the soliton solution is non-stationary(left). Nevertheless, while ${q_ + } \to 0$, we have ${\omega _0} \to 0$. At this time, the soliton is static with ZBC(right).

\begin{itemize}
\item[ $\bullet$ ] As $N_1=1$, the eigenvalue is located outside the $q_0$-circle in the upper half plane expect the imaginary axis, Figure(\ref{fig.5-11})  shows the dynamic structure of the soliton solution in this case. Different from pure imaginary eigenvalues, both the breather soliton and bright soliton are non-stationary in this case, which can be easily seen from the condition ($\theta$)(\ref{4.u}).
\end{itemize} \par
\begin{itemize}
\item[ $\bullet$ ] As $N_1=2$, Figure (\ref{fig.6}) shows the interaction between two solitons of the focusing KE equation with NZBCs. It can be seen that the difference of $T$ will lead to the phase difference of soliton solutions. Interestingly, in the case of two solitons, when $T$ changes, the phase of one of the solitons is significantly affected, while the other is almost unchanged. Figure (\ref{fig.6}) shows the interaction of two bright solitons when boundary vanishes identically.

\end{itemize} \par

\begin{figure}[htpb]
\centering
{
\begin{minipage}{6cm}
\centering
\includegraphics[width=6.4 cm]{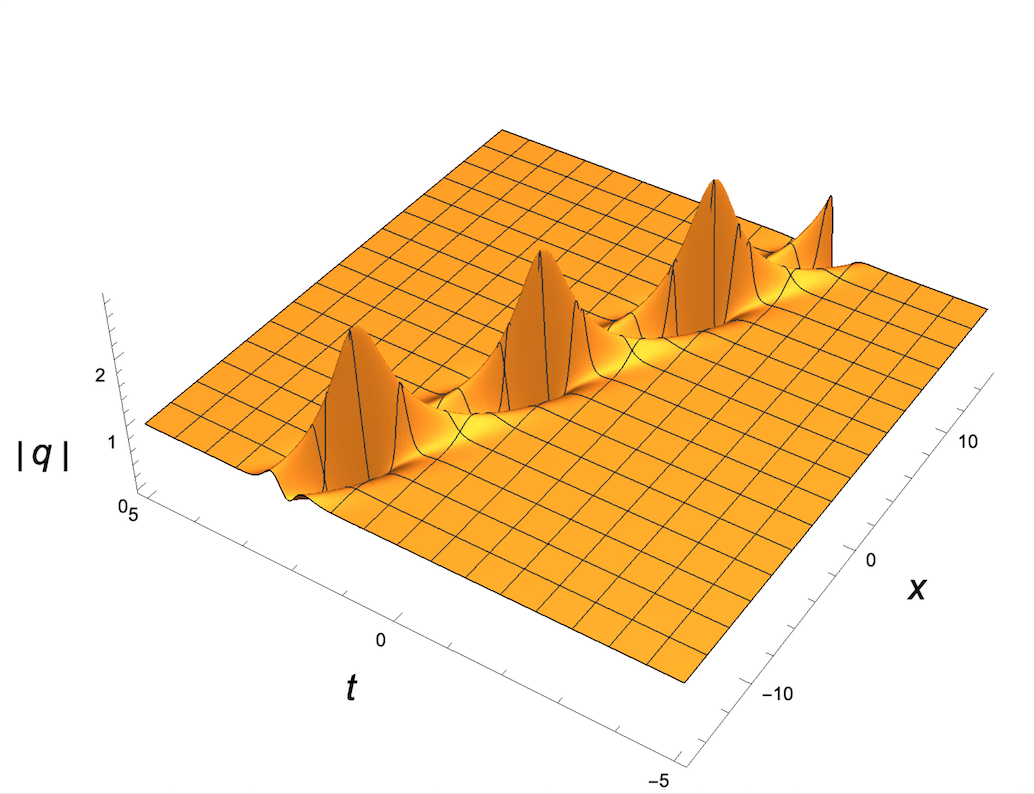}
\end{minipage}
} \qquad
{
\begin{minipage}{6cm}
\centering
\includegraphics[width=7.4cm, height=5cm]{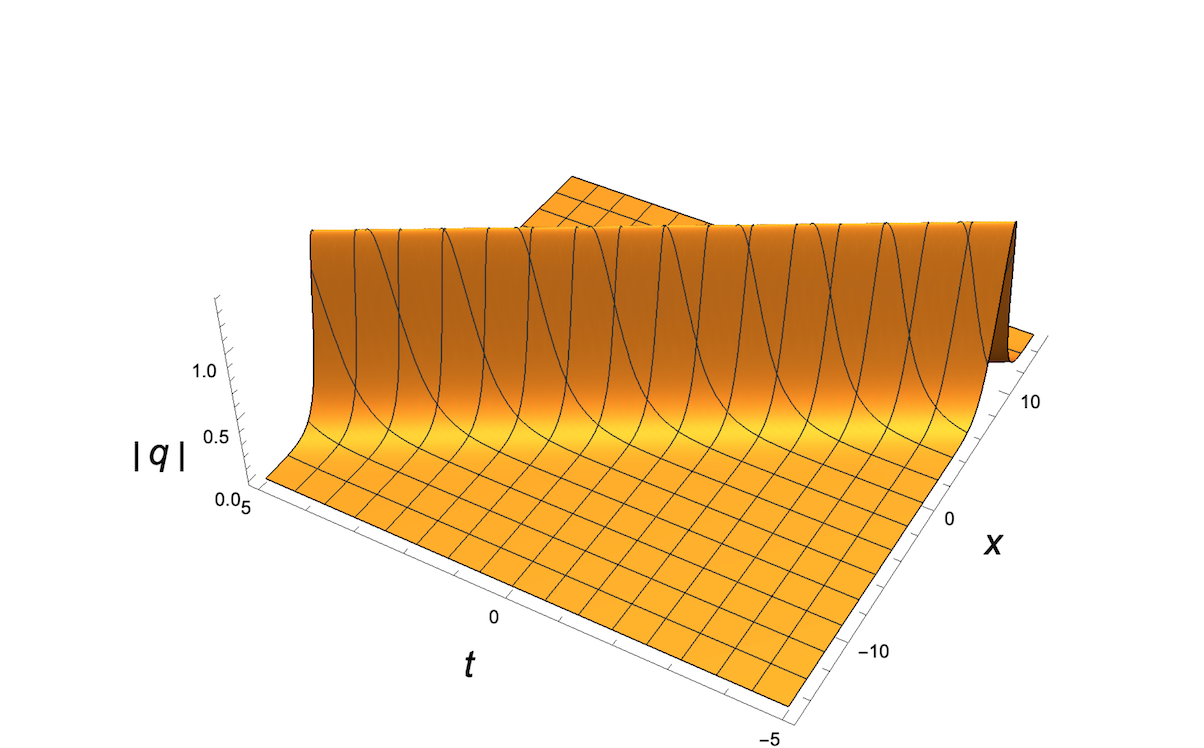}
\end{minipage}
}
\caption{Simple-pole soliton solutions of the focusing KE equation. Left: breather solution with parameters $q_0=1$, $\alpha=\theta_+=0$ , $\zeta_1=\frac{1}{2} +i$, $c_1=e^{1+i}$, $T_1=1$; Right: bright soliton through the limit ${q_ 0 } \to 0$ and take the same remaining parameters as the left side.}
\label{fig.5-11}
\end{figure}

\vspace{0cm}

\begin{figure}[htpb]
\centering
\vspace{-0.35cm}
\subfigtopskip=2pt
\subfigbottomskip=-5pt
\subfigcapskip=-5pt
\subfigure[]
{
\begin{minipage}[t]{0.3\linewidth}
\centering
\includegraphics[width=5cm]{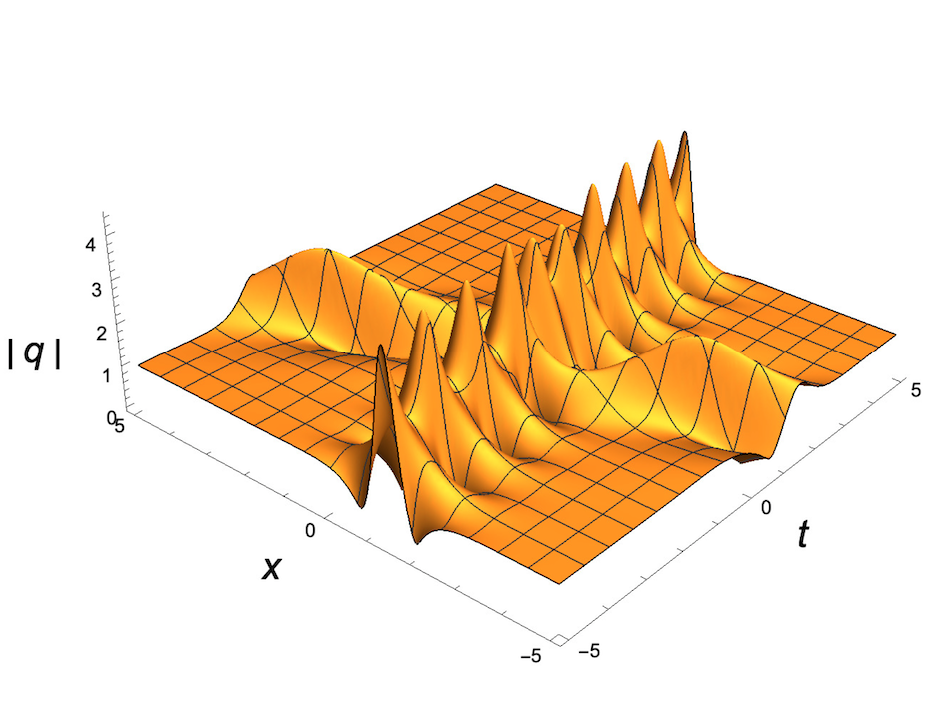}
\end{minipage}%
}%
\subfigure[]
{
\begin{minipage}[t]{0.3\linewidth}
\centering
\includegraphics[width=5cm]{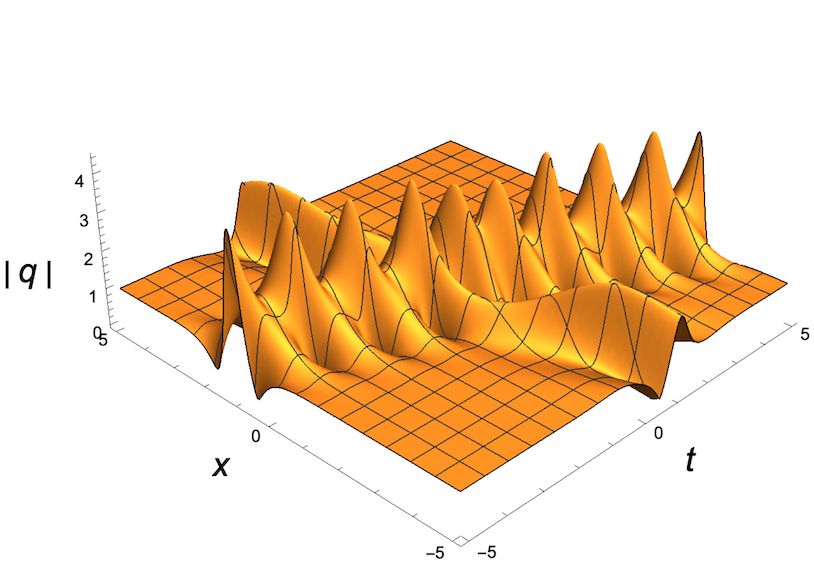}
\end{minipage}%
}%
\subfigure[]
{
\begin{minipage}[t]{0.3\linewidth}
\centering
\includegraphics[width=5cm]{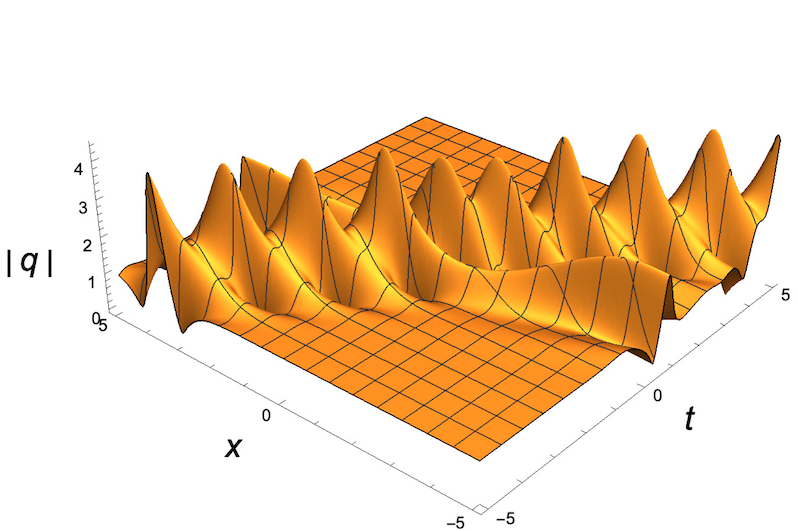}
\end{minipage}
}%

\subfigure[]
{
\begin{minipage}[t]{0.2\linewidth}
\centering
\includegraphics[width=3cm]{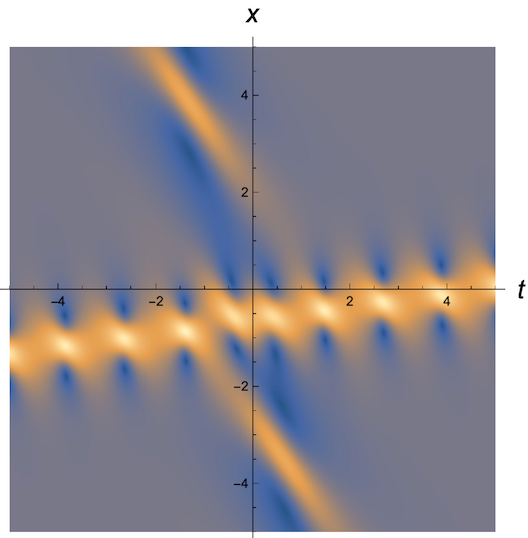}
\end{minipage}%
}
\subfigure[]
{
\begin{minipage}[t]{0.2\linewidth}
\centering
\includegraphics[width=3cm]{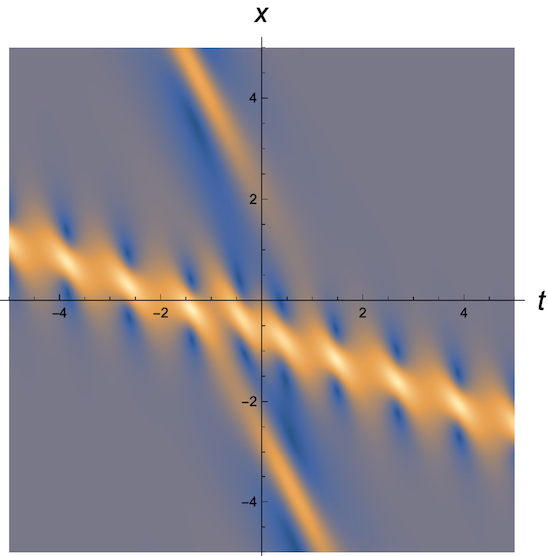}
\end{minipage}%
}
\subfigure[]
{
\begin{minipage}[t]{0.2\linewidth}
\centering
\includegraphics[width=3cm]{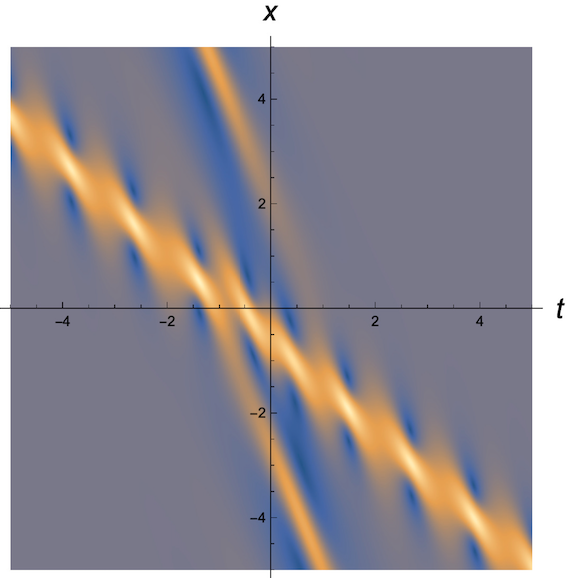}
\end{minipage}
}
\subfigure[]
{
\setlength{\abovedisplayskip}{3pt}
\begin{minipage}[t]{0.2\linewidth}
\centering
\includegraphics[width=5cm]{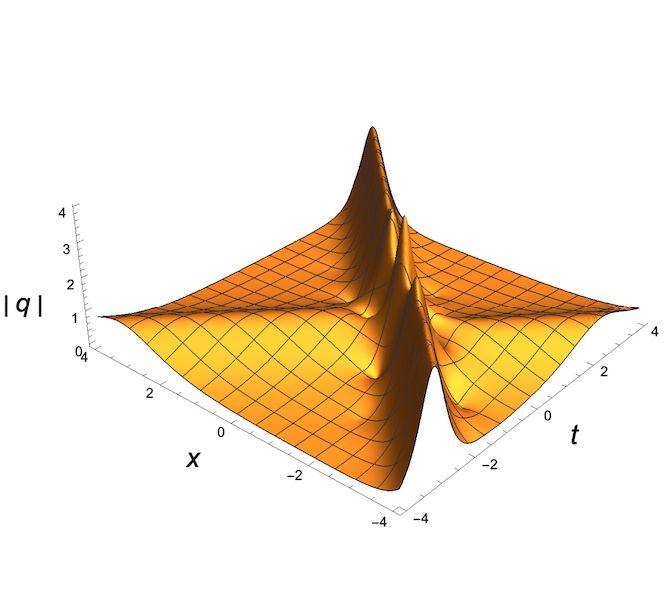}
\end{minipage}
\label{fig.6-1}
}

\centering

\caption{Interaction of two-soliton solutions of focusing KE equation with NZBCs $q_0=1$, $\alpha=\theta_+=0$: (a)$\sim$(c) two-breather solutions with parameters, $c_1=e^{2i}$, $c_2=e^{3i}$, $\zeta_1=1+i$, $\zeta_2=-1+3i$, $T=1$ (a); $T=\frac{3}{2}$ (b); $T=2$ (c). (d)$\sim$(f) the density of the two-breather solutions (a)$\sim$(c) are shown separately. (g) two-bright solitons with parameters, $c_1=e^{2i}$, $c_2=e^{3i}$, $\zeta_1=1+i$, $\zeta_2=-1+3i$, $T=1$ while ${q_ 0 } \to 0$.}
\label{fig.6}
\end{figure}

\begin{remark}\label{4.y}
The solution of the focusing KE equation (\ref{4.v}) is reduced to the explicit solution of the NLS equation (on P12 of \cite{G.G}) while the parameter $T$ vanishes identically.
\end{remark}

\section{The focusing KE equation with NZBCs case: double poles}

In the case where the discrete spectrum are double poles. We have $a(\zeta_j ) = \dot a(\zeta_j ) = 0$, and $\ddot a(\zeta_j ) \ne 0$, where $\left\{ {{\zeta _j}\left| {{\zeta _j} \in \mathbb{C},|{\zeta _j}| > } \right.{q_0},{\mathop{\rm Im}\nolimits} \zeta  > 0,j=1,2,...,N_2} \right\}$. Denoting the discrete spectrum set of the scattering problem is
$$ {U_2} = \left\{ {{\zeta _j},{{\bar \zeta }_j}, - q_0^2{\zeta ^{ - 1}_j} , - q_0^2{{\bar \zeta }^{ - 1}_j}} \right\}_{j = 1}^{{N_2}}.$$\par
Recall $ - q_0^2{\zeta ^{ - 1}_j} \mapsto {\hat \zeta _j}$, $ - q_0^2{\bar \zeta ^{ - 1}_j} \mapsto {\check \zeta _j}$.
Most of the direct scattering problems are invariabilities while $a(\zeta)$ has double-zero. when $\zeta_j$ is a double-zero of the scatter data, from the equation (\ref{2.s}) and the definition of the double-zero, obtaining
\begin{equation}\label{5.a}
\mu _ - ^{(1)}(x,t,{\zeta _j}) = {e^{2i\Theta (x,t,{\zeta _j})}}{b_j}\mu _ + ^{(1)}(x,t,{\zeta _j}),
\end{equation}
and
\begin{equation}\label{5.b}
\dot \mu _ - ^{(1)}(x,t,{\zeta _j}) = {e^{2i\Theta (x,t,{\zeta _j})}}\left( {{d_j}\mu _ + ^{(1)}({\zeta _j}) + 2i\dot \Theta ({\zeta _j})b_j\mu _ + ^{(1)}({\zeta _j}) + {b_j}\dot \mu _ + ^{(1)}({\zeta _j})} \right),
\end{equation}
where $b_j$ and $d_j$ are some constants independent of $x$ and $t$. Other columns can be obtained in a similar way.

\subsection{Residue conditions with double poles}
\begin{lemma}\label{5.i}
Suppose $\zeta_j$ is an isolated pole of order $m$ for $F(\zeta )$, then $F(\zeta )$ has the following Laurent expansion
$F(\zeta ) = \sum\limits_{n =  - m}^\infty  F_n(\zeta  - {\zeta _j})^n $
with
${F_n} = \mathop {\lim }\limits_{\zeta  \to {\zeta _j}} \frac{{{{\left[ {{{(\zeta  - {\zeta _j})}^m}F(\zeta )} \right]}^{(n + m)}}}}{{(n + m)!}}$, $0 < \left| {\zeta  - {\zeta _j}} \right| < R$,where $R$ is the convergence radius of the associated isolated poles \cite{M.J.A}.
\end{lemma} \par
Definition $\mathop {{\rm{Re}}s}\limits_{\zeta  = {\zeta _j}} [F(\zeta )] = {F_{ - 1}}$, $\mathop {{P_{ - 2}}}\limits_{\zeta  = {\zeta _j}} F[\zeta ]$ are the residuals of Laurent expansion and the coefficients of ${(\zeta  - {\zeta _j})^{ - 2}}$, respectively. We have $f$ and $g$ are analytic in region $D \in \mathbb{C}$, and  $g$ has double zeros at ${\zeta _j} \in D$ and $f({\zeta _j}) \ne 0$, $j=1,2,...,N_2$. Known by the lemma (\ref{5.i}) \par
\begin{equation}\label{5.j}
\mathop {{\rm{Re}}s}\limits_{\zeta  = {\zeta _j}} \left[ {\frac{f}{g}} \right] = \frac{{2\dot f({\zeta _j})}}{{\ddot g({\zeta _j})}} - \frac{{2f({\zeta _j}){g^{(3)}}({\zeta _j})}}{{3{{(\ddot g({\zeta _j}))}^2}}}, \qquad \mathop {{P_{ - 2}}}\limits_{\zeta  = {\zeta _j}} \left[ {\frac{f}{g}} \right] = \frac{{2f({\zeta _j})}}{{\ddot g({\zeta _j})}}.
\end{equation}

Then, by (\ref{5.a})-(\ref{5.b}), denote ${A_j} = \frac{{2{b_j}}}{{\ddot a({\zeta _j})}}$ and ${B_j} = \frac{{{d_j}}}{{{b_j}}} - \frac{{{a^{(3)}}({\zeta _j})}}{{3\ddot a({\zeta _j})}}$, obtaining
\begin{equation}\label{5.k}
\mathop {{P_{ - 2}}}\limits_{\zeta  = {\zeta _j}} \left[ {\frac{{\mu _ - ^{(1)}(x,t,\zeta )}}{{a(\zeta )}}} \right] = {A_j} \cdot {e^{2i\Theta ({\zeta _j})}}\mu _ + ^{(1)}(\zeta_j),
\end{equation}

\begin{equation}\label{5.l}
\mathop {{\mathop{\rm Re}\nolimits} s}\limits_{\zeta  = {\zeta _j}} \left[ {\frac{{\mu _ - ^{(1)}(x,t,\zeta )}}{{a(\zeta )}}} \right] = {A_j} \cdot {e^{2i\Theta ({\zeta _j})}}\left[ {\dot \mu _ + ^{(1)}({\zeta _j}) + \mu _ + ^{(1)}({\zeta _j}) \cdot \left( {{B_j} + 2i\dot \Theta ({\zeta _j})} \right)} \right],
\end{equation}
similar,
\begin{equation}\label{5.m}
\mathop {{P_{ - 2}}}\limits_{\zeta  = {{\bar \zeta }_j}} \left[ {\frac{{\mu _ - ^{(2)}(x,t,\zeta )}}{{\tilde a(\zeta )}}} \right] = {\tilde A_j} \cdot {e^{ - 2i\Theta ({{\bar \zeta }_j})}}\mu _ + ^{(2)}({\bar \zeta _j}),
\end{equation}
\begin{equation}\label{5.n}
\mathop {{\rm{Re}}s}\limits_{\zeta  = {{\bar \zeta }_j}} \left[ {\frac{{\mu _ - ^{(2)}(x,t,\zeta )}}{{\tilde a(\zeta )}}} \right] = {\tilde A_j} \cdot {e^{2i\Theta ({{\bar \zeta }_j})}}\left[ {\dot \mu _ + ^{(2)}({{\bar \zeta }_j}) + \mu _ + ^{(2)}({{\bar \zeta }_j}) \cdot \left( {{{\tilde B}_j} - 2i\dot \Theta ({{\bar \zeta }_j})} \right)} \right],
\end{equation}
where ${\tilde A_j} = \frac{{2{{\tilde b}_j}}}{{\tilde {\ddot a}({{\bar \zeta }_j})}} = \frac{{ - 2{b_j}}}{{\overline{\ddot a({\zeta _j}})}} =  - {\overline {A_j}}$, ${\tilde B_j} = \frac{{{{\tilde d}_j}}}{{{{\tilde b}_j}}} - \frac{{{a^{(3)}}({{\bar \zeta }_j})}}{{3\ddot a({{\bar \zeta }_j})}} = \frac{{{{\bar d}_j}}}{{{{\bar b}_j}}} - \frac{{{\overline{a^{(3)}({\zeta _j}}})}}{{3\overline{\ddot a({\zeta _j}})}} = {\overline {B_j}}$.

\begin{equation}\label{5.o}
\mathop {{P_{ - 2}}}\limits_{\zeta  = {{\hat \zeta }_j}} \left[ {\frac{{\mu _ - ^{(2)}(x,t,\zeta )}}{{\tilde a(\zeta )}}} \right] = {\hat A_j} \cdot {e^{ - 2i\Theta ({{\hat \zeta }_j})}}\mu _ + ^{(2)}({\hat \zeta _j}),
\end{equation}
\begin{equation}\label{5.p}
\mathop {{\rm{Re}}s}\limits_{\zeta  = {{\hat \zeta }_j}} \left[ {\frac{{\mu _ - ^{(2)}(x,t,\zeta )}}{{\tilde a(\zeta )}}} \right] = {\hat A_j} \cdot {e^{ - 2i\Theta ({{\hat \zeta }_j})}}\left[ {\dot \mu _ + ^{(2)}({{\hat \zeta }_j}) + \mu _ + ^{(2)}({{\hat \zeta }_j}) \cdot \left( {{{\hat B}_j} - 2i\dot \Theta ({{\hat \zeta }_j})} \right)} \right],
\end{equation}
where ${\hat A_j} = {A_j}\frac{{q_+q_0^4}}{{\bar q_+\zeta _j^4}}$, ${\hat B_j} \cdot \frac{{q_0^2}}{{\zeta _j^2}} = {B_j} - \frac{2}{{{\zeta _j}}}$.

\begin{equation}\label{5.q}
\mathop {{P_{ - 2}}}\limits_{\zeta  = {{\check \zeta }_j}} \left[ {\frac{{\mu _ - ^{(1)}(x,t,\zeta )}}{{a(\zeta )}}} \right] = {\check A_j} \cdot {e^{2i\Theta ({{\check \zeta }_j})}}\mu _ + ^{(1)}({\check \zeta _j}),
\end{equation}
\begin{equation}\label{5.r}
\mathop {{\rm{Re}}s}\limits_{\zeta  = {{\check \zeta }_j}} \left[ {\frac{{\mu _ - ^{(1)}(x,t,\zeta )}}{{a(\zeta )}}} \right] = {\check A_j} \cdot {e^{2i\Theta ({{\check \zeta }_j})}}\left[ {\dot \mu _ + ^{(1)}({{\check \zeta }_j}) + \mu _ + ^{(1)}({{\check \zeta }_j}) \cdot \left( {{{\check B}_j} + 2i\dot \Theta ({{\check \zeta }_j})} \right)} \right],
\end{equation}
where ${\check A_j} =  - {\overline {\hat A_j}}$, ${\check B_j} =   {\overline {\hat B_j}}$.
~~~
\subsection{Inverse problem with double poles}
The RHP condition (\ref{4.a}) -(\ref{4.b}) is still hold in the case of double poles. However, it is necessary to renormalize RHP conditions. Denoting ${\zeta _j} \mapsto {\eta _j}$, $ - q_0^2\bar \zeta _j^{ - 1} \mapsto {\eta _{{N_2} + j}}$, $ - q_0^2\zeta _j^{ - 1} \mapsto {\hat \eta _j}$, ${\bar \zeta _j} \mapsto {\hat \eta _{{N_2} + j}}$, $j=1,2,...,N_2$. By subtracting the asymptotic behavior as $\zeta  \to \infty $, $\zeta  \to 0$ and the pole contribution, then equation (\ref{4.b}) can be written as
\begin{equation}\label{5.s}
\begin{array}{*{20}{l}}
{{M_1}(x,t,\zeta ) - \mathbb{E} + i{\zeta ^{ - 1}}{\mathop{\sigma_3} }{Q_ + } - \sum\limits_{j = 1}^{2{N_1}} {\left[ {\frac{{\mathop {{\rm{Re}}s}\limits_{\zeta  = {\eta _j}} {M_1}(\zeta )}}{{\zeta  - {\eta _j}}} + \frac{{\mathop {{P_{ - 2}}}\limits_{\zeta  = {\eta _j}} {M_1}(\zeta )}}{{{{(\zeta  - {\eta _j})}^2}}} + \frac{{\mathop {{\mathop{\rm Re}\nolimits} s}\limits_{\zeta  = {{\hat \eta }_j}} {M_2}(\zeta )}}{{\zeta  - {{\hat \eta }_j}}} + \frac{{\mathop {{P_{ - 2}}}\limits_{\zeta  = {{\hat \eta }_j}} {M_2}(\zeta )}}{{{{(\zeta  - {{\hat \eta }_j})}^2}}}} \right]} }\\
{ = {M_2}(x,t,\zeta ) - \mathbb{E} + i{\zeta ^{ - 1}}{\mathop{\sigma_3} }{Q_ + } - \sum\limits_{j = 1}^{2{N_1}} {\left[ {\frac{{\mathop {{\rm{Re}}s}\limits_{\zeta  = {\eta _j}} {M_1}(\zeta )}}{{\zeta  - {\eta _j}}} + \frac{{\mathop {{P_{ - 2}}}\limits_{\zeta  = {\eta _j}} {M_1}(\zeta )}}{{{{(\zeta  - {\eta _j})}^2}}} + \frac{{\mathop {{\mathop{\rm Re}\nolimits} s}\limits_{\zeta  = {{\hat \eta }_j}} {M_2}(\zeta )}}{{\zeta  - {{\hat \eta }_j}}} + \frac{{\mathop {{P_{ - 2}}}\limits_{\zeta  = {{\hat \eta }_j}} {M_2}(\zeta )}}{{{{(\zeta  - {{\hat \eta }_j})}^2}}}} \right]}  + {M_2}(\zeta )  \tilde J(\zeta )},
\end{array}
\end{equation}
where
$$\tilde J = \left( \begin{array}{*{20}{c}}
{|r(\zeta )|^2}&{\overline{r(\zeta )}{e^{ - 2i\Theta (\zeta )}}}\\
{  r(\zeta ){e^{2i\Theta (\zeta )}}}&0
\end{array} \right).$$
The left side of the above formula is analytic in $D_1$, and the right side is analytic in $D_2$ except for the last item ${M_2}(\zeta ) \cdot \tilde J$. The asymptotics behavior of both is $O(\frac{1}{\zeta })$ as $\zeta  \to \infty $ and $O(1)$ as $\zeta  \to 0$. From (\ref{2.a.u}) and (\ref{2.a.v}), $\tilde J$ is $O(\frac{1}{\zeta })$ as $\zeta  \to \pm\infty $ and $O(\zeta)$ as $\zeta  \to 0$ along the real axis. Apply Cauchy projectors and Plemelj's formulae, the equation (\ref{5.s}) can be written as
\begin{equation}\label{5.t}
\begin{split}
M(x,t,\zeta ) =\mathbb{E} - i{\zeta ^{ - 1}}{\mathop{\sigma_3} }{Q_ + } + \sum\limits_{j = 1}^{2{N_2}} {\left[ {\frac{{\mathop {{\rm{Re}}s}\limits_{\zeta  = {\eta _j}} {M_1}(\zeta )}}{{\zeta  - {\eta _j}}} + \frac{{\mathop {{\rm{ }}{P_{ - 2}}}\limits_{\zeta  = {\eta _j}} {M_1}(\zeta )}}{{{{(\zeta  - {\eta _j})}^2}}} + \frac{{\mathop {{\rm{Re}}s}\limits_{\zeta  = {{\hat \eta }_j}} {M_2}(\hat \zeta )}}{{\zeta  - {{\hat \eta }_j}}} + \frac{{\mathop {{\rm{ }}{P_{ - 2}}}\limits_{\zeta  = {\eta _j}} {M_2}(\hat \zeta )}}{{{{(\zeta  - {{\hat \eta }_j})}^2}}}} \right]} \\  - \frac{1}{{2\pi i}}\int_{{\mathop{\Sigma} _\Omega }} {\frac{{{M_2}(\eta )\tilde J(\eta )}}{{\eta  - \zeta }}d\eta } ,
\end{split}
\end{equation}
where the integral path defined by $\int_{{\mathop{\Sigma} _\Omega }}$ is shown in Figure(\ref{fig.1}).
Note that, the solution of the Riemann-Hilbert problem (\ref{5.t}) with double poles for $M_1$ and $M_2$ differs only in the last integral item corresponding to $P_+$ and $P_-$ cauchy projector, respectively.
\begin{remark}
The solution of the Riemann Hilbert problem with double poles constitutes a closed algebraic system.
\end{remark}
\begin{proof}
From the residue relationship (\ref{5.k})-(\ref{5.r}) and the distribution of discrete spectra. Letting
$${C_j} = {A_j} \cdot {e^{2i\Theta (x,t,{\eta _j})}}, \qquad {\hat C_j} = {\hat A_j} \cdot {e^{ - 2i\Theta (x,t,{{\hat \eta }_j})}},$$
$${D_j} = {B_j} + 2i\dot \Theta (x,t,{\eta _j}), \qquad {\hat D_j} = {\hat B_j} - 2i\dot \Theta (x,t,{\hat \eta _j}), $$
By the solution of RHP, calculate the second column at points $\zeta  = {\bar \zeta _j}$ and $\zeta  =  - q_0^2 \zeta _j^{ - 1}$, obtaining
\begin{equation}\label{5.u}
\begin{split}
\mu _ + ^{(1)}(x,t,\eta ) = \left( {\begin{array}{*{20}{c}}
{{e^{ - i{\omega _ + }}}}\\
{{e^{i{\omega _ + }}}}
\end{array}} \right) \ast \left( {\begin{array}{*{20}{c}}
{ - i{\zeta ^{ - 1}}{q_ + }}\\
1
\end{array}} \right) + \sum\limits_{j = 1}^{2{N_2}} {\left[ {\frac{1}{{\zeta  - {{\hat \eta }_j}}}\left[ {{{\hat C}_j}\left( {\dot \mu _ + ^{(2)}({{\hat \eta }_j}) + \mu _ + ^{(2)}\left( {\frac{1}{{\zeta  - {{\hat \eta }_j}}} + {D_j}} \right)} \right)} \right]} \right]} \\  - \frac{1}{{2\pi i}}\int_{{\mathop{\Sigma} _\Omega }} {\frac{{{{\left( {{M_1}\tilde J} \right)}_2}}}{{\eta  - \zeta }}d\eta } .
\end{split}
\end{equation}
Taking the derivative with respct to $\zeta$ in (\ref{5.u}), obtaining
\begin{equation}\label{5.v}
\begin{split}
\dot \mu _ + ^{(1)}(x,t,\eta ) = \left( {\begin{array}{*{20}{c}}
{{e^{ - i{\omega _ + }}}}\\
{{e^{i{\omega _ + }}}}
\end{array}} \right) \ast \left( {\begin{array}{*{20}{c}}
{i{\zeta ^{ - 2}}{q_ + }}\\
0
\end{array}} \right) - \sum\limits_{j = 1}^{2{N_2}} {\left[ {\frac{1}{{{{(\zeta  - {{\hat \eta }_j})}^2}}}\left[ {{{\hat C}_j}\left( {\dot \mu _ + ^{(2)}({{\hat \eta }_j}) + \mu _ + ^{(2)}\left( {\frac{2}{{\zeta  - {{\hat \eta }_j}}} + {D_j}} \right)} \right)} \right]} \right]}  \\- \frac{1}{{2\pi i}}\int_{{\mathop{\Sigma} _\Omega }} {\frac{{{{\left( {{M_1}\tilde J} \right)}_2}}}{{{{(\eta  - \zeta )}^2}}}d\eta } .
\end{split}
\end{equation}
From the symmetry condition (\ref{2.a.f})
\begin{equation}\label{5.w}
\mu _ + ^{(1)}(x,t,\zeta ) =  - i{\zeta ^{ - 1}}{q_ + }\mu _ + ^{(2)}(x,t, - q_0^2{\zeta ^{ - 1}}),
\end{equation}
corresponding
\begin{equation}\label{5.x}
\dot \mu _ + ^{(1)}(x,t,\zeta ) = i{\zeta ^{ - 2}}{q_ + }\mu _ + ^{(2)}(x,t, - q_0^2{\zeta ^{ - 1}}) - iq_0^2{q_ + } \cdot {\zeta ^{ - 3}}\dot \mu _ + ^{(2)}(x,t, - q_0^2{\zeta ^{ - 1}}).
\end{equation}
Letting $\zeta  = {\eta _k}$ in (\ref{5.v})-(\ref{5.x}), $k=1,2,...,2N_2$, and bring (\ref{5.w}) (\ref{5.x}) into the two formulas (\ref{5.u}) (\ref{5.v}) and simplify them, obtaining
\begin{equation}\label{5.y}
\begin{split}
0 = \left( {\begin{array}{*{20}{c}}
{{e^{ - i{\omega _ + }}}}\\
{{e^{i{\omega _ + }}}}
\end{array}}\right)  \ast  \left( {\begin{array}{*{20}{c}}
{ - \frac{{i{q_ + }}}{{{\eta _k}}}}\\
1
\end{array}} \right) + \sum\limits_{j = 1}^{2{N_2}} {\left( {\frac{{{{\hat C}_j}}}{{{\eta _k} - {{\hat \eta }_j}}}\dot \mu _ + ^{(2)}({{\hat \eta }_j}) + \left( {\frac{{{{\hat C}_j}}}{{{\eta _k} - {{\hat \eta }_j}}}\left( {\frac{1}{{{\eta _k} - {{\hat \eta }_j}}} + {{\hat D}_j}} \right) + \frac{{i{q_ + }}}{{{\eta _k}}}{\delta _{j,k}}} \right)\mu _ + ^{(2)}({{\hat \eta }_j})} \right)} \\ - \frac{1}{{2\pi i}}\int_{{\mathop{\Sigma} _\Omega }} {\frac{{{{\left( {{M_2}\tilde J} \right)}_2}}}{{{{(\eta  - {\eta _k})}^2}}}d\eta } ,
\end{split}
\end{equation}
\begin{equation}\label{5.z}
\begin{split}
0 = \left( {\begin{array}{*{20}{c}}
{{e^{ - i{\omega _ + }}}}\\
{{e^{i{\omega _ + }}}}
\end{array}} \right) \ast & \left( {\begin{array}{*{20}{c}}
{\frac{{i{q_ + }}}{{\eta _k^2}}}\\
0
\end{array}} \right) + \sum\limits_{j = 1}^{2{N_2}} {\left( {\left( {\frac{{{{\hat C}_j}}}{{{{({\eta _k} - {{\hat \eta }_j})}^2}}} - \frac{{2q_0^2{q_ + }}}{{\eta _k^3}}{\delta _{k,j}}} \right)\dot \mu _ + ^{(2)}({{\hat \eta }_j})} \right)} \\
& \left. {+ \left( {\frac{{{{\hat C}_j}}}{{{{({\eta _k} - {{\hat \eta }_j})}^2}}}\left( {\frac{2}{{{\eta _k} - {{\hat \eta }_j}}} + {{\hat D}_j}} \right) + \frac{{i{q_ + }}}{{\eta _k^2}}{\delta _{j,k}}} \right)\mu _ + ^{(2)}({{\hat \eta }_j})} \right) - \frac{1}{{2\pi i}}\int_{{\mathop{\Sigma} _\Omega }} {\frac{{{{\left( {{M_2}\tilde J} \right)}_2}}}{{{{(\eta  - {\eta _k})}^2}}}d\eta } ,
\end{split}
\end{equation}
where $\delta$ represents the Kronecker delta($i.e$ ${\delta _{j,k}} = 1$, if $ j=k$, else $=0$). From these equations (\ref{5.y}) and (\ref{5.z}), for $k=1,2,...,2N_2$, $4N_2$ equations and $4N_2$ unknown functions $\mu _ + ^{(2)}(x,t,{\eta _j})$, $\dot \mu _ + ^{(2)}(x,t,{\eta _j})$, $j=1,2,...,2N_2$ are formed. Combined with the RHP solution (\ref{5.t}) a closed algebraic system of $M$ which provided by the scatting date.
\end{proof}
~~~
\subsection{Reconstruction formula for the potential with double poles}
According to the solution of RHP (\ref{5.t}) with double poles, the asymptotic expansion formula of $M(x,t,\zeta)$ as $\zeta  \to \infty $ can be derived as follows
\begin{equation}\label{5.a.a}
M(x,t,\zeta ) = E + \frac{{{M^{(1)}}(x,t,\zeta )}}{\zeta } + O(\frac{1}{{{\zeta ^2}}}), \qquad   \zeta  \to \infty ,
\end{equation}
where
\begin{equation}\label{5.a.b}
\begin{split}
{M^{(1)}}(x,t,\zeta ) &=  - i{\mathop{\sigma_3} }{Q_ + } - \frac{1}{{2\pi i}}\int_{{\mathop{\Sigma} _\Omega }} {\left( {{M_2}\tilde J} \right)} d\eta \\ +\left( {\begin{array}{*{20}{c}}
{{e^{i{w_ + }}}}\\
{{e^{ - i{w_ + }}}}
\end{array}} \right) \ast & \sum\limits_{j = 1}^{2{N_2}} \left({C_j}\left[ {\dot \mu _ + ^{(1)}(x,t,{\eta _j}) + \mu _ + ^{(1)} \left( {\frac{1}{{{\zeta} - {{ \eta }_j}}} + {{ D}_j}} \right) } \right], {\hat C_j}\left[ {\dot \mu _ + ^{(2)}(x,t,{{\hat \eta }_j}) + \mu _ + ^{(2)} \left( {\frac{1}{{{\zeta} - {{\hat \eta }_j}}} + {{\hat D}_j}} \right) }  \right] \right).
\end{split}
\end{equation}

From the sectionally meromorphic matrices (\ref{4.a}), letting ${e^{-iT\int_x^\infty  {(|q{|^2} - q_0^2)} dy{\mathop{\sigma_3} }}}M(x,t,\zeta ){e^{ - i\Theta (x,t,\zeta ){\mathop{\sigma_3} }}}$ and choose $M(x,t,\zeta)=M_2(x,t,\zeta)$. Bring it into Lax pari (\ref{2.c}) and calculate the coefficient of the 1,2 elements of $M(x,t,\zeta)$ in $\zeta^0$. Then we have the following results,
\begin{proposition}\label{5.a.d}
The reconstruction formula for the potential with double poles of the focusing Kundu-Eckhaus equation with NZBCs can be expressed as follows
\begin{equation}\label{5.a.e}
\begin{split}
q(x,t) = {e^{ i(\Upsilon- {\omega _ + )}}} \left( {{q_ + }{e^{ - i{\omega _ + }}} + i\sum\limits_{j = 1}^{2{N_2}} {{{\hat C}_j}\left( {\dot \mu _ {+1,2} ^{(2)}(x,t,{{\hat \eta }_j}) + \mu _ {+1,2} ^{(2)}(x,t,{{\hat \eta }_j}){{\hat D}_j}} \right)} } \right.\\
\left. { - {e^{ - i{\omega _ + }}}\int_{{\mathop{\Sigma} _\Omega }} {{{\left( {{M_2}\tilde J} \right)}_{1,2}}(x,t,\eta )} d\eta } \right).
\end{split}
\end{equation}
\end{proposition}
~~~
\subsection{Trace formulae and the condition $(\theta) $ with double poles}
Similar to the single zero case, when the scatter data $a(\zeta)$ and $\tilde a(\zeta)$ have double zeros, we can get the following trace formula and the condition $(\theta)$
\begin{equation}\label{5.a.f}
a(\zeta ) = \exp \left[ { - \frac{1}{{2\pi i}}\int_{{\mathop{\Sigma} _\Omega }} {\frac{{\log [1 + r(\eta )\overline {r(\bar \eta )} ]}}{{\eta  - \zeta }}d\eta } } \right]\prod\limits_{j = 1}^{{N_1}} {\frac{{{{(\zeta  - {\zeta _j})}^2}{{(\zeta  + q_0^2\bar \zeta _j^{ - 1})}^2}}}{{{{(\zeta  - {{\bar \zeta }_j})}^2}{{(\zeta  + q_0^2\zeta _j^{ - 1})}^2}}}} {e^{i{\omega _0}}},
\end{equation}
\begin{equation}\label{5.a.g}
\tilde a(\zeta ) = \exp \left[ {\frac{1}{{2\pi i}}\int_{{\mathop{\Sigma} _\Omega }} {\frac{{\log [1 + r(\eta )\overline {r(\bar \eta )} ]}}{{\eta  - \zeta }}d\eta } } \right]\prod\limits_{j = 1}^{{N_1}} {\frac{{{{(\zeta  - {{\bar \zeta }_j})}^2}{{(\zeta  + q_0^2\zeta _j^{ - 1})}^2}}}{{{{(\zeta  - {\zeta _j})}^2}{{(\zeta  + q_0^2\bar \zeta _j^{ - 1})}^2}}}} {e^{ - i{\omega _0}}},
\end{equation}
and
\begin{equation}\label{5.a.h}
\arg \left( {\frac{{{q_ + }}}{{{q_ - }}}} \right) - 2{\omega _0} = 8\sum\limits_{j = 1}^{{N_1}} {\arg ({\zeta _j})}  + \frac{1}{{2\pi }}\int_{{\mathop{\Sigma} _\Omega }} {\frac{{\log [1 + r(\zeta )\overline {r(\bar \zeta )} ]}}{\eta }d\eta } .
\end{equation}
~~~
\subsection{Reflectionless potential: double-pole soliton solutions}
From the closed algebraic system (\ref{5.y}) and (\ref{5.z}) combined with the reconstruction formula (\ref{5.a.e}). When the reflect coefficients $r(\zeta )$ and $\tilde r(\zeta )$ vanishes identically, the expression for the potential with double poles can be written as
\begin{equation}\label{5.a.i}
q(x,t) = {e^{ i(\Upsilon- 2{\omega _ +) }}}\cdot {q_ + }\left( {1 + \frac{{\left| {{N^\Delta }} \right|}}{{\left| N \right|}}} \right),
\end{equation}
where $|N|$ represents the $4{N_2} \times 4{N_2}$ determinants of matrix $N$, $|N^\Delta|$ is $(4{N_2} + 1) \times (4{N_2} + 1)$ determinant and
\begin{equation}\label{5.a.j}
N = {\left( {\begin{array}{*{20}{c}}
{n_{k,j}^{(1,1)}}&{n_{k,j}^{(1,2)}}\\
{n_{k,j}^{(2,1)}}&{n_{k,j}^{(1,2)}}
\end{array}} \right)_{(4{N_2} \times 4{N_2})}} ,
\end{equation}
\begin{equation}\label{5.a.k}
\left\{ \begin{array}{l}
n_{k,j}^{(1,1)} = {\left( {\frac{{{{\hat C}_j}}}{{{\eta _k} - {{\hat \eta }_j}}}} \right)_{(2{N_2}) \times (2{N_2})}}, \qquad \qquad \qquad \qquad k = 1,...2{N_2},j = 1,...,2{N_2}.\\
n_{k,j}^{(1,2)} = {\left( {\frac{{{{\hat C}_j}}}{{{\eta _k} - {{\hat \eta }_j}}}\left( {\frac{1}{{{\eta _k} - {{\hat \eta }_j}}} + {{\hat D}_j}} \right) + \frac{{i{q_ + }}}{{{\eta _k}}}{\delta _{j,k}}} \right)_{(2{N_2}) \times (2{N_2})}},  \quad k = 1,...2{N_2},j = 1,...,2{N_2}.\\
n_{k,j}^{(2,1)} = {\left( {\frac{{{{\hat C}_j}}}{{{{({\eta _k} - {{\hat \eta }_j})}^2}}} - \frac{{2q_0^2{q_ + }}}{{\eta _k^3}}{\delta _{j,k}}} \right)_{(2{N_2}) \times (2{N_2})}}, \qquad  k = 1,...2{N_2},j = 1,...,2{N_2}.\\
n_{k,j}^{(1,2)} = {\left( {\frac{{{{\hat C}_j}}}{{{{({\eta _k} - {{\hat \eta }_j})}^2}}}\left( {\frac{2}{{{\eta _k} - {{\hat \eta }_j}}} + {{\hat D}_j}} \right) + \frac{{i{q_ + }}}{{\eta _k^2}}{\delta _{j,k}}} \right)_{(2{N_2}) \times (2{N_2})}},  k = 1,...2{N_2},j = 1,...,2{N_2}.
\end{array} \right.
\end{equation}
\begin{equation}\label{5.a.l}
 {N^\Delta } = {\left( {\begin{array}{*{20}{c}}
{{\chi ^{\rm{T}}}}&0\\
N&\rho
\end{array}} \right)_{(4{N_2} + 1) \times (4{N_2} + 1)}},\qquad {\chi ^{\rm{T}}} = {\left( {\chi _1^{\rm T},\chi _2^{\rm T}} \right)_{(1) \times (4{N_2})}},\qquad \rho  = {\left( {\begin{array}{*{20}{c}}
{{\rho _1}}\\
{{\rho _2}}
\end{array}} \right)_{(4{N_2}) \times (1)}},
 \end{equation}
 \begin{equation}\label{5.a.m}
\chi _1^{\rm T} = {\left( {{{\hat A}_j} \cdot {e^{ - 2i\Theta (x,t,{{\hat \eta }_j})}}} \right)_{(1) \times (2{N_2})}},\chi _2^{\rm T} = {\left( {{{\hat A}_j} \cdot {e^{ - 2i\Theta (x,t,{{\hat \eta }_j})}}{{\hat D}_j}} \right)_{(1) \times (2{N_2})}},
\end{equation}
 \begin{equation}\label{5.a.n}
{\rho _1} = {\left( {\frac{1}{{{\eta _k}}}} \right)_{(2{N_2}) \times (1)}},{\rho _2} = {\left( { - \frac{1}{{\eta _k^2}}} \right)_{(2{N_2}) \times (1)}}.
\end{equation}

\begin{itemize}
\item[ $\bullet$ ] As $N_2=1$, when the discrete spectrum is a purely imaginary eigenvalue. Figure(\ref{s5.1}) shows the double-pole soliton solutions of the focusing KE equation with NZBCs. It can been seen that one-soliton solution with double poles has two hump. At this time, the soliton solution with NZBCs is non-stationary. when the boundary conditions vanishes identically, the corresponding soliton solution is static, which can be easy proved from the condition ($\theta$)(\ref{5.a.h})
\end{itemize}

\begin{figure}[htpb]
\centering
{
\begin{minipage}{6cm}
\centering
\includegraphics[width=6.4 cm]{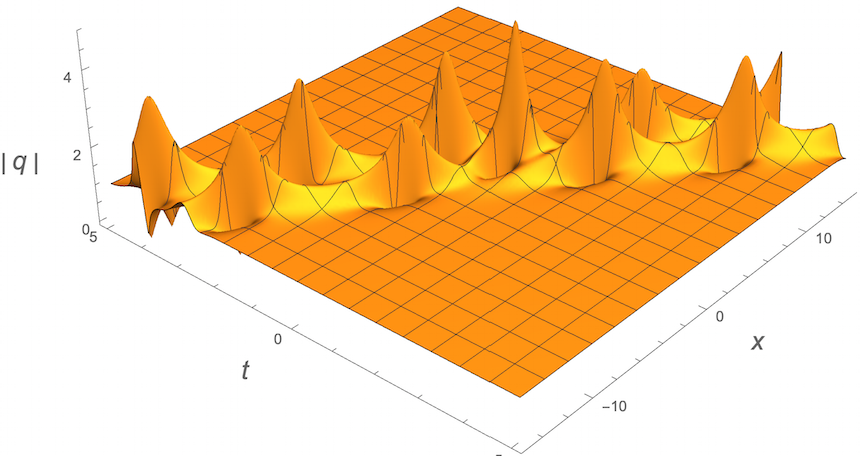}
\end{minipage}
} \qquad
{
\begin{minipage}{6cm}
\centering
\includegraphics[width=6.4cm]{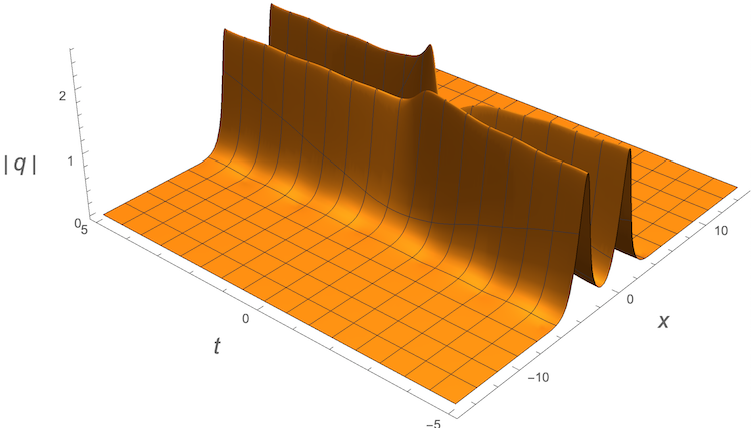}
\end{minipage}
}
\caption{Double-pole soliton solutions of the focusing KE equation. Left: interaction of the breather-breather solution with parameters $q_0=1$, $\alpha=\theta_+=0$, $\zeta_1=1$, $A_1=2e^i$, $B_1= - \frac{1}{3}+e^{-3i}$, $T=1$; Right: interaction of the bright-bright solitons through the limit ${q_ + }\to0$ and take the same remaining parameters as the left side.}
\label{s5.1}
\end{figure}

\begin{itemize}
\item[ $\bullet$ ] As $N_2=1$, the eigenvalue is located outside the $q_0$-circle in the upper half plane expect the imaginary axis. Figure(\ref{s5.2}) shows the dynamic behaviors of the double-pole soliton solutions of the focusing KE equation with NZBCs. Different from pure imaginary eigenvalues. Both the breather-breather solutions and bright-bright solitons are non-stationary.
\end{itemize}

\begin{figure}[htpb]
\centering
{
\begin{minipage}{6cm}
\centering
\includegraphics[width=6.4 cm]{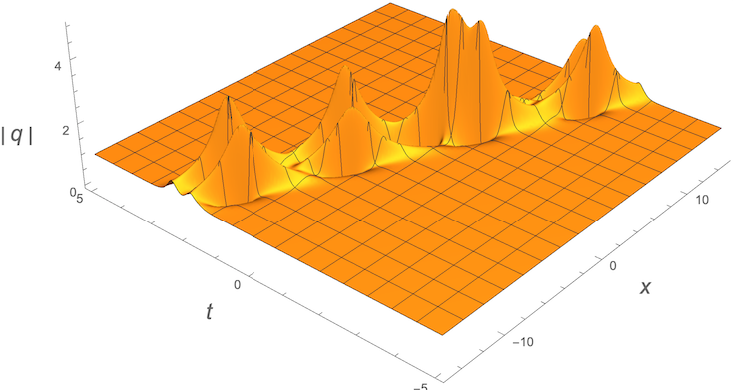}
\end{minipage}
} \qquad
{
\begin{minipage}{6cm}
\centering
\includegraphics[width=6.4cm]{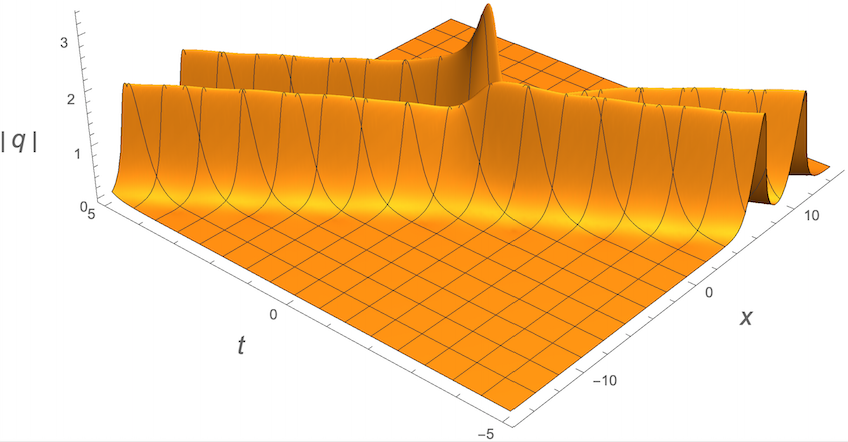}
\end{minipage}
}
\caption{Double-pole soliton solutions of the focusing KE equation. Left: interaction of the breather-breather solution with parameters $q_0=1$, $\alpha=\theta_+=0$, $\zeta_1=1+2i$, $A_1=2e^i$, $B_1= - \frac{1}{3}+e^{-3i}$, $T=1$; Right: interaction of the bright-bright solitons through the limit ${q_ + }\to0$ and take the same remaining parameters as the left side.}
\label{s5.2}
\end{figure}

\begin{figure}[htpb]
\centering
{
\begin{minipage}{0.3\linewidth}
\centering
\includegraphics[width=5 cm]{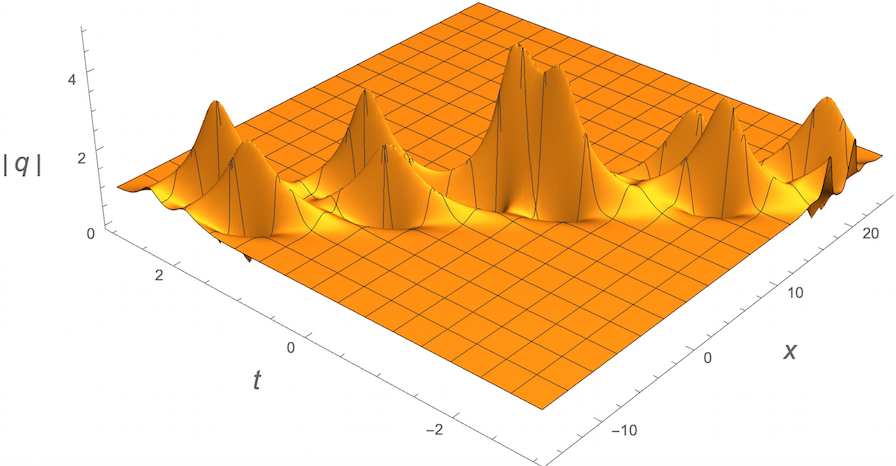}
\end{minipage}
}
{
\begin{minipage}{0.3\linewidth}
\centering
\includegraphics[width=5cm]{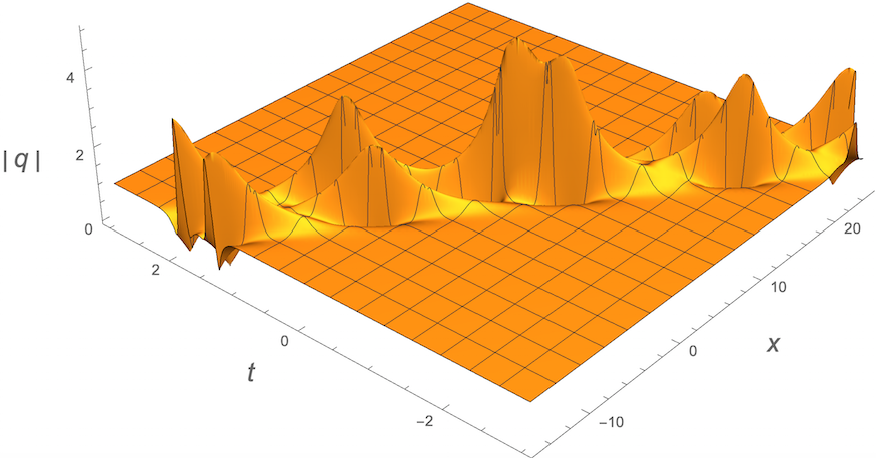}
\end{minipage}
}
{
\begin{minipage}{0.3\linewidth}
\centering
\includegraphics[width=5cm]{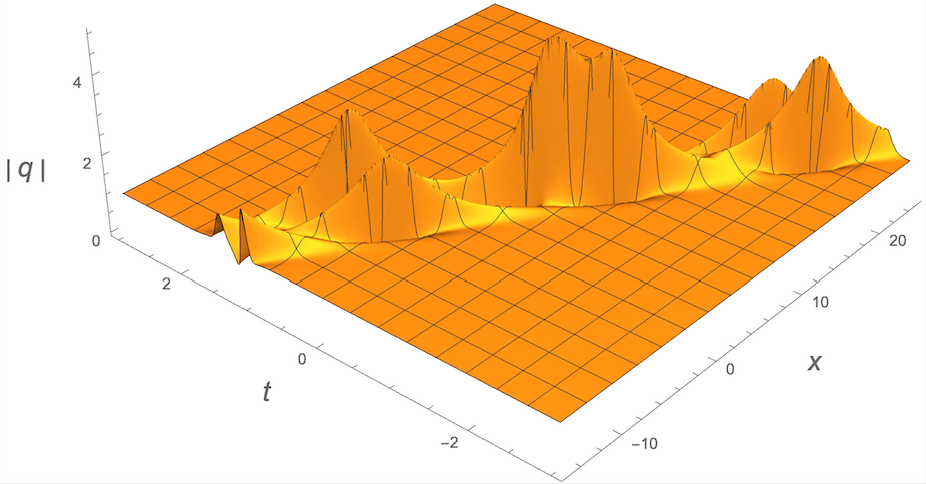}
\end{minipage}
}

\subfigure[]
{
\begin{minipage}{0.3\linewidth}
\centering
\includegraphics[width=3.5 cm]{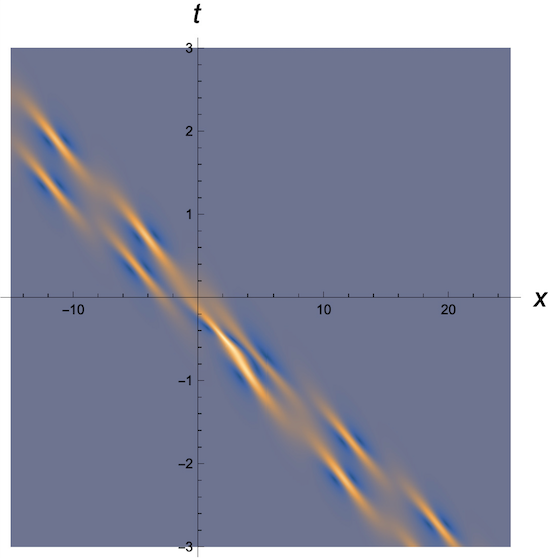}
\end{minipage}
}
\subfigure[]
{
\begin{minipage}{0.3\linewidth}
\centering
\includegraphics[width=3.5 cm]{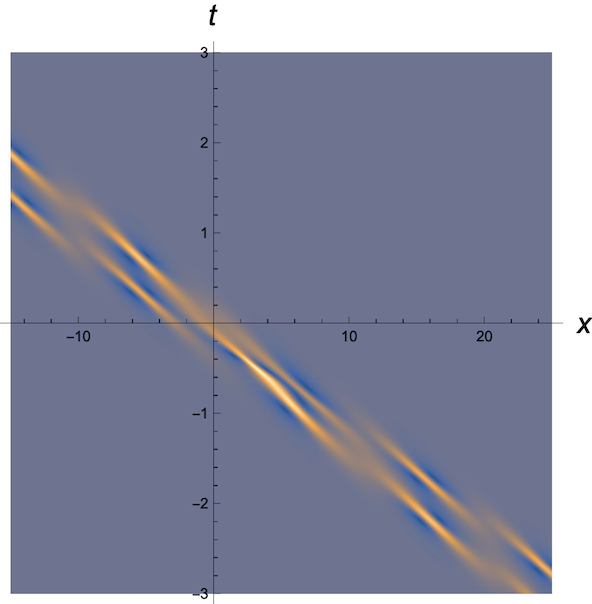}
\end{minipage}
}
\subfigure[]
{
\begin{minipage}{0.3\linewidth}
\centering
\includegraphics[width=3.5 cm]{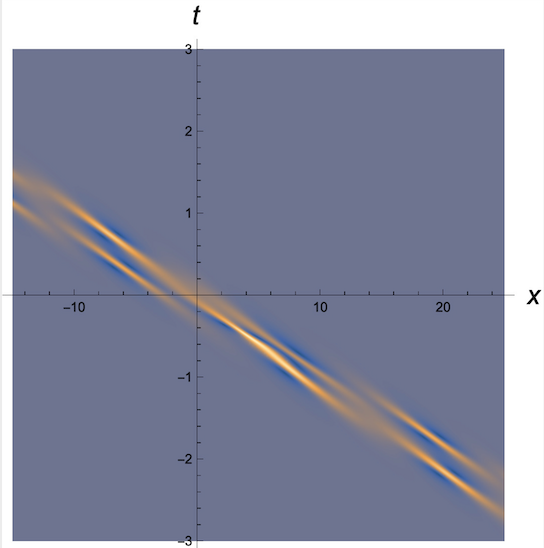}
\end{minipage}
}
\centering
\caption{Double-pole soliton solutions of the focusing KE equation with parameters $q_0=1$, $\alpha=\theta_+=0$, $\zeta_1=1+2i$, $A_1=2e^i$, $B_1= - \frac{1}{3}+e^{-3i}$, $T=1$(left), $T=\frac{3}{2}$(middle) and $T=2$(right). (a)$\sim$(c) display the density of the corresponding double-pole soliton solutions.}
\label{s5.3}
\end{figure}

\begin{itemize}
\item[ $\bullet$ ] As $N_2=1$, Figure(\ref{s5.3}) shows the influence of different $T$ on the spatial structure of the KE equation with NZBCs and double poles. It can be seen that the difference of $T$ will cause the phase and the speed of space movement of the double-pole soliton solutions to be significantly different. With the amplification of $T$, the spatial velocity of the soliton solutions will increases and the phase difference between $q_+$ and $q_-$ will decreases.
\end{itemize}

\section{Relationship between simple-pole and double-pole by RHP}
From the numerical simulation of the simple-pole and the double-pole soliton solutions (see(\ref{fig.6}),(\ref{s5.3})), it is found that the two humps of the double-pole are close to each other, which reveals whether the double-pole can be obtained by the limit form of the two simple poles.
\begin{theorem}\label{6.00}
If we choose ${\tilde c_{{j_1}}} + {\tilde c_{{j_2}}} = {\hat A_{{N_2} + {j_2}}}{\hat B_{{N_2}  + {j_2}}}, {\tilde c_{{N_1} + {j_1}}} + {\tilde c_{{N_1} + {j_2}}} = {\hat A_{{j_2}}}{\hat  B_{{j_2}}}, \delta {\tilde c_{{j_2}}} = {\hat A_{{N_2} + {j_2}}}, \delta {\tilde c_{{N_1} + {j_2}}} = {\hat A_{{j_2}}},$ then the double-pole soliton solutions can be obtained by the limit form of the two simple-pole  soliton solutions.
\end{theorem}
We prove this theorem in two steps, recall the potential formulas of simple-pole (\ref{4.n}) and double-pole (\ref{5.a.e}) in the case of reflectionless,
\begin{equation}\label{6.1}
q(x,t) = {e^{ i(\Upsilon- {\omega _ +) }}} \{ {q_ + }{e^{ - i{\omega _ + }}} + i\sum\limits_{j = 1}^{2{N_1}} {{{\tilde c}_j}{e^{ - 2i\Theta ({{\bar \eta }_j})}}\mu _{ + 1,2}^{(2)}({{\bar \eta }_j})  } \} .
\end{equation}
 \begin{equation}\label{6.2}
q(x,t) = {e^{ i(\Upsilon -{\omega _ + )}}} \{ {{q_ + }{e^{ - i{\omega _ + }}} + i\sum\limits_{j = 1}^{2{N_2}} {{{\hat C}_j}\left( {\dot \mu _ {+1,2} ^{(2)}({{\hat \eta }_j}) + \mu _ {+1,2} ^{(2)}({{\hat \eta }_j}){{\hat D}_j}} \right)} } \}.
\end{equation}
\begin{lemma}\label{6.3}
Randomly select $\zeta_{j_1}$ and ${\zeta_{j_1}}+\delta$ in $a (\zeta_{j_1})$ with simple zeros, if ${\tilde c_1} + {\tilde c_2} \mapsto {\hat A_{{N_2} + {j_2}}}{\hat B_{{N_2} +  {j_2}}},\delta {\tilde c_{{j_2}}} \mapsto {\hat A_{{N_2} + {j_2}}},{\tilde c_{{N_1} + {j_1}}} + {\tilde c_{{N_1} + {j_2}}} \mapsto {\hat A_{{j_2}}} {\hat B_{{j_2}}},\delta {\tilde c_{{N_1} + {j_2}}} \mapsto {\hat A_{{j_2}}}$, while taking the limit $\delta  \to 0$, the equation (\ref{6.1}) and equation (\ref{6.2}) are identical in form.
\end{lemma}
\begin{proof}
In fact, we just need to prove that the last item is consistent.
For the equation (\ref{6.1}), we have
 \begin{equation}\label{6.4}
 \begin{split}
{\tilde c_{{j_1}}}{e^{ - 2i\Theta ({{\bar \zeta }_{{j_1}}})}}\mu _{ + 1,2}^{(2)}({\bar \zeta _{{j_1}}}) + {\tilde c_{{j_2}}}{e^{ - 2i\Theta ({{\bar \zeta }_{{j_2}}})}}\mu _{ + 1,2}^{(2)}({\bar \zeta _{{j_2}}}) + {\tilde c_{{N_1} + {j_1}}}&{e^{ - 2i\Theta ( - q_0^2\zeta _{{j_{_1}}}^{ - 1})}}\mu _{ + 1,2}^{(2)}( - q_0^2\zeta _{{j_{_1}}}^{ - 1}) \\ + &{\tilde c_{{N_1} + {j_2}}}{e^{ - 2i\Theta ( - q_0^2\zeta _{{j_2}}^{ - 1})}}\mu _{ + 1,2}^{(2)}( - q_0^2\zeta _{{j_2}}^{ - 1}) \\  = ({\tilde c_{{j_1}}} + {\tilde c_{{j_2}}}){e^{ - 2i\Theta ({{\bar \zeta }_{{j_1}}})}}\mu _{ + 1,2}^{(2)}({\bar \zeta _{{j_1}}}) + {\tilde c_{{j_2}}}{e^{ - 2i\Theta ({{\bar \zeta }_{{j_2}}})}}\mu _{ + 1,2}^{(2)}({\bar \zeta _{{j_2}}}) &- {\tilde c_{{j_2}}}{e^{ - 2i\Theta ({{\bar \zeta }_{{j_1}}})}}\mu _{ + 1,2}^{(2)}({\bar \zeta _{{j_1}}}) \\+ ({\tilde c_{{N_1} + {j_1}}} + {\tilde c_{{N_1} + {j_2}}}){e^{ - 2i\Theta ( - q_0^2\zeta _{{j_{_1}}}^{ - 1})}}\mu _{ + 1,2}^{(2)}( - q_0^2\zeta _{{j_{_1}}}^{ - 1}) +& {\tilde c_{{N_1} + {j_2}}}{e^{ - 2i\Theta ( - q_0^2\zeta _{{j_2}}^{ - 1})}}\mu _{ + 1,2}^{(2)}( - q_0^2\zeta _{{j_2}}^{ - 1})\\-& {\tilde c_{{N_1} + {j_2}}}{e^{ - 2i\Theta ( - q_0^2\zeta _{{j_1}}^{ - 1})}}\mu _{ + 1,2}^{(2)}( - q_0^2\zeta _{{j_1}}^{ - 1}),
\end{split}
\end{equation}
letting $\zeta_{j_2}=\zeta_{j_1}+\delta$, and $\delta \to 0$ ,then
 \begin{equation}\label{6.5}
 \begin{split}
 = ({\tilde c_{{j_1}}} + {\tilde c_{{j_2}}}){e^{ - 2i\Theta ({{\bar \zeta }_ {{j_1}}})}}\mu _{ + 1,2}^{(2)}({\bar \zeta _{{j_1}}}) + \delta {\tilde c_ {{j_2}}}\left[ { - 2i\dot \Theta ({{\bar \zeta }_{{j_1}}}){e^{ - 2i\Theta ({{\bar  \zeta }_{{j_1}}})}}\mu _{ + 1,2}^{(2)}({{\bar \zeta }_{{j_1}}}) + {e^{ - 2i\Theta  ({{\bar \zeta }_{{j_1}}})}}\dot \mu _{ + 1,2}^{(2)}({{\bar \zeta }_{{j_1}}})}  \right] \\+ ({\tilde c_{{N_1} + {j_1}}} + {\tilde c_{{N_1} + {j_2}}}){e^{ - 2i\Theta  ( - q_0^2\zeta _{{j_1}}^{ - 1})}}\mu _{ + 1,2}^{(2)}( - q_0^2\zeta _{{j_1}}^{ - 1})  \\+ \delta {\tilde c_{{N_1} + {j_2}}}\left[ { - 2i\dot \Theta ( - q_0^2\zeta _ {{j_1}}^{ - 1}){e^{ - 2i\Theta ( - q_0^2\zeta _{{j_1}}^{ - 1})}}\mu _{ + 1,2}^ {(2)}( - q_0^2\zeta _{{j_1}}^{ - 1}) + {e^{ - 2i\Theta ( - q_0^2\zeta _{{j_1}}^{ -  1})}}\dot \mu _{ + 1,2}^{(2)}( - q_0^2\zeta _{{j_1}}^{ - 1})} \right].
\end{split}
\end{equation}
If denoted
$${\tilde c_{{j_1}}} + {\tilde c_{{j_2}}} = {\hat A_{{N_2} + {j_2}}}{\hat B_{{N_2}  + {j_2}}}, \quad {\tilde c_{{N_1} + {j_1}}} + {\tilde c_{{N_1} + {j_2}}} = {\hat A_{{j_2}}}{\hat  B_{{j_2}}},$$
$$\delta {\tilde c_{{j_2}}} = {\hat A_{{N_2} + {j_2}}},\quad \delta {\tilde c_{{N_1} + {j_2}}} = {\hat A_{{j_2}}},$$
obtained
 \begin{equation}\label{6.5}
= {\hat c_{{j_2}}}(\dot \mu _{ + 1,2}^{(2)}( - q_0^2\zeta _{{j_2}}^{ - 1}) + \mu  _{ + 1,2}^{(2)}( - q_0^2\zeta _{{j_2}}^{ - 1}){\hat D_{{j_2}}}) + {\hat c_{{N_2} +  {j_2}}}(\dot \mu _{ + 1,2}^{(2)}({\bar \zeta _{{j_2}}}) + \mu _{ + 1,2}^{(2)}({\bar  \zeta _{{j_2}}}){\hat D_{{N_2} + {j_2}}})
\end{equation}\par
It is shown that the potential formula of double-pole (\ref{6.2}) can be obtained from the equation (\ref{6.1}) based on the operation. Since the $\mu _{ + 1,2}^{(2)}({{\bar \eta }_j})$ and $\mu _ {+ 1,2} ^{(2)}({{\hat \eta }_j})$ are unknown, it is not enough to explain that the double-pole can be obtained by the limit form of the two simple poles.
\end{proof}
\begin{lemma}\label{6.6}
If ${\tilde c_1} + {\tilde c_2} \mapsto {\hat A_{{N_2} + {j_2}}}{\hat B_{{N_2} +  {j_2}}},\delta {\tilde c_{{j_2}}} \mapsto {\hat A_{{N_2} + {j_2}}},{\tilde c_{{N_1} + {j_1}}} + {\tilde c_{{N_1} + {j_2}}} \mapsto {\hat A_{{j_2}}} {\hat B_{{j_2}}},\delta {\tilde c_{{N_1} + {j_2}}} \mapsto {\hat A_{{j_2}}}$, then the algebraic system (\ref{4.h})-(\ref{4.i}) and (\ref{5.y})-(\ref{5.z}) are equivalent.
\end{lemma}
\begin{proof}
To illustrate this situation briefly, we only choose $\mu  _{ + 1}^{(1)}( \zeta_{j_1})$,  and other columns can be proved similarly. For equation (\ref{4.h}), we have
 \begin{equation}\label{6.7}
 \begin{split}
\mu _ + ^{(1)}({\zeta _{{j_1}}}) = \left( {\begin{array}{*{20}{c}}  {{e^{ - i{\omega _ + }}}i\zeta _j^{ - 1}{q_ + }} \\   {{e^{ - i{\omega _ + }}}} \end{array}} \right) +  \cdots  + \frac{{{{\tilde c}_{{j_1}}}{e^{ - 2i\Theta  ({{\bar \zeta }_{{j_1}}})}}}}{{{\zeta _{{j_1}}} - {{\bar \zeta }_{{j_1}}}}}\mu _ +  ^{(2)}({\bar \zeta _{{j_1}}}) + \frac{{{{\tilde c}_{{j_2}}}{e^{ - 2i\Theta ({{\bar  \zeta }_{{j_2}}})}}}}{{{\zeta _{{j_1}}} - {{\bar \zeta }_{{j_2}}}}}\mu _ + ^{(2)} ({\bar \zeta _{{j_2}}}) +  \cdots  \\+ \frac{{{{\tilde c}_{{N_1} + {j_1}}}{e^{ - 2i \Theta ({{\hat \zeta }_{{j_1}}})}}}}{{{\zeta _{{j_1}}} - {{\hat \zeta }_ {{j_1}}}}}\mu _ + ^{(2)}({\hat \zeta _{{j_1}}}) + \frac{{{{\tilde c}_{{N_1} +  {j_2}}}{e^{ - 2i\Theta ({{\hat \zeta }_{{j_2}}})}}}}{{{\zeta _{{j_1}}} - {{\hat  \zeta }_{{j_2}}}}}\mu _ + ^{(2)}({\hat \zeta _{{j_2}}}) +  \cdots ,
 \end{split}
\end{equation}
letting $\zeta_{j_2}=\zeta_{j_1}+\delta$, and $\delta \to 0$ ,then
\begin{equation}\label{6.8}
\begin{split}
 \mu _ + ^{(1)}({\zeta _{{j_1}}}) =  \left( {\begin{array}{*{20}{c}}  {{e^{ - i{\omega _ + }}}i\zeta _j^{ - 1}{q_ + }} \\   {{e^{ - i{\omega _ + }}}} \end{array}} \right)  +  \cdots  + \frac{{({{\tilde c}_{{j_1}}} + {{\tilde c}_ {{j_2}}}){e^{ - 2i\Theta ({{\bar \zeta }_{{j_1}}})}}}}{{{\zeta _{{j_1}}} - {{\bar  \zeta }_{{j_1}}}}}\mu _ + ^{(2)}({\bar \zeta _{{j_1}}})  \\ + \delta {\tilde c_ {{j_2}}}\left[ {\frac{{{e^{ - 2i\Theta ({{\bar \zeta }_{{j_1}}})}}\left( { - 2i\dot  \Theta ({{\bar \zeta }_{{j_1}}})\mu _ + ^{(2)}({{\bar \zeta }_{{j_1}}}) + \dot \mu  _ + ^{(2)}({{\bar \zeta }_{{j_1}}})} \right)}}{{{\zeta _{{j_1}}} - {{\bar \zeta }_ {{j_2}}}}} + \frac{{{e^{ - 2i\Theta ({{\bar \zeta }_{{j_1}}})}}\mu _ + ^{(2)} ({{\bar \zeta }_{{j_1}}})}}{{{{({\zeta _{{j_1}}} - {{\bar \zeta }_{{j_2}}})}^2}}}}  \right] +  \cdots  \\  + \frac{{({{\tilde c}_{{N_1} + {j_1}}} + {{\tilde c}_{{N_1} +  {j_2}}}){e^{ - 2i\Theta ({{\hat \zeta }_{{j_1}}})}}}}{{{\zeta _{{j_1}}} - {{\hat  \zeta }_{{j_1}}}}}\mu _ + ^{(2)}({\hat \zeta _{{j_1}}}) \\ + \delta {\tilde c_{{N_1} +  {j_2}}}\left[ {\frac{{{e^{ - 2i\Theta ({{\hat \zeta }_{{j_1}}})}}\left( { - 2i\dot  \Theta ({{\hat \zeta }_{{j_1}}})\mu _ + ^{(2)}({{\hat \zeta }_{{j_1}}}) + \dot \mu  _ + ^{(2)}({{\hat \zeta }_{{j_1}}})} \right)}}{{{\zeta _{{j_1}}} - {{\hat \zeta }_ {{j_1}}}}} + \frac{{{e^{ - 2i\Theta ({{\hat \zeta }_{{j_1}}})}}\mu _ + ^{(2)} ({{\hat \zeta }_{{j_1}}})}}{{{{({\zeta _{{j_1}}} - {{\hat \zeta }_{{j_1}}})}^2}}}}  \right] +  \cdots .
\end{split}
\end{equation}
If denoted
\begin{equation}\label{6.9}
\begin{split}
{\tilde c_{{j_1}}} + {\tilde c_{{j_2}}} = {\hat A_{{N_2} + {j_2}}}{\hat B_{{N_2}  + {j_2}}}, \quad {\tilde c_{{N_1} + {j_1}}} + {\tilde c_{{N_1} + {j_2}}} = {\hat A_{{j_2}}}{\hat  B_{{j_2}}},\\
\delta {\tilde c_{{j_2}}} = {\hat A_{{N_2} + {j_2}}},\qquad \qquad \qquad \delta {\tilde c_{{N_1} + {j_2}}} = {\hat A_{{j_2}}},
\end{split}
\end{equation}
obtained
\begin{equation}\label{6.10}
\begin{split}
\mu _ + ^{(1)}({\zeta _{{j_1}}}) = \left( {\begin{array}{*{20}{c}}  {{e^{ - i{\omega _ + }}}i\zeta _j^{ - 1}{q_ + }} \\   {{e^{ - i{\omega _ + }}}} \end{array}} \right) +  \cdots  + \frac{{{{\hat c}_{{j_1}}}\dot \mu _ + ^{(2)} ({{\hat \zeta }_{{j_1}}})}}{{{\zeta _{{j_1}}} - {{\hat \zeta }_{{j_1}}}}} + \frac {{{{\hat c}_{{j_1}}}\mu _ + ^{(2)}({{\hat \zeta }_{{j_1}}})}}{{{{({\zeta _{{j_1}}}  - {{\hat \zeta }_{{j_1}}})}^2}}} + \frac{{{{\hat c}_{{j_1}}}\mu _ + ^{(2)}({{\hat  \zeta }_{{j_1}}})}}{{{\zeta _{{j_1}}} - {{\hat \zeta }_{{j_1}}}}}{\bar D_{{j_1}}} +   \cdots  \\+ \frac{{{{\hat c}_{{N_2} + {j_1}}}\mu _ + ^{(2)}({{\bar \zeta }_ {{j_1}}})}}{{{\zeta _{{j_1}}} - {{\bar \zeta }_{{j_1}}}}} + \frac{{{{\hat c}_{{N_2}  + {j_1}}}\mu _ + ^{(2)}({{\bar \zeta }_{{j_1}}})}}{{{{({\zeta _{{j_1}}} - {{\bar  \zeta }_{{j_1}}})}^2}}} + \frac{{{{\hat c}_{{N_2} + {j_1}}}\mu _ + ^{(2)}({{\bar  \zeta }_{{j_1}}})}}{{{\zeta _{{j_1}}} - {{\bar \zeta }_{{j_1}}}}}{\bar D_{{N_2} +  {j_1}}} +  \cdots ,
\end{split}
\end{equation}
which is equivalent to equation (\ref{5.u}), since the simple-pole and the double-pole have the same symmetry, it is pointed out that in the case of (\ref{6.9}), the algebraic system of (\ref{4.h})-(\ref{4.i}) and (\ref{5.y})-(\ref{5.z}) are equivalent, then $\mu _{ + 1,2}^{(2)}({{\bar \eta }_j})$ and $\mu _ {+ 1,2} ^{(2)}({{\hat \eta }_j})$ are consistent.
\end{proof}
Theorem (\ref{6.00}) can be proved by Lemma (\ref{6.3}) and Lemma(\ref{6.6}).

\section{Conclusion and remark}
In this paper, the original Lax pairs of the defocusing and focusing Kundu-Eckhaus equation are transformed into a new compatible Lax pairs by a gauge transformed, and the associated Riemann-Hilbert problems for both defocusing and focusing Kundu-Eckhaus equation with non-zero boundary are established. For the defocusing KE equation case, we prove that the scattering data only can have simple zeros, hence we just construct the simple poles soliton solutions. However, for the focusing case, the corresponding simple-pole and double-pole solutions are obtained by solving the corresponding Riemann-Hilbert problems and dynamic behaviors of the modulo of the solution are studied by assigning values to the parameters. The influence factors of the asymptotic phase difference are explained by condition ($\theta$). Different from the usual zero-boundary conditions (ZBCs), the scattering problem with NZBCs is much more complicated, due to a two-sheet Riemann surface need to be introduced to over come the multi-value of the square root, more symmetry and discrete spectral distributions as well as regularized the RHP conditions. Finally, we clarify that the double-pole soliton solutions can be obtained by the limit form of the two simple poles soliton solutions with NZBC by Riemann-Hilbert problem.\par
This paper still has an open problem that $q$ contains the phase factors $ e^{i \omega _+}$, where ${\omega _ + }$ is some definite integral. It is worth noting that calculating the phase ${\omega _ + }$ is not an easy task, but the modulo of the nonlinear equation could explict obtained, and the phase ${\omega _ + }$ can be calculated for some simple soliton solutions by direct integration. The exact solution of the ${\omega _ + }$ phase and the relationship between simple-pole and multi-pole      soliton solutions will be our main work in the future.

\begin{remark}
When we finished this work, we were told that there are two papers on arXiv about the focusing KE equation with NZBCs. One of these is Lili Wen and Engui Fan's work \cite{li-li}, but we should point out that the difference between our work and the work in \cite{li-li}. First, if we set $\alpha=0$ in our initial value data in our work, it will reduce to the work in \cite{li-li}. Second, we not only obtain the simple poles soliton solutions, but also obtain the double-pole soliton solutions.  Another is Jin-Jie Yang and Shou-fu Tian's work \cite{yang}. Although the simple-pole and double-pole soliton solutions were obtained, it seems need to be checked the gauge transformation used in that paper. As we consider the problem for the NZBCs, not zero boundary conditions now.
\par
However, there are two main differences between our work and the works in \cite{li-li} and \cite{yang}. First, we explicitly calculate the phase function $\omega_+$ of the one-soliton solution of the defocusing and focusing KE equation, see (\ref{2.1d2}) and (\ref{4.x}), respectively. This phase function comes from the requirement of the normalize condition of the Riemann-Hilbert problem for the function $M(x,t,\zeta)$. Second, the relationship between the simple-pole solutions and double-pole solutions were not considered in their papers. Here, we show that the double poles solutions can be viewed as some proper limit of the two simple poles solutions under the NZBCs at both side, which is the main innovation of our this paper.
\end{remark}

{\bf Acknowledgements}
This work is supported by Natural Science Foundation of China under NO.11971313 and Shanghai natural science foundation under project NO.19ZR1434500.

\end{document}